\documentclass{article}

\newcommand\longVOnly[1]{#1} \newcommand\shortVOnly[1]{}   \newcommand\futureVOnly[1]
                     {}

\newcommand\colorblind[1]{}
\newcommand\nonColorblind[1]{#1}

\longVOnly{}
\shortVOnly{}

\longVOnly{\usepackage[
  a4paper,
  margin=15mm,
]{geometry}}
\usepackage[utf8]{inputenc} 
\usepackage{xspace}
\longVOnly{
  \usepackage{amsthm}
  \usepackage{authblk}
}
\usepackage{hyperref}
\usepackage{mathtools}
\longVOnly{\usepackage[english]{babel}}
\usepackage{enumitem}
\usepackage{csquotes}
\usepackage{bussproofs}
\usepackage{bbold}

\usepackage[T1]{fontenc} 
\usepackage{stmaryrd} 
\usepackage{cmll}
\usepackage{xcolor}

\longVOnly{
  \usepackage{amssymb} 
  \usepackage{amsmath}
  }
\usepackage{stmaryrd}

\shortVOnly{\usepackage[datamodel=acmdatamodel,style=acmnumeric,]{biblatex}}
\longVOnly{\usepackage{biblatex}}
\usepackage[disable]{todonotes}
\usepackage{cmll}

\usepackage{tikz-cd}
\usetikzlibrary{arrows}
\usetikzlibrary{matrix}
\usetikzlibrary{arrows,positioning,calc}

\usepackage{xpatch}

\makeatletter
\patchcmd\blx@bblinput{\blx@blxinit}
                      {\blx@blxinit
                      }{}{\fail}
\makeatother

\longVOnly{
\newtheorem{theorem}{Theorem}[section]
\newtheorem{lemma}[theorem]{Lemma}

\newtheorem{definition}[theorem]{Definition}

\newtheorem{proposition}[theorem]{Proposition}
\newtheorem{conjecture}[theorem]{Conjecture}
}
\newtheorem{remark}{Remark}

\newcommand{\inv}[1]{#1^{\text{-}1}}
\newcommand{\pb}[2]{#1_{\uparrow #2}}

\newcommand{\tpb}[4]{#1_{\uparrow #2\uparrow #3\uparrow #4}}
\newcommand{\qpb}[5]{#1_{\uparrow #2\uparrow #3\uparrow #4\uparrow #5}}

\newcommand{\dom}{\mathtt{dom}}

\newcommand*{\LL}{\mathsf{LL}}

\newcommand{\Set}{\mathtt{Set}}

\newcommand*{\Rel}{\mathtt{Rel}}

\newcommand*{\Ob}[1]{\mathtt{ob}\!\left(#1\right)}

\nonColorblind{
\colorlet{indexcolor}{violet}
\colorlet{basecolor}{blue}
\colorlet{cutcolor}{red}
\colorlet{defcolor}{red}
\colorlet{subtypecolor}{green!50!black}
\colorlet{subtypecolorL}{green!70!black}
\colorlet{proofirrcolor}{gray}
\colorlet{typescolor}{gray}
}
\colorblind{
\definecolor{WongViolet}{RGB}{204,121,167} 
\definecolor{WongRed}{RGB}{213,94,0} 
\definecolor{WongBlue}{RGB}{0,114,176} 
\definecolor{WongYellow}{RGB}{240,128,66} 
\definecolor{WongGreen}{RGB}{0,158,115}
\definecolor{WongLBlue}{RGB}{86,180,233}
\definecolor{WongOrange}{RGB}{230,159,0}
\colorlet{indexcolor}{WongViolet}
\colorlet{basecolor}{WongBlue}
\colorlet{cutcolor}{WongRed}
\colorlet{defcolor}{WongRed}
\colorlet{subtypecolor}{WongGreen}\colorlet{subtypecolorL}{WongGreen}\colorlet{proofirrcolor}{WongYellow}\colorlet{typescolor}{WongYellow}}

\newcommand*{\eqvdash}{\mathop{\dashv\vdash}}

\newcommand*{\IndLL}{\mathsf{IndLL}}
\newcommand*{\OIndLL}{\mathsf{OIndLL}}

\newcommand*{\indvdashA}[1]{\raisebox{0.6em}{$#1$}}
\newcommand*{\indvdash}[1]{\mathop{\vdash_{\!\!\!\mathpalette\indvdashA{\color{indexcolor}#1}}}}

\newcommand\indfun[2]{{\color{basecolor}#1\!\left({\color{black}#2}\right)}}
\DeclareMathOperator*\subtype{{\color{subtypecolor}\sqsubseteq}}
\DeclareMathOperator*\suptype{{\color{subtypecolor}\sqsupseteq}}
\DeclareMathOperator*\eqtype{{\color{subtypecolor}\equiv}}
\DeclareMathOperator*\reduce{{\color{cutcolor}\rightsquigarrow}}
\DeclareMathOperator*\preduce{{\color{cutcolor}\Rightarrow}}
\newcommand{\oplusm}[2]{\prescript{}{\color{indexcolor}#1\!}{\oplus}\raisebox{-.1em}{$_{\!\color{indexcolor}#2}$}}
\newcommand{\withm}[2]{\prescript{}{\color{indexcolor}#1\!}{\with}\raisebox{-.1em}{$_{\!\color{indexcolor}#2}$}}

\newcommand\mum[1]{\mu_{{\!\color{indexcolor}#1}}}
\newcommand\num[1]{\nu_{{\!\color{indexcolor}#1}}}
\newcommand\bangm[1]{!_{{\color{indexcolor}#1}}}
\newcommand\wnm[1]{?_{{\!\color{indexcolor}#1}}}
\newcommand\varm[2]{{\color{indexcolor}#1(}#2{\color{indexcolor})}}

\newcommand{\optional}[1]{{\color{proofirrcolor}#1}}
\newcommand{\svdots}{\raisebox{0pt}[1em][0pt]{\vdots}}

\def\dashedColoredScoreFiller#1{\hbox to2.4mm{{\color{#1}\hss\vrule width1.8mm height0.4pt depth0.0pt\hss}}}

\def\coloredScoreFiller#1{\hbox to2.8mm{{\color{#1}\vrule width3mm height0.4pt depth0.0pt}}}
\def\coloredLine#1{\gdef\theScoreFiller{\coloredScoreFiller{#1}}\ignorespaces
}

\def\dashedColoredLine#1{\gdef\theScoreFiller{\dashedColoredScoreFiller{#1}}\ignorespaces
}

\newcommand\subtypLine{\dashedColoredLine{subtypecolorL}}
\newcommand\functorialLine{\dashedColoredLine{basecolor}}
\newcommand\cutLine{\coloredLine{cutcolor}}
\newcommand\dasheddefLine{\dashedColoredLine{defcolor}}
\newcommand\productiveLine{\LeftLabel{\raisebox{-.7em}{p}\!\!}}

\newcommand*{\1}{\mathbb 1}
\newcommand*{\0}{\mathbb 0}

\newcommand*{\J}{\mathbb J}

\newcommand*{\N}{\mathbb N}

\newcommand*{\id}{\mathsf{id}}
\newcommand*{\init}{\mathsf{init}}
\newcommand*{\term}{\mathsf{term}}

\newcommand\raiseRel[1]{\mathrel{\raisebox{.7em}{$#1$}}}

\newcommand*{\sem}[2]{ \left\llbracket #2 \right\rrbracket_{#1}}

\makeatletter
\newcommand*{\doublerightarrow}[2]{\mathrel{
  \settowidth{\@tempdima}{$\scriptstyle#1$}
  \settowidth{\@tempdimb}{$\scriptstyle#2$}
  \ifdim\@tempdimb>\@tempdima \@tempdima=\@tempdimb\fi
  \mathop{\vcenter{
    \offinterlineskip\ialign{\hbox to\dimexpr\@tempdima+1em{##}\cr
    \rightarrowfill\cr\noalign{\kern.5ex}
    \rightarrowfill\cr}}}\limits^{\!#1}_{\!#2}}}
\newcommand*{\triplerightarrow}[1]{\mathrel{
  \settowidth{\@tempdima}{$\scriptstyle#1$}
  \mathop{\vcenter{
    \offinterlineskip\ialign{\hbox to\dimexpr\@tempdima+1em{##}\cr
    \rightarrowfill\cr\noalign{\kern.5ex}
    \rightarrowfill\cr\noalign{\kern.5ex}
    \rightarrowfill\cr}}}\limits^{\!#1}}}
\newcommand{\colim@}[2]{\vtop{\m@th\ialign{##\cr
    \hfil$#1\operator@font lim$\hfil\cr
    \noalign{\nointerlineskip\kern1.5\ex@}#2\cr
    \noalign{\nointerlineskip\kern-\ex@}\cr}}}
\newcommand{\colim}{\mathop{\mathpalette\colim@{\rightarrowfill@\textstyle}}\nmlimits@
}
\makeatother

   \newcommand\skiplength{5pt}

\newcommand\fnote[1]{\todo{#1}}

\shortVOnly{
  \renewcommand*{\indvdashA}[1]{\raisebox{0.4em}{$#1$}}
  \renewcommand*{\indvdash}[1]{\mathop{\vdash_{_{\!\!\!\mathpalette\indvdashA{\color{indexcolor}#1}}}}}
  \renewcommand{\oplusm}[2]{\raisebox{-.1em}{$_{\color{indexcolor}#1\!}$}{\oplus}\raisebox{-.1em}{$_{\hspace{-.15em}\color{indexcolor}#2}$}}
  \renewcommand{\withm}[2]{\prescript{}{\color{indexcolor}#1\!}{\with}\raisebox{-.1em}{$_{\!\color{indexcolor}#2}$}}

}

\title{An Indexed Linear Logic for Idempotent Intersection Types}

\longVOnly{
  \author[1]{Flavien Breuvart}
  \affil[1]{Univ. USPN, Sorbonne Paris Cit\'e, LIPN, UMR 7030, CNRS, F-93430 Villetaneuse, France}
  \author[2]{Federico {Olimpieri}}
  \affil[2]{School of Mathematics, University of Leeds}
}

\begin{document}

\longVOnly{\maketitle}

\begin{abstract} Indexed Linear Logic has been introduced by Bucciarelli and Ehrhard and can be seen as a logical presentation of both non-idempotent intersection types and the relational semantics of linear logic.

We introduce an idempotent variant of Indexed Linear Logic. We give a fine-grained reformulation of the syntax by exposing implicit parameters and by unifying several operations on formulae via the notion of base change. Idempotency is achieved by means of an appropriate subtyping relation. We carry on an in-depth study of IndLL as a logic, showing how it determines a refinement of classical linear logic and establishing a terminating cut-elimination procedure. Cut-elimination is proved to be confluent up to an appropriate congruence induced by the subtyping relation.

\end{abstract}

\shortVOnly{\maketitle}

\shortVOnly{
  \paragraph{Long version and compagnon paper} 
  An extended version~\cite{BreuvOlimp24} is available on the Arxiv \href{https://arxiv.org/abs/2401.14126}{https://arxiv.org/abs/2401.14126}. 
}
\section{Introduction}

The original \emph{Indexed Linear Logic} ($\OIndLL$) has been introduced by Bucciarelli and Ehrhard~\cite{IndLL1, IndLL2} as a syntactic counterpart to some aspects of linear logic relational semantics. The semantic analysis brought to the discovery that one could define an extension of linear logic, where formulae and proofs depends on a \emph{choice of resources}. A significant outcome of this extension is the capability to disregard this choice and recover standard linear logic formulae and proofs. Intuitively, an indexed formula can be understood as an indexed family of elements that live in the relational interpretation of its underlying linear logic formula.

The present paper, together with its companion~\cite{BEMOlongv}, represent a first step towards a\emph{ modularisation} of Indexed Linear Logic, twenty years after its first introduction. Since the initial publications on the topic, a number of logical systems have undergone modularisation, that is, they have been defined with respect to some external algebraic structure. This structure provides means for analysis and/or inference, but is irrelevant from a purely logical/computational perspective.

Examples of such logical frameworks include type systems with graded monads and/or graded exponentials~\cite{GaboardiKOBU16,FujiiKM16}, which are structurally similar to indexed linear logic but does not include it, lacking in the way they treat additives and dependency. \shortVOnly{In the literature, some logical systems do implement this kind of dependency, albeit partially. $\IndLL$ is one of them, so is the Bounded Linear Logic~\cite{BLL} (BLL) for example. The first one still lacks a modular generalisation, hence the present companion papers. The second one has recently been made modular~\cite{FukiKatsu21}, but fails to take into account linear logic additives.}

\paragraph{IndLL: a logic of Intersection Types}
Introduced in the early 2000's to study sequentiality, indexed linear logic ($\IndLL$) \cite{IndLL1,IndLL2}  takes its name from its main syntactic feature: a formula is defined over an appropriate \emph{set of indexes}, that we call its \emph{locus}. The intuition behind this is that formulae of $\IndLL$ corresponds to \emph{family of refinements} of standard linear logic formulae. This intuition can be explained by the means of \emph{relational semantics} \cite{BEM07, ong:rel}. In the relational model of linear logic, formulae are interpreted as sets and proofs as relations. Since the work of de Carvalho \cite{carv:sem}, it is well established that, given a $\LL$ formula $ A,$ elements $ a \in \sem{}{A} $ can be seen as formulae of an \emph{auxiliary language}, which \emph{approximates} $\LL$. Through the semantics one can reconstructs all the possible executions of a program, each elements of the interpretation of a term corresponding to its computational beahviour in a particular environement. In this way, the relational semantics constitutes an \emph{approximation theory} of programs. Indexed linear logic is a way to give a logical presentation of the relational semantics and this approximation phenomenon. 

This approximation theory can also be defined as a system of \emph{intersection types}, as shown in \cite{carv:sem}. Ehrhard proved that also indexed linear logic gives rise to appropriate intersection types, that are non-idempotent and \emph{uniform}~\cite{er:intind}. Non-idempotency means that the condition $ a \cap  a = a $ does not hold. In this way, the intersection type system becomes \emph{resource sensitive}. By uniform we mean intersection types that are refinement of the same formula of linear logic (or the same simple type when restricted to the $\lambda$-calculus). Hence, intersections are only allowed between types of the same `shape'. The main interest, for us, is that this uniform restriction allows intersection types to be presented as a \emph{logic system} contained in $\IndLL$, as shown in \cite{er:intind}. This restriction to ``uniform'' intersection types can be simply understood as focusing on typed programming languages. Here ``intersection types'' are considered \emph{a posteriori} or \emph{à la Curry}, while the ``simple types'' of programs are considered \emph{a priori} or \emph{à la Church}. Our generalisation of the indexed system could then produce a \emph{modular} notion of uniform intersection types. This would lead to a purely logical and remarkably general approach to intersection types, that is quite different in nature from the others known approaches in the literature, such as \cite{MazzaPV18, ol:intdist}. However, we leave the proper development of these speculations to future work. In the present paper we shall focus on obtaining an indexed logic whose corresponding intersection type system is \emph{idempotent}.

\paragraph{Related work}
Ehrhard and Bucciarelli originally introduced $\OIndLL$ as the internal syntax of the phase model of linear logic. They, later, used this framework, with the help of Ehrhard's PhD students Alexandra Bac~\cite{IndLL2ndO} and Pierre Boudes~\cite{Boudes11}, to explore the notion of hypercoherence~\cite{Ehrhard93}. Indeed, provable formulae in $\OIndLL$ are all sets of coherent relational points, \emph{i.e.}, sets of points of the relational model~\cite{BEM07} which would be coherent in the hypercoherent model of $\LL$. Indexed linear logic has also been used by Grellois and Melliès~\cite{GrelloisM15b} to get a better grasp of intersection types in presence of type fixpoints (more precisely for higher order recursive schemes). It has also recently been used by Hamano~\cite{Hamano20} to interpret additives in ludics.

\paragraph{Our contribution} We aim to give a fresh presentation of $\IndLL$, by defining its fundamental features in an a more modular fashion. Moreover, we also want to obtain an idempotent version of this logical framework. While $ \OIndLL $ corresponds to the multiset-based relational semantics and hence to non-idempotent intersection types, we want to find the good framework to present the \emph{Scott semantics} of linear logic and hence \emph{idempotent} intersection types. For this paper we chose to highlight the relationship between $\IndLL$ and intersection types, keeping its denotational semantics counterpart between the lines. 
Our major contribution are the following:
\begin{itemize}
\item A new presentation of $\IndLL$ syntax by making it more explicit and ready for generalisation. 
 Two fundamental operations on the syntax are introduced: \emph{base change} and \emph{subtyping} (Section \ref{sec:sub}). 
\item An in-depth study of the relationship between our (idempotent) $\IndLL$ and (idempotent) intersection types.
\item The definition of a cut elimination procedure, that is shown to be normalising and confluent up to an appropriate congruence on proofs (Section \ref{sec:cut}). The congruence is induced by the subtyping relation and it is necessary due to technical issues that arises in the interaction between indexes and cut elimination. 
\end{itemize}

Section \ref{sec:sem} gives some intuition about the denotational semantics for our system. The full investigation of this line will be the topic of another forthcoming paper. Section \ref{sec:conc} discusses the relationship between our system and the original indexed one and possibilities to suitably generalise our syntax. Indeed, while in this paper we restrict to indexed formulae defined over sets, in the last section we discuss how the category of loci could be suitably generalised to appropriate semi-extensive categories.

\subsection{Preliminaries}
\paragraph{Classical Linear Logic}
Familiarity with linear logic ($\LL$)~\cite{LLdef} is expected from the reader. In particular, we expect the reader to know about its multiplicative fragment $(1,\otimes)$, its additive fragment $(\top,\with)$ and how the exponential $(!)$ links the two through Seely isomorphism $!(A\with B)\eqvdash !A\otimes!B$.

We are using a specific presentation of linear logic which is the classical linear logic. This means that (i) we can have multiple formulae both left and right of the sequents, (ii) every operation has a de Morgan dual obtained through an involutive negation, and (iii) the linear arrow will only exist as a construction $A\multimap B:=A^\bot\parr B$ through the de Morgan dual $(\parr)$ of $(\otimes)$.

\longVOnly{
  \paragraph{Fixpoints and $\mu \LL$}
}

\paragraph{Set-Theoretic Notions} Given two function $ f : A \to B $ and $g : B \to C $ we write $ f ; g : A \to C $ for their composition.
\emph{Quasi-injective functions} are functions $f:I\rightarrow J$ such that $\inv{f}(x):=\{y\mid f(y)=x\}$ is finite for every $x\in J$. This property will often appear as condition, but it is never crucial, except for Theorems~\ref{th:ITformulae} and~\ref{th:ITproofs}.

The \emph{pullback} of two function $f:I\rightarrow K$ and $g:J\rightarrow K$ is the set $I\times_KJ:=\{(x,y)\mid f(x)=g(y)\}$. The pullback of $f$ along $g$ is the projection $\pb{f}{g}:I\times_KJ\rightarrow I$. By abuse of notation we use the same notation $\pb{g}{f}$ for the pullback of $g$ along $f$. The pullback of a quasi-injective function is always quasi-injective.

Two functions with the same target are said to be \emph{orthogonal} if their pullback is empty and jointly surjective if their images cover their target. When two injections $i$, and $j$ are both orthogonal and jointly surjective, we use the notation $i\Bot j$. We will also often use the fact that $\pb i f \Bot \pb j f$ whenever $i\Bot j$ (this is called \emph{extensivity}).

The \emph{coproduct} of two sets $I$ and $J$ is their disjoint union, which we denote as either $I\oplus J$ or $ J\uplus K$. We denote by $\iota_1:I\rightarrow I\oplus J$ and $\iota_2:J\rightarrow I\oplus J$ the two orthogonal, jointly-surjective injections.

Finally, we will use $\init_I :\emptyset \rightarrow I$, or just $\init$, for the \emph{initial} (quasi-injective) function. Similarly, we will use $[f,g]:I\oplus J\rightarrow K$ for the \emph{copairing} of $f:I\rightarrow K$ and $g:J\rightarrow K$, which is quasi-injective iff $f$and $g$ both are.

\paragraph{Intersection types} We define the set of \emph{intersection types} $\mathcal{IT} $ by the following inductive grammar: 
 \[a ::=\ \alpha\ \mid\ {[a_1...a_n]} \rightarrow {a}  \mid {a} \with {\bullet} \mid {\bullet} \with {a} \longVOnly{\mid {a} \oplus {\bullet} \mid {\bullet} \oplus {a} } \qquad n \in \mathbb{N}\] 
where $[a_1...a_n]$ is a finite multiset and $\alpha$ an IT-variable. The intersection type is then strictly associate, commutative but not idempotent. In order to get idempotency, we could follow different paths.
 In the literature, there are at least three ways to proceed:
\begin{itemize}[leftmargin=0em,itemindent=1em]
\item \cite{Berline} Exponentials are sets and the subtyping relation is a preorder whose equivalence equates $[a,b]\simeq[a,a',b]$ whenever $a'\simeq a$.
\item \cite{Breuvart14} Exponentials are antichains, i.e., sets with no comparable elements, and the subtyping relation is a partial order.
\item \cite{IT,er:scott} Exponentials are multisets and the subtyping relation is a preorder whose equivalence equates the types $[a,b]\simeq [a,a,a',b]$ whenever $a'\simeq a$. In this way, the subtyping deals with the idempotency issue, performing a contraction of different copies of the same type.
\end{itemize}
We opt for the last option which will simplify the relationship with $\IndLL$. Intersection types will then be implicitly considered up-to the equivalence $(\simeq)=(\le)\cap(\ge)$ extracted from the subtyping preorder generated by the following rules:
\begin{center}
  \AxiomC{$\forall i,\exists j, b_i\le a_j$}
  \AxiomC{$\!\!a\le b$}
  \BinaryInfC{${{[a_1...a_n]} \rightarrow {a}}\le {{[b_1...b_k]}\rightarrow {b}}$}
  \DisplayProof\
  \AxiomC{$a\le b$}
  \UnaryInfC{$\!{a} \with {\bullet} \le {b}\with {\bullet}\!$}
  \DisplayProof\ 
  \AxiomC{$a\le b$}
  \UnaryInfC{$\!{\bullet}\with {a}\le {\bullet} \with {b}\!$}
  \DisplayProof
  \longVOnly{\ 
    \AxiomC{$a\le b$}
    \UnaryInfC{$a{\oplus}\bullet\le b{\oplus}\bullet$}
    \DisplayProof
    \AxiomC{$a\le b$}
    \UnaryInfC{$\bullet{\oplus}a\le \bullet{\oplus}b$}
    \DisplayProof
  }
\end{center}

We want to see intersection types as appropriate \emph{refinements} of simple types. This idea comes directly from the relational semantics of linear logic. In order to do so, first we define an inductive grammar for the set of \emph{simple types} $ \mathcal{ST}: $
 \[A,B ::=\ \tau\ \mid\ A \rightarrow B \mid A \with B  \longVOnly{\mid A \oplus B}\]
where $\alpha $ is a ST-variable.
The \emph{uniform} intersection types $\mathcal{I\!T}_{\!\!\!\mathtt{simple}}$ are the intersection types that \emph{refine} a simple type $A$, as defined by the following relation:
\begin{center}
  \AxiomC{\vphantom{A}}
  \UnaryInfC{$\alpha:\tau$}
  \DisplayProof\  
  \AxiomC{$\forall i, a_i: A_1$}
  \AxiomC{$a: A_2$}
  \BinaryInfC{${[a_1...a_n]} \rightarrow { a} : {A_1} \rightarrow  {\tau_2}$}
  \DisplayProof\\[.2em] 
  \AxiomC{$a: A_2$}
  \UnaryInfC{${\bullet} \with {a}: A_1\with A_2$}
  \DisplayProof\ 
  \AxiomC{$a: A_1$}
  \UnaryInfC{${a} \with {\bullet}:A_1 \with A_2$}
  \DisplayProof
  \longVOnly{\quad
  \AxiomC{$a: A_2$}
  \UnaryInfC{$\bullet{\oplus}a:A_1 \oplus A_2$}
  \DisplayProof\quad
  \AxiomC{$a: A_1$}
  \UnaryInfC{$a{\oplus}\bullet:A_1\oplus A_2$}
  \DisplayProof}
\end{center}
The \emph{intersection type system} for a $\lambda $-calculus with pairings is defined by induction as follows:
\begin{center}
  \AxiomC{$\exists i,a_i\le a$}
  \UnaryInfC{$\Gamma,x:[a_1...a_n]\vdash x:a$}
  \DisplayProof\quad 
  \AxiomC{$\Gamma,x:[a_1...a_n]\vdash t:b$}
  \UnaryInfC{$\Gamma\vdash \lambda x.t : [a_1...a_n]{\rightarrow}b$}
  \DisplayProof\\[.3em]
  \AxiomC{$\Gamma\vdash t_1 : [a_1...a_n]{\rightarrow}b$}
  \AxiomC{$\forall i,\ \Gamma\vdash t_2 : a_i$}
  \BinaryInfC{$\Gamma\vdash t_1t_2 : b$}
  \DisplayProof\ 
  \AxiomC{$\Gamma\vdash t_1 : a$}
  \UnaryInfC{$\Gamma\vdash (t_1,t_2) : a{\with}\bullet$}
  \DisplayProof\\[.3em]
  \AxiomC{$\Gamma\vdash t_2 : a$}
  \UnaryInfC{$\Gamma\vdash (t_1,t_2) : \bullet{\with}a$}
  \DisplayProof\quad
  \AxiomC{$\Gamma\vdash t : a{\with}\bullet$}
  \UnaryInfC{$\Gamma\vdash \mathtt{fst}(t) : a$}
  \DisplayProof\quad
  \AxiomC{$\Gamma\vdash t : \bullet{\with}a$}
  \UnaryInfC{$\Gamma\vdash \mathtt{snd}(t) : a$}
  \DisplayProof
  \longVOnly{
    \quad
    \AxiomC{$\Gamma\vdash t_1 : a$}
    \UnaryInfC{$\Gamma\vdash \mathtt{i_1}(t_1) : a{\oplus}\bullet$}
    \DisplayProof\quad
    \AxiomC{$\Gamma\vdash t_2 : a$}
    \UnaryInfC{$\Gamma\vdash \mathtt{i_2}(t_2) : \bullet{\oplus}a$}
    \DisplayProof\\[.3em]
    \AxiomC{$\forall i, \Gamma\vdash t : a_i{\oplus}\bullet$}
    \AxiomC{$\Gamma,x:[a_1...a_n]\vdash t_1 : b$}
    \BinaryInfC{$\Gamma\vdash \mathtt{match}\ t\ \mathtt{with}\ \mathtt{i_1}(x)\mapsto t_1 \ \mathtt{i_2}(x)\mapsto t_2  : b$}
    \DisplayProof\\[.3em]
    \AxiomC{$\forall i, \Gamma\vdash t : \bullet{\oplus}a_i$}
    \AxiomC{$\Gamma,x:[a_1...a_n]\vdash t_2 : b$}
    \BinaryInfC{$\Gamma\vdash \mathtt{match}\ t\ \mathtt{with}\ \mathtt{i_1}(x)\mapsto t_1 \ \mathtt{i_2}(x)\mapsto t_2  : b$}
    \DisplayProof
    }
\end{center}

The idempotency of the system is determined both by the quotient induced by subtyping and the implicit \emph{contraction} of contexts in the rules with multiple hypothesis.

\paragraph{Coloration}
Along the article we use several colours as syntax highlighting. The highlighting does not carry any additional information and its purpose is to ease the reading of complex equations. The authors also have a colourblind-friendly version using the Wong colour palette, which is available on request.
 
\section{Formulae and subtyping}\label{sec:sub}

\subsection{Formulae}

We assume that we are given, for any (small) set $I$, an infinite set $\mathtt{var}(I)$ of variables called ``type variables over $I$''.

\begin{definition}[Pre-formulae]
  The pre-formulae of $\IndLL$ are given by the same grammar as \longVOnly{$\mu$}LL excepts for \longVOnly{fixpoints, }exponentials and variables \longVOnly{$X \in \bigcup_{I\in\Set}\mathtt{var}(I)$} which are supplemented by annotations. For the sake of readability, we use a \nonColorblind{violet}\colorblind{pink} syntax highlighting on annotations:
  \longVOnly{
  \begin{align*}
    X &\in \bigcup_{I\in\Ob\Set}\mathtt{var}(I)\\
    u,f,i,j &\ \text{quasi injective functions}\\
    A,B &::=\ \varm{f}{X}\ \mid\ \varm{f}{X}^\bot\ \mid\ \1 \ \mid\ \bot\ \mid\ \0\ \mid\ \top\ \mid\ A\otimes B\ \mid\ A\parr B\ \\
    &\quad\ \mid\ A\oplusm{i}{j} B\ \mid\ A\withm{i}{j} B\ \mid\ \bangm uA\ \mid\ \wnm uA\ \mid\ \mum{f}X.A\ \mid\ \num{f}X.A
  \end{align*}
  }
  \shortVOnly{
    \begin{align*}
      A,B &::=\
      \varm{f}{X}\ \mid \1 \mid \top\mid A\otimes B\mid A\withm{i}{j} B\mid \bangm uA\\
      &\quad\mid\varm{f}{X}^\bot\mid \bot\mid \0\mid  A\parr B\ \mid A\oplusm{i}{j} B\mid \wnm uA
    \end{align*}
    where $X$ is a variable among one of the $\mathtt{var}(I)$ and where each of the $u,f,i,j$ are quasi-injective functions.
  }
\end{definition}
\begin{definition}[Negation]
 Negation is defined inductively over pre-formulae as usual,\longVOnly{ (except for the variables for which $(\varm{f}{X}^\bot)^\bot := \varm{f}{X}$)}
  \begin{align*}
    \1^\bot &:= \bot &
    (A\otimes B)^\bot &:= A^\bot\parr B^\bot &
    (\bangm uA)^\bot &:= \wnm u (A^\bot)\\
    \bot^\bot &:= \1 &
    (A\parr B)^\bot &:= A^\bot\otimes B^\bot &
    (\wnm uA)^\bot &:= \bangm u (A^\bot)  \\
    \top^\bot &:= \0 &
    (A\withm{i}{j} B)^\bot &:= A^\bot\oplusm{i}{j} B^\bot &
    \longVOnly{(\mum fX.A)^\bot &:= \num fX.A^\bot}
    \shortVOnly{(\varm{f}{X})^\bot&:=\varm{f}{X}^\bot}\\
    \0^\bot &:= \top &
    (A\oplusm{i}{j} B)^\bot &:= A^\bot\withm{i}{j} B^\bot &
    \longVOnly{(\num fX.A)^\bot &:= \mum fX.A^\bot}
    \shortVOnly{(\varm{f}{X}^\bot)^\bot &:= \varm{f}{X}}
  \end{align*}
\end{definition}
Since negation is blind regarding indexation, and \emph{vice versa}, many proofs and definition will focus on positive operators.

\begin{definition}
  We write $\underline A$ for the LL formula obtained from the indexed formula $A$ by forgetting the index annotations:
  \begin{align*}
    \underline{\varm{f}{X}} &:= \underline{X} &
    \underline{\varm{f}{X}^\bot} &:= \underline{X}^\bot &
    \longVOnly{
      \underline{\mum gX.A} &:= \mu X.\underline{A} &
      \underline{\num gX.A} &:= \nu X.\underline{A}\\
      \underline{\1} &:= \1 &
      \underline{A\otimes B} &:= \underline{A}\otimes \underline{B} &
      \underline{\bot} &:= \bot &
      \underline{A\parr B} &:= \underline{A}\parr \underline{B} \\
      \underline{\top} &:= \top &}
    \underline{A\withm{i}{j} B} &:= \underline{A}\with \underline{B}  &
    \longVOnly{\underline{\0} &:= \0 &}
    \shortVOnly{\\}
    \underline{A\oplusm{i}{j} B} &:= \underline{A}\oplus \underline{B}
    \longVOnly{\\}\shortVOnly{&}
    \underline{\bangm uA} &:= !\underline{A}  &
    \underline{\wnm uA} &:= ?\underline{A} 
  \end{align*}
\end{definition}

\begin{definition}[Formulae over loci]

We define a \emph{correctness relation} $I\vdash A \ \mathtt{def}  ,$ where $ I  $ is a set and $A  $ a preformula, by induction on $A $ as follows: 
  \begin{itemize}
  \item $I\vdash \varm{f}{X} \ \mathtt{def}$ iff $X\in \mathtt{var}(J)$ and $f:I\rightarrow J$ for some set $J$.
  \item $I\vdash \bangm uA\ \mathtt{def}$ and $I\vdash \wnm uA $ iff $J\vdash A\ \mathtt{def}$ and $u:J\rightarrow I$ for some set $J$. 
    \longVOnly{
    \item $I\vdash \mum fX.A\ \mathtt{def}$ iff $J\vdash A\ \mathtt{def}$ and $X\in \mathtt{var}(J)$ with $f:I\rightarrow J$
    }
  \item $I\vdash \1\ \mathtt{def}$ always.
  \item $I\vdash A\otimes B\ \mathtt{def}$ iff $I\vdash A\ \mathtt{def}$ and $I\vdash B\ \mathtt{def}$.
  \item $I\vdash \0\ \mathtt{def}$ iff $I=\emptyset .$
  \item $K\vdash A\oplusm{i}{j} B\ \mathtt{def}$ iff all are verified:
    \begin{itemize}
    \item $I\vdash A\ \mathtt{def}$ and $J\vdash B\ \mathtt{def}.$
    \item $i:I\rightarrow K$ and $j:J\rightarrow K.$
    \item $i$ and $j$ are orthogonal injection covering $K$, i.e., $i$ and $j$ are the left and right injections of $I\uplus J \simeq K$.
    \end{itemize}
  \item $I\vdash A\ \mathtt{def}^\bot$ iff $I\vdash A\ \mathtt{def} .$
  \end{itemize}
  
  If $ I\vdash A \: \mathtt{def} $ we say that $ I $ is the \emph{locus } of $A $ and $ A $ is \emph{defined} under~$ I.$ A \emph{formula} of $\IndLL$ is a preformula s.t. there exists a set $ I$ with $ I\vdash A \ \mathtt{def} . $

\end{definition}

\subsection{Intuitive explanation of formulae}

This is a good point to introduce informally the connection with intersection types. One can see a formula $A$ over a locus $I$ as a $I$-indexed family of intersection types refining the type $\underline A$.
\begin{itemize}
\item An $I$-indexed set of intersection types refining $\underline A \otimes \underline B$ is given, for any $x\in I$ by a couple of an intersection type refining $\underline A$ and another refining $\underline B$. Thus the definition of $A\otimes B$ as two sets indexed by the same $I$.
\item An $I$-indexed set of intersection types refining $\1$ is necessarily the $I$ copies of the only intersection type for $\1$. 
\item A $K$-indexed set of intersection types refining $\underline A \oplus \underline B$ is given by a partition $K\simeq I\uplus J$ between the intersection types of the form $a\oplus \bullet$, morally in $\underline A$, and those of $\bullet\oplus b$, morally in $\underline B$. Thus the definition of $A\oplusm ij B$, where we keep track of $i$ and $j$ such that $[i,j]:(I\uplus J) \simeq K$ because it simplifies the proof system.
\item Since $\0$ is not refined by any intersection type, an $I$-indexed set of intersection types refining $\0$ can only be an empty set, thus $I$ has to be empty.
\item The exponentials are the crucial case, since their refinement is connected to the proper intersection type constructor. An $I$-indexed set of intersection types refining $ \bangm uA$  or $ \wnm uA $ is given by a $ J$-indexed set of refinements for $ A $ together with the function $ u : J \to I . $ Now, $ u $ is quasi-injective, meaning that $ f^{-1}(i) $ is finite for all $ i \in I . $ An element  $ i \in I$ should be thought as the index of the \emph{intersection} of the refinements indexed by the elements of $ f^{-1}(i)$.
\end{itemize}

Indexes are invariant under negation. This happens because negation concerns only the \emph{logical} aspect of formulae.

\subsection{Base change}

Having access, not just to one intersection type, but a family of intersection types gives a new dimension with which we can manipulate formulae. Indeed, we can restrict to a subset of refinements of an LL formula, merge two isomorphic refinements, or reindexing the family with an isomorphic locus $I'\simeq I$.

All of these operations happen to be variants of a single one: a contravariant \emph{action} of functions $f:I\rightarrow J$ which transforms a formula $A$ under the locus $J$ into a formula $\indfun{f}{A}$ under the locus $I$.

One have to keep in mind that the syntactic behaviour of this base change is akin to negation: it is not a proper operator but a transformation of formulae that flows through their syntactic trees until it reaches variables. We use use a blue syntax highlighting for this operation. Base change was also present but not spelled out explicitly in the original $\IndLL$. However, they only had injective base change (which, as we will see, corresponds to the non-impotency) and they separate re-indexing and restriction of loci into two separate operations.

\begin{definition}[Base changes]
  Given a  formula $J\vdash A\ \mathtt{def}$ and a quasi-injective function $f:I\rightarrow J$, we can define ${I\vdash\indfun{f}{A}\ \mathtt{def}}$ by induction on $A $ as follows.\\
The multiplicative operators simply distribute over $f$:
  \begin{align*} 
    \indfun{f}{\1} &:= \1 &
    \indfun{f}{A\otimes B} &:= \indfun{f}{A}\otimes \indfun{f}{B} \\
    \indfun{f}{\bot} &:= \bot &
    \indfun{f}{A\parr B} &:= \indfun{f}{A}\parr \indfun{f}{B} 
  \end{align*}
  The additive units are trivial since, the only morphism ${f:I\rightarrow J}$ targeting $J=\emptyset$ is the identity. For the additive operators we exploit the pullback construction. Let $ K \vdash \ {A\withm{i}{j} B} \mathtt{def} $ or $ K \vdash {A\oplusm{i}{j} B} \ \mathtt{def} $. 
  \begin{align*}
    \indfun{f}{\top} &:= \top &
    \indfun{f}{A\withm{i}{j} B} &:= \indfun{\pb{f}{i}}{A}\withm{\pb{i}{f}}{\pb{j}{f}} \indfun{\pb{f}{j}}{B}  \\
    \indfun{f}{\0} &:= \0 &
    \indfun{f}{A\oplusm{i}{j} B} &:= \indfun{\pb{f}{i}}{A}\oplusm{\pb{i}{f}}{\pb{j}{f}} \indfun{\pb{f}{j}}{B}  
  \end{align*}
  where $\pb{f}{i}$ (resp. $\pb{f}{j}$) is the pullback of $f:I\rightarrow J$ along $i : I \hookrightarrow K$ (resp. $ j : J \hookrightarrow K $). The base change of exponential is similarly given by:
  \begin{align*}
    \indfun{f}{\bangm uA} &:= \bangm{\pb{u}{f}}\indfun{\pb{f}{u}}{A}  &
    \indfun{f}{\wnm uA} &:= \wnm{\pb{u}{f}}\indfun{\pb{f}{u}}{A}
  \end{align*}
  \shortVOnly{
  For the variable case we use pre-composition:
  \begin{align*}
    \indfun{f}{\varm{g}{X}} &:= \varm{(f;g)}{X} &
    \indfun{f}{\varm{g}{X}^\bot} &:= \varm{(f;g)}{X}^\bot
  \end{align*}
  }
  \longVOnly{
  The base change of fixpoints and variables is stacked in the indices:
  \begin{align*}
    \indfun{f}{\mum gX.A} &:= \mum{f;g}X.A & \indfun{f}{\varm{g}{X}} &:= \varm{(f;g)}{X}\\
    \indfun{f}{\num gX.A} &:= \num{f;g}X.A & \indfun{f}{\varm{g}{X}^\bot} &:= \varm{(f;g)}{X}^\bot
  \end{align*}
  }
\end{definition}
Notice the similarities between (co-)products and exponentials: in both cases, the indexed morphism acts on the base-change morphism and vice-versa. In particular, assuming $ f : I \to J , $ for the exponentials we have a transformation
\[
\AxiomC{$K\vdash A\ \mathtt{def}$}
\AxiomC{\hspace{-1.5em}$u:K{\rightarrow} J$}
\BinaryInfC{$J\vdash \bangm u A\ \mathtt{def}$}
\DisplayProof
\Rightarrow
\AxiomC{$I{\times_{J}}K\vdash \indfun{\pb{f}{u}}{A}\ \mathtt{def}$}
\AxiomC{\hspace{-1.5em}$\pb{u}{f}:I{\times_{J}}K{\rightarrow} I$}
\BinaryInfC{$I\vdash \bangm{\pb{u}{f}}\indfun{\pb{f}{u}}{A}\ \mathtt{def}$}
\DisplayProof
\]
The additive case is slightly more subtle as we have to verify that the pair of pullbacks $(\pb{i}{f},\pb{j}{f})$ are the injections of a co-product, but this is always true in $\mathtt{Set}$, this property is called extensivity.

The intuition for the exponential case is the following: Let us consider an indexed formula $\bangm u A$ seen as a $J$-indexed set $\{[a_k]_{k\in u^{-1}(x)}\}_{x\in J}$ of $\left(u^{-1}(x)\right)$-indexed intersection types over $\underline A$. If $f$ is an inclusion $f:I\subseteq J$; then, up-to-iso, the pullback $\pb u f$ is the restriction of $u$ to $u^{-1}(I)$ and $\pb{f}{u}$ is the inclusion of $u^{-1}(I)$ into $K$, thus $\bangm{\pb{u}{f}}\indfun{\pb{f}{u}}{A}$ corresponds to $\{[a_k]_{k\in u^{-1}(x)}\}_{x\in I}$ which is the restriction of $\bangm u A$ to the bags indexed by $x\in I$. Similarly, if $f:I\twoheadrightarrow J$ is a surjective but non-bijective function; then, up-to-iso, $\bangm{\pb{u}{f}}\indfun{\pb{f}{u}}{A}$ corresponds to $\{[a_k]_{k\in u^{-1}(f(y))}\}_{y\in J}$ which are the same types, but with duplicates (recall that an indexed set is more like a multiset).

\begin{lemma}
  \hspace{1em}
    $\underline{\indfun{f}{A}}=\underline{A}$
  \hspace{2em}
    and
  \hspace{2em}
    $\indfun{f}{A^\bot} = \indfun{f}{A}^\bot$ .
\end{lemma}

\longVOnly{
\begin{definition}[Variable substitutions and $\alpha$-equivalence]
  The $\alpha$-equivalence is defined in a standard way:
  \begin{center}
    \AxiomC{$X,Y\in\mathtt{var}(codom(f))$}
    \UnaryInfC{$\mum{f}X.A = \mum fY.A[Y/X]$}
    \DisplayProof
  \end{center}
  The substitution is defined under variables by applying the pending $f$ and eventual pending $-^\bot$:
  \begin{align*}
    \varm{f}{X}[A/X] &:= \indfun{f}{A} &
    \varm{f}{X}^\bot[A/X] &:= \indfun{f}{A}^\bot\\
  \end{align*}
  On binders $\mu$ and $\nu$, it is defied, as usual, only for $X\neq Y$ and for $Y$ not appearing in $A$, implicitly using $\alpha$-equivalence:
  \begin{align*}
    (\mum fY.B)[A/X] &:= \mum fY.B[A/X] &
    (\num fY.B)[A/X] &:= \num fY.B[A/X] \\
  \end{align*}
  It distributes with all the other operators.
\end{definition}
}

\subsection{Subtyping}
Idempotent intersection types are defined together with a subtyping relation $(\le)$, which is critical to get subject reduction. In our indexed linear logic, it is even more so. We use a green syntax highlighting for the subtyping.

Recall that we see intersections as multisets. For $\IndLL$, we have to go a bit further and to see them as $I$-indexed sets $\bangm{u}A$ for $u:I\rightarrow \1$, that behave like multisets. In order to achieve idempotency through subtyping, we need the following conditions to be satisfied.
\begin{itemize}
\item $\bangm{u}A \subtype \bangm{g\circ u}\indfun{g}{A}$ when $g:J\rightarrow I$ is a surjection, i.e., the $J$-indexed set $\bangm{g\circ u}\indfun{g}{A}$ has, at position $x\in J$ a copy the $g(x)$th element of $\bangm{u}A$, 
\item $\bangm{u}A \subtype \bangm{g\circ u}\indfun{g}{A}$ when $g:J\subseteq I$, i.e., i.e., the $J$-indexed set $\bangm{g\circ u}\indfun{g}{A}$ is a subset of $\bangm{u}A$.
\end{itemize}
Remarkably, these two cases are subsumed by just one rule of subtyping.

Our subtyping also serves an additional goal: it equates index-renaming and equivalent base changes. Indeed, as we shall see in Theorem~\ref{lemma:order_base_change}, the base change is not functorial since $\indfun{f}{\indfun{g}{A}}\neq \indfun{(f;g)}{A}$ but it will be `pseudofunctorial' relatively to the equivalence resulting of the subtyping preorder, i\textit{i.e.}, $\indfun{f}{\indfun{g}{A}}\eqtype \indfun{f;g}{A}$.

In the definition of our subtyping relation we are following Melliès and Zeilberger's philosophy~\cite{MellZeil13}: $\IndLL$ can be seen as type refinement of linear logic, which correspond to the refinement functor $A\mapsto\underline{A}$. The subtyping $A\subtype A'$ should correspond to derivable proofs $A\indvdash I A'$ which are mapped to the identity following this functor. This assertion will be made precise in Proposition~\ref{MZ_stype_subtyping}.
 Moreover, the subtyping will be \emph{proof-relevant}. This means that the syntactic derivation of a subtyping $A\subtype B$ will be explicitly used. 
 This happens because we are seeing subtyping as corresponding to $\IndLL$  proofs that refines identity proofs of LL. In the last section of the article, we will quotient formulae and derivations up to the congruence generated by the subtyping. This is required to get a confluent cut-elimination procedure.

\begin{definition}[Subtyping]
  \shortVOnly{
    The subtyping relation between formulae over the same locus is the relation generated by the following rules:
    \begin{itemize}[noitemsep,topsep=0pt,parsep=0pt,partopsep=0pt,leftmargin=1em,itemindent=0em]
    \item for the product:
      \begin{center}
        \AxiomC{$i\Bot j'\quad j\Bot i'$}
        \AxiomC{$\indfun{\pb{i'}{i}}{A}\subtype \indfun{\pb{i}{i'}}{A'}$}
        \AxiomC{$\indfun{\pb{j'}{j}}{B}\subtype \indfun{\pb{j}{j'}}{B'}$}
        \TrinaryInfC{$A\withm{i}{j}B \subtype A'\withm{i'}{j'} B'$}
        \DisplayProof
      \end{center}
      where $i\Bot j'$, here, means that the pullback of $i$ and $j'$ is empty.\footnote{By abuse of notation, we may omit the two first conditions, but they are not redundant and shall always be verified.}
    \item for the other positive cases:
      \begin{center}
        \AxiomC{$\indfun{g}{A}\subtype A'$}
        \UnaryInfC{$\bangm{u}A \subtype \bangm{u\circ g}A'$}
        \DisplayProof\  
        \AxiomC{\vphantom{A}}
        \UnaryInfC{$\1 \subtype \1$}
        \DisplayProof\ 
        \AxiomC{\vphantom{A}}
        \UnaryInfC{$\top \subtype \top$}
        \DisplayProof\ 
        \AxiomC{$A\subtype A'$}
        \AxiomC{$B\subtype B'$}
        \BinaryInfC{$A\otimes B \subtype A'\otimes B'$}
        \DisplayProof        
      \end{center}
    \item and the symmetric versions for negative operators so that
      \begin{center}
        \AxiomC{$A\suptype B$}
        \dashedLine
        \UnaryInfC{$A^\bot \subtype B^\bot$}
        \DisplayProof
      \end{center}
    \end{itemize}  }
  \longVOnly{
    The subtyping relation between formulae over the same locus is the relation generated by the following rules, which are to be considered coinductively over fixpoints:
    \begin{itemize}
    \item the most important rules:
      \begin{center}
        \AxiomC{$i\Bot j'$ \quad $j\Bot i'$}
        \AxiomC{$\indfun{\pb{i'}{i}}{A}\subtype \indfun{\pb{i}{i'}}{A'}$}
        \AxiomC{$\indfun{\pb{j'}{j}}{B}\subtype \indfun{\pb{j}{j'}}{B'}$}
        \TrinaryInfC{$A\withm{i}{j}B \subtype A'\withm{i'}{j'} B'$}
        \DisplayProof\\[2em]
        \AxiomC{$\indfun{g}{A}\subtype A'$}
        \UnaryInfC{$\bangm{u}A \subtype \bangm{g ; u}A'$}
        \DisplayProof
        \hskip 20pt
\AxiomC{$\indfun{f}{A[\mu X.A/X]}\subtype \indfun{g}{A'[\mu X.A'/X]}$}
        \productiveLine
        \UnaryInfC{$\mum{f}X.A \subtype \mum{g}Y.A'$}
        \DisplayProof
      \end{center}
      where $i\Bot j'$, here, means that the pullback of $i$ and $j'$ is empty.
    \item their symmetric version
      \begin{center}
        \AxiomC{$i\Bot j'$ \quad $j\Bot i'$}
        \AxiomC{$\indfun{\pb{i'}{i}}{A}\subtype \indfun{\pb{i}{i'}}{A'}$}
        \AxiomC{$\indfun{\pb{j'}{j}}{B}\subtype \indfun{\pb{j}{j'}}{B'}$}
        \TrinaryInfC{$A\withm{i}{j}B \subtype A'\withm{i'}{j'} B'$}
        \DisplayProof\\[2em]
        \AxiomC{$A\subtype \indfun{g}{A'}$}
        \UnaryInfC{$\wnm{u\circ g}A \subtype \wnm{u}A'$}
        \DisplayProof
        \hskip 20pt
\AxiomC{$\indfun{f}{A[\nu X.A/X]}\subtype \indfun{g}{A'[\nu X.A'/X]}$}
        \productiveLine
        \UnaryInfC{$\num{f}X.A \subtype \num{g}Y.A'$}
        \DisplayProof
      \end{center}
    \item and the contextual cases:
      \begin{center}
        \AxiomC{$A\subtype A'$}
        \AxiomC{$B\subtype B'$}
        \BinaryInfC{$A\otimes B \subtype A'\otimes B'$}
        \DisplayProof\hskip 20pt
        \AxiomC{$A\subtype A'$}
        \AxiomC{$B\subtype B'$}
        \BinaryInfC{$A\parr B \subtype A'\parr B'$}
        \DisplayProof\\[.5em]
        \AxiomC{}
        \UnaryInfC{$\varm{f}{X}\subtype \varm{f}{X}$}
        \DisplayProof\hskip 2pt
        \AxiomC{}
        \UnaryInfC{$\varm{f}{X}^\bot\subtype \varm{f}{X}^\bot$}
        \DisplayProof\hskip 2pt
        \AxiomC{}
        \UnaryInfC{$\1\subtype \1$}
        \DisplayProof\hskip 2pt
        \AxiomC{}
        \UnaryInfC{$\0\subtype \0$}
        \DisplayProof\hskip 2pt
        \AxiomC{}
        \UnaryInfC{$\bot\subtype \bot$}
        \DisplayProof\hskip 2pt
        \AxiomC{}
        \UnaryInfC{$\top\subtype \top$}
        \DisplayProof
      \end{center}
    \end{itemize}
  }
\end{definition}

\longVOnly{
\begin{lemma}
  The equivalence  is partially derived by the inference rules (as well as contextual cases): 
      \begin{center}
        \AxiomC{$i\Bot j'$ \quad $j\Bot i'$}
        \AxiomC{$\indfun{\pb{i'}{i}}{A}\eqtype \indfun{\pb{i}{i'}}{A'}$}
        \AxiomC{$\indfun{\pb{j'}{j}}{B}\eqtype \indfun{\pb{j}{j'}}{B'}$}
        \dashedLine
        \TrinaryInfC{$A\withm{i}{j}B \eqtype A'\withm{i'}{j'} B'$}
        \DisplayProof\\
        \AxiomC{$\indfun{g}{A}\eqtype A'$}
        \AxiomC{$g$ iso}
        \dashedLine
        \BinaryInfC{$\bangm{u}A \eqtype \bangm{u\circ g}A'$}
        \DisplayProof\hskip 20pt
        \AxiomC{$A[\varm{g}{Y}/X]\eqtype \indfun{g}{A'}$}
        \AxiomC{$g$ iso}
        \dashedLine
        \BinaryInfC{$\mum{f}X.A \eqtype \mum{g ; u}Y.A'$}
        \DisplayProof
      \end{center}
\end{lemma}
}

The exponential rule only makes sense whenever the domain of $ u$ is the same as the codomain of $ g. $ More explicitly, if $g : X \to Y  $ and $ Y \vdash A \ \mathsf{def} ,$ by definition of base change we have $ X \vdash g(A) \ \mathsf{def} .$ Now, since $ g(A) \subtype A' $ they share the same locus, that is $X $. To get a well-defined formula $ K \vdash \oc_{u} A \ \mathsf{def}$ we need that $ u : Y \to K $ for some set $ K . $ 

  The subtyping of additives is quite complex, seemingly more so than the exponential where the actual subtyping appears. This is due to the peculiar structure of additives, that needs to keep track of two different branches of indexes.

\begin{lemma}
  If $A\subtype B$, then $\underline A=\underline B$\longVOnly{ up-to $\alpha$-equivalence}.
\end{lemma}

\begin{lemma}[Functoriality of base change]\label{lemma:order_base_change}
  This subtyping determines pseudofunctoriality of base change, \textit{i.e.}, the following are derivable:
  \begin{center}
    \AxiomC{$\begin{matrix}f_1;\dots;f_n =\\h;g_1;\dots;g_k\end{matrix}$}
    \AxiomC{$\begin{matrix}f'_1;\dots;f'_{n'}=\\h;g'_1;\dots;g'_{k'}\end{matrix}$}
    \AxiomC{$\begin{matrix}\indfun{g_1}{\cdots\indfun{g_k}{A}\cdots}\\\subtype \indfun{g'_1}{\cdots\indfun{g'_{k'}}{A'}\cdots}\end{matrix}$}
    \dashedLine
    \TrinaryInfC{$\indfun{f_1}{\cdots\indfun{f_n}{A}\cdots}\subtype \indfun{f'_1}{\cdots\indfun{f'_{n'}}{A'}\cdots}$}
    \DisplayProof
  \end{center}
 Once the reflexivity and transitivity of the subtyping are proven, these conditions  can be seen as equivalent to their simplified versions:
  \begin{center}
    \AxiomC{$A\subtype A'$}
    \dashedLine
    \UnaryInfC{$\indfun{f}{A} \subtype \indfun{f}{A'}$}
    \DisplayProof\hskip 20pt
    \AxiomC{}
    \dashedLine
    \UnaryInfC{$\indfun{\id}{A}\subtype A$}
    \DisplayProof\hskip 20pt
    \AxiomC{}
    \dashedLine
    \UnaryInfC{$\indfun{f}{\indfun{g}{A}}\subtype \indfun{(g\circ f)}{A}$}
    \DisplayProof
  \end{center}
  Where the dashed line is there to recall that this is a proof transformation, since we are working in a proof-relevant framework.
\end{lemma}
\shortVOnly{
  \begin{proof}
    By induction on $A$, the difficult case are the exponentials, where an appropriate pullback construction has to be carried out in order to establish the subtyping.
  \end{proof}
}

\longVOnly{
  \begin{proof}
    We proceed by mutual induction on $A$. Multiplicative cases are direct consequences of the IH. Fixpoints, atoms and units are trivial. Thus, only remains the additives and exponentials. By duality, we will only present one of each:
    \begin{itemize}
    \item if $A= B\withm i j C$, then, $\underline{B\withm i j C}=\underline{\indfun{g_1}{\cdots\indfun{g_k}{A}\cdots}}=\underline{\indfun{g'_1}{\cdots\indfun{g'_{k'}}{A'}\cdots}}=\underline{A'}$,
      thus $A'$ is of the form $A'=B'\withm {i'} {j'} C'$.\\
      The hypothesis can only be obtained if the following are true: 
      \begin{align}
        \indfun{\qpb {i'} {{g'_{k'}}}\cdots{{g'_1}} {\tpb i {{g_k}} \dots {{g_1}}}}{\indfun{{\pb {{g_1}} {\tpb i {{g_k}}\cdots{{g_2}}}}}{\cdots\indfun{\pb{{{g_k}}}{i}}{B}}}&\subtype \indfun{\qpb i {{g_k}} \cdots {{g_1}} {\tpb {i'} {{g'_{k'}}}\cdots{{g'_1}}}}{\indfun{\pb{{g'_1}}{\tpb {i'} {{g_k}}\cdots {{g_1}}}}{\cdots\indfun{\pb{{g'_{n'}}}{i'}}{B'}}}\\
        \indfun{\qpb {j'} {{g'_{k'}}}\cdots{{g'_1}} {\tpb j {{g_k}} \dots {{g_1}}}}{\indfun{{\pb {{g_1}} {\tpb j {{g_k}}\cdots{{g_2}}}}}{\cdots\indfun{\pb{{{g_k}}}{j}}{C}}}&\subtype \indfun{\qpb j {{g_k}} \cdots {{g_1}} {\tpb {j'} {{g'_{k'}}}\cdots{{g'_1}}}}{\indfun{\pb{{g'_1}}{\tpb {j'} {{g_k}}\cdots {{g_1}}}}{\cdots\indfun{\pb{{g'_{n'}}}{j'}}{C'}}}
      \end{align}
      while the conclusion can be obtained if we prove both of the following inequations:
      \begin{align}
        \qpb {i'} {{f'_{n'}}}\cdots{{f'_1}} {\tpb i {{f_n}} \dots {{f_1}}} ({\pb {{f_1}} {\tpb i {{f_n}}\cdots{{f_2}}}}(\cdots(\pb{{{f_n}}}{i}(B)))\subtype \qpb i {{f_n}} \cdots {{f_1}} {\tpb {i'} {{f'_{n'}}}\cdots{{f'_1}}}(\pb{{f'_1}}{\tpb {i'} {{f_n}}\cdots {{f_1}}}(\cdots\pb{{f'_{n'}}}{i'}(B')))\\
        \qpb {j'} {{f'_{n'}}}\cdots{{f'_1}} {\tpb j {{f_n}} \dots {{f_1}}} ({\pb {{f_1}} {\tpb j {{f_n}}\cdots{{f_2}}}}(\cdots(\pb{{{f_n}}}{j}(B)))\subtype \qpb j {{f_n}} \cdots {{f_1}} {\tpb {j'} {{f'_{n'}}}\cdots{{f'_1}}}(\pb{{f'_1}}{\tpb {j'} {{f_n}}\cdots {{f_1}}}(\cdots\pb{{f'_{n'}}}{j'}(B')))
      \end{align}
      Since they are symmetric, we will only prove the first.\\
      We have
      \begin{align*}
        &\qpb {i'} {{f'_{n'}}}\cdots{{f'_1}} {\tpb i {{f_n}} \dots {{f_1}}};\tpb i {{f_n}} \dots {{f_1}};h;g_1;\dots;g_k\\
        &=\qpb {i'} {{f'_{n'}}}\cdots{{f'_1}} {\tpb i {{f_n}} \dots {{f_1}}};\tpb i {{f_n}} \dots {{f_1}};f_1;\dots;f_n\\
        &=\qpb {i'} {{f'_{n'}}}\cdots{{f'_1}} {\tpb i {{f_n}} \dots {{f_1}}};\pb {{f_1}} {\tpb i {{f_n}}\cdots{{f_2}}};\dots;\pb{{{f_n}}}{i};i
      \end{align*}
      thus there is a unique $s$ such that
      $$s;\tpb i {{g_k}} \dots {{g_1}} = \tpb i {{f_n}} \dots {{f_1}};h $$
      $$s;\pb {{g_1}} {\tpb i {{g_k}}\cdots{{g_2}}};\dots;\pb{{{g_k}}}{i} = \pb {{f_1}} {\tpb i {{f_n}}\cdots{{f_2}}};\dots;\pb{{{f_n}}}{i}$$
      We can obtain similarly the unique $s'$ such that
      $$s';\tpb {i'} {{g'_n}} \dots {{g'_1}} = \tpb {i'} {{f'_n}} \dots {{f'_1}};h $$
      $$s';\pb {{g'_1}} {\tpb {i'} {{g'_{k'}}}\cdots{{g'_2}}};\dots;\pb{{{g'_{k'}}}}{{i'}} = \pb {{f'_1}} {\tpb {i'} {{f'_{n'}}}\cdots{{f'_2}}};\dots;\pb{{{f'_{n'}}}}{{i'}}$$
      With those, we can obtain the unique $h'$ such that
      $$h';\qpb {i'} {{g'_{k'}}}\cdots{{g'_1}} {\tpb i {{g_k}} \dots {{g_1}}} = \qpb {i'} {{f'_{n'}}}\cdots{{f'_1}} {\tpb i {{f_n}} \dots {{f_1}}};s $$
      $$h';\qpb {i} {{g_{k}}}\cdots{{g_1}} {\tpb {i'} {{g'_k}} \dots {{g'_1}}} = \qpb {i} {{f_{n}}}\cdots{{f_1}} {\tpb {i'} {{f'_{n'}}} \dots {{f'_1}}};s' $$
      With this $h'$ we can apply our induction hypothesis.
    \item If $A=\bangm uB$, then $!\underline B = \underline{g_1(\cdots g_k(\bangm uB)\cdots)}=\underline {g'_1(\cdots g'_{k'}(A')\cdots)}=\underline A'$, thus $A'$ is of the form $A'=\bangm{u'}B'$.\\
      The hypothesis can only be obtained if $\tpb{u'}{g'_{k'}}\cdots{g'_1} = \tpb{u}{g_k}\cdots{g_1}\circ v$ for some $v$ and if
      $$v(\pb{{g_1}}{\tpb{u}{g_k}\cdots{g_2}}(\cdots(\pb{{g_k}}{u}(B)))) \subtype \pb{{g'_1}}{\tpb{u'}{g'_{k'}}\cdots{g'_2}}(\cdots(\pb{{g'_{k'}}}{u'}(B')))$$
      We write $\vec f=f_1;\dots f_n$, $\vec g=g_1;\cdots;g_{k}$, $\vec f'=f'_1;\dots f'_{n'}$ and $\vec g'=g'_1;\cdots;g'_{k'}$. We also write $e_f$ the isomorphisms between, the pullback of $\vec f$ and $u$, and the composition of the $n$ pullbacks of each $f_i$ and $u$. Similarly, we write $e_{f'}$, $e_g$ and $e_{g'}$ the other isos between pullback of the compositions and the the compositions of pullbacks. We also write $w=e_{g'};v;\inv{e_{g}}$\\
      We obtain the following diagram:
      \begin{center}
        \begin{tikzpicture}
          \node (B) at (0,2) {};
          \node (I) at (2,2) {};
          \node (K) at (4,1) {};
          \node (J) at (4,3) {};
          \node (I') at (6,2) {};
          \node (B') at (8,2) {};
          \node (pbg) at (1.5,4) {};
          \draw[-] (1.4,3.75) to (1.58,3.65);
          \draw[-] (1.7,3.85) to (1.58,3.65);
          \node (pbf) at (1.5,0) {};
          \draw[-] (1.4,.25) to (1.58,.35);
          \draw[-] (1.7,.15) to (1.58,.35);
          \node (pbg') at (6.5,4) {};
          \draw[-] (6.6,3.75) to (6.42,3.65);
          \draw[-] (6.3,3.85) to (6.42,3.65);
          \node (pbf') at (6.5,0) {};
          \draw[-] (6.6,.25) to (6.42,.35);
          \draw[-] (6.3,.15) to (6.42,.35);
\draw[->] (B) to node[auto] {$u$} (I);
          \draw[->] (B') to node[auto] {$u'$} (I');
          \draw[<-] (I) to node[auto] {$\vec f$} (K);
          \draw[<-] (I') to node[auto] {$\vec f'$} (K);
          \draw[<-] (I) to node[auto] {$\vec g$} (J);
          \draw[<-] (I') to node[auto] {$\vec g'$} (J);
\draw[->] (pbf) to node[auto] {$\pb{\vec f}u$} (B);
          \draw[->] (pbf) to node[auto] {$\pb u{\vec f}$} (K);
          \draw[->] (pbg) to node[above left] {$\pb{\vec g}u$} (B);
          \draw[->] (pbg) to node[auto] {$\pb u{\vec g}$} (J);
          \draw[->] (pbf') to node[below right] {$\pb{\vec f'}{u'}$} (B');
          \draw[->] (pbf') to node[auto] {$\pb {u'}{\vec f'}$} (K);
          \draw[->] (pbg') to node[auto] {$\pb{\vec g'}{u'}$} (B');
          \draw[->] (pbg') to node[auto] {$\pb {u'}{\vec g'}$} (J);
\draw[->] (K) to node[auto] {$h$} (J);
          \draw[->,dashed] (pbf') to node[above right] {$s$} (pbg');
\draw[->] (pbg') to[bend right =20] node[auto] {$w$} (pbg);
          \draw[->,dashed] (pbf') to[bend left =20] node[auto] {$w'$} (pbf);
        \end{tikzpicture}
      \end{center}
      Where the two middle triangles are the two hypothesis on $h$, the four symmetric rectangles are the pullbacks, and the above triangle is the hypothesis on $v$ (lifted to $w$ via the isomorphisms).\\
      The rightmost dashed function $s$ is the only function such that $s;\pb{\vec g'}{u'}=\pb{\vec f'}{u'}$ and $s;\pb{u'}{\vec g'}=\pb{u'}{\vec f'};h$.\\
      The bottom function $w'$ is the only function such that $w';\pb{u}{\vec f}=\pb{u'}{\vec f}$ and $w';\pb{\vec f}u=s;w;\pb{\vec g}u$.\\
      We now use $v'=\inv{e_{f'}};w';e_f$ and $h'=\inv{e_{f'}};s;e_{g'}$.
      Since $w';\pb{\vec f}u=s;w;\pb{\vec g}u$, which, lifted via the isomorphisms, gives
      $$v';\pb{{f_1}}{\tpb u{f_n}\cdots{f_2}};\dots;\pb{{f_n}}u=h';v;\pb{{g_1}}{\tpb u{g_k}\cdots{g_2}};\dots;\pb{{g_k}}u\ ,$$
      and since $\pb{\vec f'}{u'}=s;\pb{\vec g'}{u'}$, which, lifted via the isomorphisms, gives
      $$\pb{{f'_1}}{\tpb {u'}{f'_{n'}}\cdots{f'_2}};\dots;\pb{{f'_{n'}}}{u'}=h';\pb{{g'_1}}{\tpb {u'}{g'_{k'}}\cdots{g'_2}};\dots;\pb{{g'_{k'}}}{u'}\ ,$$
      we can apply our induction hypothesis and get:
      $$v'(\pb{{f_1}}{\tpb{u}{f_n}\cdots{f_2}}(\cdots(\pb{{f_n}}{u}(B)))) \subtype \pb{{f'_1}}{\tpb{u'}{f'_{n'}}\cdots{f'_2}}(\cdots(\pb{{anti-symmetryf'_{k'}}}{u'}(B')))$$
      Since $\pb{u'}{\vec f'}=w';\pb{u}{\vec f}$  so that $\tpb{u'}{f'_{n'}}\cdots{f'_1}=v';\tpb{u}{f_{n}}\cdots{f_1}$, we can apply the $!$-rule of the $\subtype$ system and conclude.
    \end{itemize}
  \end{proof}
}

\begin{lemma}
  \longVOnly{The relation} $(\subtype)$ is a preorder compatible with substitution:
  \begin{center}
  \AxiomC{\vphantom{A}}
  \dashedLine
  \UnaryInfC{$A\subtype A$}
  \DisplayProof\hskip 20pt
  \AxiomC{$A\subtype B$}
  \AxiomC{$B\subtype C$}
  \dashedLine
  \BinaryInfC{$A \subtype C$}
  \DisplayProof\hskip 20pt
  \AxiomC{$A\subtype A'$}
  \AxiomC{$B\subtype B'$}
  \dashedLine
  \BinaryInfC{$A[B/X] \subtype A'[B'/X]$}
  \DisplayProof
  \end{center}
\end{lemma}
\begin{proof}
  The reflexivity and substitution lemma can be proven by an immediate induction on $A$.

  \shortVOnly{The transitivity can be proven by induction on $\underline A$.}
  \longVOnly{
  The transitivity can be proven coinductively by destroying $\underline A$. As before, the multiplicative are immediate application of IH and the atoms are trivial. By duality, we are only showing one side of each operator:
  \begin{itemize}
  \item If $A=A_1\withm {i^a}{j^a} A_2$, then, necessarily, $B=B_1\withm {i^b}{j^b} B_2$ and $C=C_1\withm {i^c}{j^c} C_2$ with $\pb{i^b}{i^a}(A_1)\subtype \pb{i^a}{i^b}(B_1)$, $\pb{i^c}{i^b}(B_1)\subtype \pb{i^b}{i^c}(C_1)$, $\pb{j^b}{j^a}(A_2)\subtype \pb{j^a}{j^b}(B_2)$ and $\pb{j^b}{j^b}(B_2)\subtype \pb{j^b}{j^c}(C_2)$.\\
    By orthogonality, each of the $\pb{i^\_}{i ^\_}$ and $\pb{j^\_}{j^\_}$ are bijections. Thus, by Lemma~\ref{lemma:order_base_change}, $\pb{i^c}{i^a}(A_1)\subtype (\pb{i^c}{i^a};\inv{{\pb{i^b}{i^a}}};\pb{i^a}{i^b})(B_1)$, and $(\pb{i^a}{i^c};\inv{{\pb{i^b}{i^c}}};\pb{i^c}{i^b})(B_1)\subtype \pb{i^a}{i^c}(C_1)$. Since $\pb{i^c}{i^a};\inv{{\pb{i^b}{i^a}}};\pb{i^a}{i^b}=\pb{i^a}{i^c};\inv{{\pb{i^b}{i^c}}};\pb{i^c}{i^b}$, we can apply our induction hypothesis which gives that $\pb{i^c}{i^a}(A_1)\subtype\pb{i^a}{i^c}(C_1)$. Symmetrically, we can obtain $\pb{j^c}{j^a}(A_2)\subtype\pb{j^a}{j^c}(C_2)$, which is sufficient to conclude.
  \item If $A=\bangm{u}A'$, then, necessarily, $B=\bangm{u\circ g}B'$ and $C=\bangm{u\circ g\circ h}C'$ with $\indfun{g}{A}\subtype B$ and $\indfun{h}{B}\subtype C$. By Lemma~\ref{lemma:order_base_change}, we have $(g\circ h)(A)\subtype h(B)$ and by applying IH, $(g\circ h)(A)\subtype C$, which concludes.
  \item If $A=\mum fX.A'$, then, necessarily, $B=\mum{g}Y.B'$ and $C=\mum{h}Z.C'$ with $\indfun{f}{A'[\mu X.A'/X]} \subtype \indfun{g}{B'[\mu Y.B'/Y]} \subtype \indfun{h}{C'[\mu Z.C'/Z]}$, we can apply the productive fixpoint rule to prove $\mum fX.A'\subtype \mum hZ.C'$ and generate a proof of $\indfun{f}{A'[\mu X.A'/X]} \subtype \indfun{h}{C'[\mu Z.C'/Z]}$
\end{itemize}
  }
\end{proof}

Since we have a preorder, there is an equivalence $(\eqtype) = (\suptype)\cap (\subtype)$, which is also a congruence in the sense that is compatible with type substitution.
  \longVOnly{
    \begin{center}
        \AxiomC{$A\subtype A'$}
        \AxiomC{$A'\subtype A$}
        \BinaryInfC{$A \eqtype A'$}
        \DisplayProof\\
    \end{center}

By abuse of notation, 
\AxiomC{}
\UnaryInfC{$f_1;\dots;f_n =f'_1;\dots;f'_{n'}$}
\DisplayProof
will denote
\[
\AxiomC{$f_1;\dots;f_n =h$}
\AxiomC{$f'_1;\dots;f'_{n'}=h$}
\AxiomC{}
\dashedLine
\UnaryInfC{$A\subtype A$}
\dashedLine
\TrinaryInfC{$\indfun{f_1}{\cdots\indfun{f_n}{A}\cdots}\subtype \indfun{f'_1}{\cdots\indfun{f'_{n'}}{A}\cdots}$}
\DisplayProof
.\]

\begin{lemma}
  The fixpoint rules have a non-complete but inductive version, corresponding to Park':
    \begin{center}
      \AxiomC{$A[\varm{g}{Y}/X]\subtype \indfun{g}{A'}$}
      \AxiomC{$Y\not\in \mathtt{FV}(A)$}
      \dashedLine
      \BinaryInfC{$\mum{f}X.A \subtype \mum{g\circ f}Y.A'$}
      \DisplayProof\hskip 20pt
      \AxiomC{$\indfun{g}{A}\subtype A'[\varm{g}{X}/Y]$}
      \AxiomC{$Y\not\in \mathtt{FV}(A)$}
      \dashedLine
      \BinaryInfC{$\num{g\circ f}X.A \subtype \num fY.A'$}
      \DisplayProof
    \end{center}
\end{lemma}
\begin{proof}
  We generate the proof coinductively (above a productive rule) by applying the lemma itself:\\
  \AxiomC{$A[\varm{g}{Y}/X]\subtype \indfun{g}{A'}$}
  \dashedLine
  \UnaryInfC{$\mu X.A\subtype \mum{g}Y.A'$}
  \dashedLine
  \UnaryInfC{$A[\mu X.A/X]\subtype A[\mum{g}Y.A'/X]$}
  \AxiomC{$A[\varm{g}{Y}/X]\subtype \indfun{g}{A'}$}
  \dashedLine
  \UnaryInfC{$ A[\mum{g}Y.A'/X]\subtype \indfun{g}{A'}[\mu Y.A'/Y]$}
  \dashedLine
  \BinaryInfC{$A[\mu X.A/X]\subtype \indfun{g}{A'}[\mu Y.A'/Y]$}
  \dashedLine
  \UnaryInfC{$\indfun{f}{A[\mu X.A/X]} \subtype \indfun{(g{\circ}f)}{A'[\mu Y.A'/X]}$}
  \productiveLine
  \UnaryInfC{$\mum{f}X.A \subtype \mum{g\circ f}Y.A'$}
  \DisplayProof
\end{proof}
}

\subsection{Embedding of idempotent intersection types}

In this section, we will discuss the formal embedding of idempotent intersection types as indexed formulae modulo $(\eqtype)$.

Before proceeding to formally establishing the result, we need to verify that an indexed formula over $I$ can really be seen as a $I$-indexed collection of singleton formulae $(\indfun{\mathtt{cst}_x}{A})_{x\in I}$, where, for $x \in I $ ,$ \mathtt{cst}_x : \1 \to I $ is the function $ \ast \mapsto x $.

\begin{lemma}\label{decomposition}
 Given $A$ and $B$ defined over the same locus $I$ we have:
  \[
  \AxiomC{$\forall x\in I, \indfun{\mathtt{cst}_x}{A}\subtype\indfun{\mathtt{cst}_x}{B}$}
  \dashedLine
  \UnaryInfC{$A\subtype B$}
  \DisplayProof
  \]
\end{lemma}
\begin{proof}
  By induction of $A$, with only \longVOnly{two}\shortVOnly{one} difficult case\longVOnly{s}:\\
  $\bullet$\quad If $\forall x\in I, \indfun{\mathtt{cst}_x}{\bangm u A}\subtype\indfun{\mathtt{cst}_x}{\bangm v B}$.
    Then $\pb{v}{\mathtt{cst}_x}=\pb{u}{\mathtt{cst}_x}\circ g_x$ with $\indfun{g_x}{\indfun{\pb{{\mathtt{cst}_x}}{u}}{A}}\subtype \indfun{\pb{{\mathtt{cst}_x}}{v}}{B}$.
    The pullback of $v$ and $\mathtt{cst}_x$ is isomorphic to $u^{-1}(x)$, similarly for $v$, by pre and post composing $g_x$ with these isos, we obtain a function $g'_x:v^{-1}(x)\rightarrow u^{-1}(x)$ with $\indfun{g'_x}{\indfun{i_x}{A}}\subtype \indfun{j_x}{B}$ where $i_x:u^{-1}(x)\hookrightarrow \mathtt{dom}(u)$ and $j_x:v^{-1}(x)\hookrightarrow \mathtt{dom}(v)$ are the inclusions.\\
    For any $y\in v^{-1}(x)$ we get $\indfun{\mathtt{cst}_y}{\indfun{g'_x}{\indfun{i_x}{A}}}\subtype \indfun{\mathtt{cst}_y}{\indfun{j_x}{B}}$, thus $\indfun{\mathtt{cst}_{i_x(g'_x(y))}}{A}\subtype \indfun{\mathtt{cst}_{j_x(y)}}{B}$.\\
    Let $g'=\coprod_x g'_x : \mathtt{dom}(v){\rightarrow} \mathtt{dom}(u)$, so that ${i_x(g'(y)) = g'_x(y)}$. Then we have $\indfun{\mathtt{cst}_{y}}{\indfun{g'}{A}}\subtype \indfun{\mathtt{cst}_{y}}{B}$ for any $y\in \mathtt{dom}(v)$. By IH, we get $\indfun{g'}{A}\subtype B$ which concludes.
  \longVOnly{\\[.2em]
  $\bullet$\quad If $\forall x\in I, \indfun{\mathtt{cst}_x}{\mum f X. A}\subtype\indfun{\mathtt{cst}_x}{\mum g Y. B}$.\\
    Then $\mum {\mathtt{cst}_{f(x)}} X. A\subtype \mum {\mathtt{cst}_{g(x)}} Y. B$, thus there is $h_x$ such that
    $\mathtt{cst}_{g(x)}=\mathtt{cst}_{f(x)}\circ h_x$, but $h_x:\1\rightarrow\1$ can only be the identity, thus $g(x)=h(x)$ for all $x$, and $g=h$ which concludes by reflexivity of $(\subtype)$.
    }
\end{proof}

\begin{lemma}\label{formulaeintersection}
  Let $I$ be a set and let $\left[\J_x\vdash A_x\ \mathtt{def}\right]_{x\in I}$ be an indexed set such that $\underline A_x=\underline A_y$ for each $x,y\in I$.\\
  Then\longVOnly{, if the $A_x$ are fixedpoint-free,} there is $\coprod_{x\in I}J_x\vdash \bigwedge_xA_x\ \mathtt{def}$ such that $\indfun{\mathtt{inj}_{J_x}}{\bigwedge_xA_x}\eqtype A_x$.
\end{lemma}
\begin{proof}
  It is defined by:
  \begin{itemize}
  \item $\bigwedge_x\varm{f_x}{Y} := \varm{f}{Y}$ where $f((x,z)):=f_x(z)$,
  \item $\bigwedge_x(A_x\otimes B_x):=\bigwedge_xA_x\otimes\bigwedge_xB_x$,
  \item $\bigwedge_x(A_x\oplusm{i_x}{j_x} B_x) := \bigwedge_{x\in I_1}(A_x)\oplusm{\coprod_x i_x}{\coprod_x j_x}\bigwedge_{x\in I_2}(B_x)$ where $I_1=\mathtt{dom}(\coprod_x i_x)$ and $I_2=\mathtt{dom}(\coprod_x i_x)$, 
  \item $\bigwedge_x(\bangm{u_x}A_x):=\bangm{\coprod_x u_x}\bigwedge_{x}A_x$,
  \item $\bigwedge_x(A_x^\bot):=\left(\bigwedge_xA_x\right)^\bot$.\vspace{-1.5em}
  \end{itemize}
\end{proof}

\begin{definition}\label{def:emb_it}
  Let $\mathcal{A}$ the indexed types of the grammar
  $$A,B::= \varm f X \mid (\bangm u A)^\bot\parr B  \mid A\withm ij B$$
  defined over the locus $\{*\}$. For any $f:\1\rightarrow I$ and $X\in\mathtt{var}(I)$, we first fix a variable of intersection type $\alpha_{\varm f X}$.\\
  We define the interpretation of $\mathcal{A}$ into intersection types by:
  \begin{center}
    $\llbracket \varm f X \rrbracket := \alpha_{\varm f X}\quad\llbracket (\bangm u A)^\bot\parr B \rrbracket := \left[\llbracket\indfun{\mathtt{cst}_x}{A}\rrbracket\right]_{x\in\mathtt{dom}(u)}\rightarrow \llbracket B \rrbracket$\\[.5em]
    $\llbracket A\withm{\id}{\texttt{init}} \rrbracket := \llbracket A\rrbracket{\with}\bullet\quad\quad \llbracket A\withm{\texttt{init}}{\id} \rrbracket := \bullet{\with}\llbracket B\rrbracket$
    \longVOnly{
      \\[.5em]
      $\llbracket A\oplusm{\id}{\texttt{init}} \rrbracket := \llbracket A\rrbracket{\oplus}\bullet\quad\quad \llbracket A\oplusm{\texttt{init}}{\id} \rrbracket := \bullet{\oplus}\llbracket B\rrbracket$
      }
  \end{center}
\end{definition}

\begin{lemma}
  $\llbracket A\rrbracket \le \llbracket B\rrbracket$ iff $A\subtype B$. 
\end{lemma}
\begin{proof}
  By induction on $A$ using Lemma~\ref{decomposition}
  \longVOnly{
  \begin{itemize}
  \item If $\alpha_{\varm f X}=\llbracket \varm f X \rrbracket \le \llbracket B \rrbracket$ then $\llbracket B \rrbracket=\alpha_{\varm f X}$ which means that $B=\varm f X$.
  \item If $\llbracket (\bangm u A_1)^\bot\parr A_2 \rrbracket= \left(\bigwedge_{x\in\mathtt{dom}(u)}\llbracket\indfun{\mathtt{cst}_x}{A_1}\rrbracket\right)\rightarrow \llbracket A_2 \rrbracket \le \llbracket B\rrbracket$
    then, $B$ is also an arrow of the form $B=(\bangm v B_1)^\bot\parr B_2$, thus
    $$ \left(\bigwedge_{x\in\mathtt{dom}(u)}\llbracket\indfun{\mathtt{cst}_x}{A_1}\rrbracket\right)\rightarrow \llbracket A_2\rrbracket \le  \left(\bigwedge_{y\in\mathtt{dom}(v)}\llbracket\indfun{\mathtt{cst}_y}{B_1}\rrbracket\right)\rightarrow \llbracket B_2\rrbracket$$
    which means that :
    $$ \left(\forall x\in\mathtt{dom}(u),\exists y\in\mathtt{dom}(v),\ \llbracket\indfun{\mathtt{cst}_x}{A_1}\rrbracket\ge \llbracket\indfun{\mathtt{cst}_y}{B_1}\rrbracket\right) \quad\text{and}\quad \llbracket A_2\rrbracket \le \llbracket B_2\rrbracket$$
    By applying IH, we get:
    $$ \bigl(\forall x\in\mathtt{dom}(u),\exists y\in\mathtt{dom}(v),\ \indfun{\mathtt{cst}_x}{A_1}\suptype \indfun{\mathtt{cst}_y}{B_1}\bigr) \quad\text{and}\quad A_2\subtype B_2$$
      Since $\mathtt{dom}(u)$ is finite, we can Skolemize and get $g:\mathtt{dom}(u)\rightarrow\mathtt{dom}(v)$ such that:
      $$ \bigl(\forall x\in\mathtt{dom}(u),\ \indfun{\mathtt{cst}_x}{A_1}\suptype \indfun{\mathtt{cst}_{g(x)}}{B_1}\bigr) \quad\text{and}\quad A_2\subtype B_2$$
      By Lemma~\ref{decomposition}, we get
      $$ A_1\suptype \indfun{g}{B_1} \quad\text{and}\quad A_2\subtype B_2$$
      Which concludes since $u=v\circ g$
    \item If $(\bangm u A_1)^\bot\parr A_2\subtype B$ then $\underline B = (!\underline A_1)^\bot\parr\underline A_2$, thus $B= (\bangm v B_1)^\bot\parr B_2\subtype B$, and, necessarily
      $$ \bangm u A_1 \suptype \bangm v B_1 \quad\quad\texttt{and}\quad\quad A_2\subtype B_2$$
      thus $u=v\circ g$ such that $ A_1 \suptype \indfun{g}{B_1}$ and for any $x\in\texttt{dom}{u}$, 
      $$\indfun{\mathtt{cst}_x}{A_1}\suptype\indfun{\mathtt{cst}_x}{\indfun{g}{B_1}}\eqtype \indfun{\mathtt{cst}_{g(x)}}{B_1}$$
      Which shows that
      $$ \bigl(\forall x\in\mathtt{dom}(u),\exists y\in\mathtt{dom}(v),\ \indfun{\mathtt{cst}_x}{A_1}\suptype \indfun{\mathtt{cst}_y}{B_1}\bigr) \quad\text{and}\quad A_2\subtype B_2$$
  \end{itemize}
  }
\end{proof}

\begin{theorem}\label{th:ITformulae}
  \shortVOnly{$\quad\mathcal{A}/\eqtype \ =\ \mathcal{IT}_{\mathtt{simple}}/\simeq$}
  \longVOnly{$$\mathcal{A}/\eqtype \quad =\quad \mathcal{IT}_{\mathtt{simple}}/\simeq$$}
\end{theorem}
\begin{proof}
  It only remains to show that $\llbracket-\rrbracket$ is surjective, which corresponds exactly to the fact that the intersection types are refining a simple type, i.e., if $a:\tau$, then there is an indexed formula $a^*$ over $\1$ such that $\underline a^*=\tau$ and $\llbracket a^*\rrbracket = a$. For this we use the intersection:\\[.2em]
  \hspace*{2cm}$ ([a_1,...,a_n]\rightarrow b)^* = \wnm{i\mapsto *}\bigwedge_i a_i^{*\bot}\parr b^* $
\end{proof}

This shows that, at the level of formulae, intersection types under a simple types corresponds to indexed formulae over the same type, up-to the equivalence induced by subtyping. At the level of proofs and terms, we will see that the corespondance is preserved.

\longVOnly{
\begin{proposition}
  Let $\bar{\mathbb \N} := \{*\}\uplus \bar{\N}$ the set of completed natural number seen as a coproduct. We denote by $\texttt{int} := \mum{\id_{\bar{\mathbb N}}} (\1\oplusm{\iota_1}{\iota_2} X)$ the indexed integers.
  
  Then, any formula $A$ provable in Ind$_\wedge\mu$LL such that $\underline A= \mu X.(\1\oplus X)$ is of the form $A=f(\texttt{int})$, more precisely :
  $$\vdash \mum fX.(\1\oplusm ij \varm{g}{X}) \quad \implies\quad \exists h,\quad \mum fX.(\1\oplusm ij \varm{g}{X}) \ \eqtype \ h(\texttt{int})$$

  In particular, up-to equivalence, the formulae above $\mu X.\1\oplus X$ over the locus~$\{*\}$ are exactly completed natural numbers.
\end{proposition}
\begin{proof}
  For all $n$, we define, by recursion $K_n,K_{\ge n}\subseteq K$, $f_n:K_n\rightarrow L$ where $L$ is the codomain of $f$ :
  \begin{itemize}
  \item $K_{\ge 0}:= \mathtt{dom}(f)$ and $f_0:=f$, 
  \item $K_n:=\{k\in K_{\ge n}\mid \exists x, i(x)=f_n(k)\}$,
  \item $K_{\ge n}:=\{k\in K_{\ge n} \mid \exists y, j(y)=f_n(k)\}$,
  \item $f_{n+1}(k):=g(j^{-1}(f_n(k)))$
  \end{itemize}
  Since $i$ and $j$ is an orthogonal pair, there are bijections $K_{\ge n}\simeq K_n\uplus K_{ge n{+}1}$.
  We write $i_n:K_n\rightarrow K_{\ge_n}$ and $j_n:K_{\ge n{+}1}\rightarrow K_{\ge_n}$ the corresponding injections.\\
  Let $K_\infty:= \bigcap_n{K_{\ge n}}$ so that $K_{\ge n}=K_\infty\uplus\biguplus_{n'\ge n}K_{n'}$.\\
  Let $h_n :K_{\ge n}\rightarrow n$ be defined uniquely by $k\in K_{h_n(k)}$.\\
  Up-to-iso we have the following diagrams :
  \begin{center}
    \begin{tikzpicture}
      \node (N) at (-1,1) {$\bar{\mathbb N}$};
      \node (N2) at (-1,2) {$\bar{\mathbb N}$};
      \node (K) at (1,1) {$K_{\ge n}$};
      \node (K2) at (1,2) {$K_{\ge n{+}1}$};
      \node (L) at (3,1) {$L$};
      \node (J) at (3,2) {$J$};
      \node (L2) at (5,2) {$L$};
      \node (N3) at (-1,0) {$\1$};
      \node (K3) at (1,0) {$K_{n}$};
      \node (I) at (3,0) {$I$};
\draw[-] (.7,1.9) to (.7,1.7) to (.9,1.7);
      \draw[-] (1.3,1.9) to (1.3,1.7) to (1.1,1.7);
      \draw[-] (.7,.1) to (.7,.3) to (.9,.3);
      \draw[-] (1.3,.1) to (1.3,.3) to (1.1,.3);
\draw[->] (K) to node[auto] {$h_n$} (N);
      \draw[->] (K2) to node[above] {$h_{n{+}1}$} (N2);
      \draw[->] (N2) to node[auto] {$\iota_2$} (N);
      \draw[->] (K2) to node[auto] {$j_n$} (K);
      \draw[->] (K) to node[auto] {$f_n$} (L);
      \draw[->] (K2) to (J);
      \draw[->] (J) to node[auto] {$j$} (L);
      \draw[->] (J) to node[below] {$g$} (L2);
      \draw[->] (K2) to[bend left =20] node[auto] {$f_{n{+}1}$} (L2.north west);
      \draw[->] (K3) to (N3);
      \draw[->] (N3) to node[auto] {$\iota_1$} (N);
      \draw[->] (K3) to node[auto] {$i_n$} (K);
      \draw[->] (K3) to (I);
      \draw[->] (I) to node[auto] {$i$} (L);
    \end{tikzpicture}
  \end{center}
  For simplicity, we write the equivalence proof up-to-iso, we should add isomorphic base-changes each time we implicitly apply an iso :
  \begin{center}
    \AxiomC{}
    \UnaryInfC{$\1\eqtype\1$}
    \AxiomC{coinduction}
    \dashedLine
    \UnaryInfC{$\mum {f_{n{+}1}} X.(\1\oplusm ij \varm{g}{X})\eqtype \indfun{h_{n{+}1}}{\texttt{int}}$}
    \BinaryInfC{$\1\oplusm {i_n}{j_n} \mum {f_{n{+}1}} X.(\1\oplusm ij \varm{g}{X}) \eqtype \1\oplusm{i_n}{j_n} \indfun{h_{n{+}1}}{\texttt{int}}$}
    \productiveLine
    \UnaryInfC{$\mum {f_n}X.(\1\oplusm ij \varm{g}{X}) \eqtype  \indfun{h_n}{\texttt{int}}$}
    \DisplayProof
  \end{center}
\end{proof}

\begin{proposition}
  Let $\texttt{lam}:=\mu X.\nu Y.(?X^\bot\parr Y)$ the LL encoding of CbN untyped lambda terms.

  The type $A:=\mum {\id_{\1}} X. \num {\id_{\1}} Y.(\wnm{\id_\1}X^\bot\parr Y)$ is an indexed term above $\texttt{lam}$ over the locus $\{*\}$ which have proofs refining any proofs of $\texttt{lam}$. An other way to state it is that $A$ can ``type'' any lambda-term, even diverging ones.
\end{proposition}

}

\longVOnly{
It is possible to add productive fixpoints over additives to represent ADTs, preserving this theorem. Howeover, while adding non-productive type fixpoints allow the encoding of untyped lambda terms as proofs of $\mu X.\nu Y.(?X^\bot\parr Y)$, their indexed formulae counterpart are different.
 This paradox can be understood through the light of the corresponding denotational model. Intersections types fir untyped $\lambda$-calculus correspond to the so-called ``Engler'', or ``free'' filter models, but their are many such models among which, some are not sensible. Indexed linear logic will have to ``choose'' among these models an it does so in a canonical way. Unfortunately, we obtain the largest one in place of the smallest one. As future work, we intend to look at the $\OIndLL$ version, which is based on the relational model. In the relational world, every intersection type model is sensible, which means that the $\OIndLL$ interpretation of untyped $\lambda$-calculus should characterize productivity under some condition (probably the absence of the $\emptyset$ as locus all along derivations).}

\section{The proof system}\label{sec:sys}

\subsection{Derivation rules}

Sequents of $\IndLL$ are sequences of formulae defined over the same locus.  The logic we are describing is classical, with exchange laws, and involutive negations. In order to simplify the system we will consider sequents up-to usual symmetries; other authors tend to fix a notation with every formulae on the right of the sequent. However, we prefer to present the system in a way that (i) is more accessible to people unused to classical LL, (ii) will make the proofs easier to manipulate, and (iii) makes explicit the fact that rules (such as the cut) are implicitly pointing a specific formula. We are aware that this quotient is highly non-trivial and is broken for some logic, but the difficulties are tied with the negation, which, in our system, is exactly the negation of linear logic; by adopting this notation, we are dismissing a technical burden that is orthogonal to the point we are focusing on.

\begin{definition}[Sequents]
  Sequents are of the form $\Gamma \indvdash I \Delta$ for $\Gamma$ and $\Delta$ sequences of formulae $A$ such that $I\vdash A\ \mathtt{def}$.\\
  They are considered up-to structural exchanges and symmetries:
  \begin{align*}
    \Gamma_1,\Gamma_2 \indvdash I \Delta_1,A,\Delta_2\quad &=\quad  \Gamma_1,A^\bot,\Gamma_2 \indvdash I \Delta_1,\Delta_2 \\
    \Gamma \indvdash I \Delta \quad&=\quad   \Gamma,\1 \indvdash I \Delta
  \end{align*}
\end{definition}

Derivation rules follow those of linear logic, with two major differences: 1) definition of formulae over the same locus must be preserved and 2) base changes are performed in the hypotheses.

The second is quite unusual, it means, in particular, that the sub-formula property is not strictly maintained. Don't worry though, even if it does not hold, one can consider $\indfun{f}{A}$ as a refinement of $A$, and thus the sub-formula property holds up-to refinement.

The base change express the way the information is propagated along the derivation tree in a non-even way, it means that copies of the same sub-formula can behave very differently in different places of the proof because a different base change operation could be applied.

The fact that base changes appear only at the level of the hypothesis leads to a nice generalisation of the base change operation to whole proofs.\longVOnly{\footnote{Disclamer: in a previous version could apear the equivalent folowing introduction of the $\oplus$
  \AxiomC{$\indfun{i}\Gamma\indvdash {J} A$}
  \AxiomC{\hspace*{-1.2em}\optional{$\emptyset\vdash B\ \mathtt{def}$}}
  \BinaryInfC{$\Gamma\indvdash I A\oplusm{i}{\init} B$}
  \DisplayProof
  Which is correct, provided that everything follow this choice, but we were wrongly mixing both version.
}}

\begin{definition}[Proofs]
  Provable sequents are defined \shortVOnly{inductively}\longVOnly{co-inductively (i.e. can have infinite proofs), but infinite branches must have infinitely many $\nu^\infty$-introduction (or equivalently, $\mu^\infty$-elimination)}:
  \longVOnly{\renewcommand\skiplength{4em}}
  \begin{center}
    \AxiomC{$\Gamma\indvdash I A$}
    \AxiomC{$A\indvdash I \Delta$}
    \BinaryInfC{$\Gamma\indvdash I \Delta$}
    \DisplayProof\hskip \skiplength
    \AxiomC{\optional{$I\vdash f(X)\ \mathtt{def}$}}
    \UnaryInfC{${\varm{f}{X}}\indvdash I {\varm{f}{X}}$}
    \DisplayProof\hskip \skiplength
    \AxiomC{$\Gamma\indvdash I A$}
    \AxiomC{$\Delta\indvdash I B$}
    \BinaryInfC{$\Gamma,\Delta\indvdash I A\otimes B$}
    \DisplayProof\vspace{0.7em}\\
    \AxiomC{\vphantom{$A\indvdash I$}}
    \UnaryInfC{$\indvdash I \1$}
    \DisplayProof\hskip \skiplength
    \AxiomC{$\Gamma\indvdash I A,B$}
    \UnaryInfC{$\Gamma\indvdash I A\parr B$}
    \DisplayProof\hskip \skiplength
    \AxiomC{\optional{$\emptyset\vdash\Gamma\ \mathtt{def}$}\vphantom{$\indfun{i}{\Gamma}\indvdash I$}}
    \UnaryInfC{$\Gamma\indvdash \emptyset \top$}
    \DisplayProof\hskip \skiplength
    \AxiomC{$\indfun{i}{\Gamma}\indvdash I A$}
    \AxiomC{$\hspace*{-1em}\indfun{j}{\Gamma}\indvdash J B$}
    \BinaryInfC{$\Gamma\indvdash {K} A\withm{i}{j} B$}
    \DisplayProof\vspace{0.7em}\\
    \AxiomC{$\Gamma\indvdash {I} \indfun{\inv i}A$}
    \AxiomC{\hspace*{-1.2em}\optional{$\emptyset\vdash B\ \mathtt{def}$}}
    \BinaryInfC{$\Gamma\indvdash I A\oplusm{i}{\init} B$}
    \DisplayProof\hskip \skiplength
    \AxiomC{$\Gamma\indvdash {J} \indfun{\inv i}B$}
    \AxiomC{\hspace*{-1.2em}\optional{$\emptyset\vdash A\ \mathtt{def}$}}
    \BinaryInfC{$\Gamma\indvdash I A\oplusm{\init}{i} B$}
    \DisplayProof\vspace{0.7em}\\
    \longVOnly{
      \AxiomC{$\Gamma\indvdash I \indfun{f}{A[\mum{\id}X.A/X]}$}
      \UnaryInfC{$\Gamma\indvdash I\mum fX.A$}
      \DisplayProof\hspace{1cm}
      \AxiomC{$\Gamma\indvdash I \indfun{f}{A[\num{\id}X.A/X]}$}
      \productiveLine
      \UnaryInfC{$\Gamma\indvdash I\num fX.A$}
      \DisplayProof\vspace{0.7em}\\
      }
    \AxiomC{$\Gamma\indvdash I \wnm{w}A,\wnm{w}A$}
\UnaryInfC{$\Gamma\indvdash I \wnm{w}A$}
    \DisplayProof\hskip \skiplength
    \AxiomC{$\Gamma\indvdash I \indfun{f}{B}$}
    \AxiomC{$u\circ f = \id$}
    \BinaryInfC{$\Gamma\indvdash I \wnm{u}B$}
    \DisplayProof\vspace{.7em}\\
    \AxiomC{$\Gamma\indvdash I $}
    \UnaryInfC{$\Gamma\indvdash I \wnm{u}B$}
    \DisplayProof\hskip \skiplength
    \AxiomC{$\indfun{v}{\bangm{w_1}A_1},\dots,\indfun{v}{\bangm{w_n}A_n}\indvdash J B$}
\UnaryInfC{$\bangm{w_1}A_1,\dots,\bangm{w_n}A_n\indvdash I \bangm vB$}
    \DisplayProof
  \end{center} 
\end{definition}

  The \nonColorblind{gray}\colorblind{yellow} highlighting  is there to emphasise the fact those hypotheses are redundant since only well-defined formulae are considered, we only add them as reminder.
  
  In contrast, the second hypothesis of a dereliction, $u\circ f = \id$, and more precisely, the choice of the function $f$, are relevant for the cut elimination procedure. It is actually the only choice that is not syntax-directed along those rule apart from the cut. This is directly related with the subtyping that also arises in the dereliction only.

It is worth noting that the multiplicative rules are linear on their hypothesis but can duplicate their locus,  additives can duplicate their hypothesis but are linear on their locus and exponentials can duplicate both.

Let us discuss the $\with$-introduction, $\oplus$-introduction, dereliction and promotion in detail. In the $\with$-introduction, we want to prove a formula $A\withm ij B$ under the locus $K$ and the hypothesis $\Gamma$. This means that we have a $K$-indexed family of intersection types of the shape $\Gamma_k\rightarrow (a_k\with\bullet)$ or $\Gamma_k\rightarrow(\bullet\with b_k)$, depending weather $k\in K$ is in the image of $i$ of $j$. More precisely, we have $I$-indexed intersection types of the form $\Gamma_{i(x)}\rightarrow (a_x\with\bullet)$ for $x\in I$ and $J$-indexed intersection types of the form $\Gamma_{j(y)}\rightarrow (\bullet\with b_y)$ for $y\in J$. In the first case, we have to show that $\Gamma_{i(x)}\rightarrow a_x$ and in the second that $\Gamma_{j(y)}\rightarrow b_y$. Once we state that $\Gamma_{i(x)}$ is the $x$-indexed intersection type of $\indfun{i}{\Gamma}$, the rule makes complete sense.

For the $\oplus$ introduction, first notice that $A\oplusm{i}{\init}B$ is only correct if $i$ is a bijection. In fact, we should intuitively consider that $i$ is exactly the identity, as other bijections arise only for generality and to get a nice, uniform presentation. With $i=\id$ and up-to some equivalence, the rule would become
\AxiomC{$\Gamma\indvdash {I} A$}
\dashedLine
\UnaryInfC{$\Gamma\indvdash I A\oplusm{\id}{\init} B$}
\DisplayProof
which is the expected rule. Notice, though, that the $\oplus$-introduction use a $i$ base change and the $\with$-introduction (which is the $\oplus$-elimination) use a~$\inv{i}$ base change, one inverse to the other.

For the dereliction, also called (co)unit, let us first consider the case where $I=\1$ is a singleton. In this case, $u:J\rightarrow \1$ is a set and $f:\1\rightarrow J$ is selecting an element of $u$. Hence, $\indfun{f}{B}$ is the selection of the $f$'s element of the bag $\bangm{u}B$, which is the dereliction of intersection types. The current version is the same thing but applied in parallel over the $I$ indexed set of bags.

If also for the promotion rule we restrict $I$ to be a singleton, one can see $!_v B$ as a $J$-indexed set of intersection types refining $B$. Hence, we can understand the sequent $\indfun{v}{\bangm{w_1}A_1},\dots,\indfun{v}{\bangm{w_n}A_n}\indvdash J B$ as a $J$-indexed family of intersection types, one for each element of $!_vB$, with the duplication of the exact same context \textit{via} base change.

Note that, contrary to the additive rules, dereliction and the promotion are not complete inverse one to the other, in the sense that the base change for the dereliction use only a right inverse of the base-change for the promotion, this exposes a loss of information and an implicit use of subtyping.

\subsection{Subtyping and base change in proofs}

The subtyping, at the level of proofs, can be restricted to the sole dereliction rule. Mirroring what happens in intersection type systems, nonetheless, the full subtyping rule remains derivable. This is not surprising once recalling that a subtyping sequent is nothing else than refinements of an identity proof.

Notice that this derivation of subtyping will be essential for to the cut elimination. Indeed, implicit subtyping will have to be propagated along cut elimination, together with some appropriate base changes. In order to the cut-elimination procedure locally deterministic, the subtyping derivation will be explicitly applied.

\begin{lemma}
  Given a subtyping derivation $ \rho$ of conclusion $ A \subtype B $ over the locus $ I$ we can build an $\IndLL$ proof $ \pi_{\rho} $ of conclusion $ A \vdash_I B. $
  We use the following notation to refer to $\pi_\rho$:
  $$\AxiomC{$\rho$}\UnaryInfC{$A \subtype B$}\subtypLine\UnaryInfC{$A \vdash_I B$}\DisplayProof:=\pi_\rho$$
\end{lemma}
\begin{proof}
  We precede by \shortVOnly{induction, presenting only the most interesting cases.} 
  \longVOnly{
    co-induction, bypassing half of the cases by duality:
  }
 \shortVOnly{\\[.2em]
  \[ \rho =  \bottomAlignProof
    \AxiomC{$i\Bot j'\quad j\Bot i'$}
    \AxiomC{$\indfun{\pb{i'}{i}}{A}\subtype \indfun{\pb{i}{i'}}{A'}$}
    \AxiomC{$\indfun{\pb{j'}{j}}{B}\subtype \indfun{\pb{j}{j'}}{B'}$}
    \TrinaryInfC{$A\withm{i}{j}B\subtype A'\withm{i'}{j'} B'$}
    \DisplayProof \]
  
    \[\pi_{\rho} \raiseRel{:=}\ 
    \bottomAlignProof
    \AxiomC{$\indfun{\pb{i'}{i}}{A}\subtype \indfun{\pb{i}{i'}}{A'}$}
    \UnaryInfC{$\indfun{\inv{\pb{i}{i'}}}{\indfun{\pb{i'}{i}}{A}}\subtype A'$}
    \subtypLine
    \UnaryInfC{$\indfun{\inv{\pb{i}{i'}}}{\indfun{\pb{i'}{i}}{A}}\indvdash I A'$}
\UnaryInfC{$\indfun{\pb{i'}{i}}{A}\withm{\pb{i}{i'}}{\init} \indfun{\init}{B}\indvdash I A'$}
    \AxiomC{$\indfun{\pb{j'}{j}}{B}\subtype \indfun{\pb{j}{j'}}{B'}$}
    \dashedLine
    \UnaryInfC{$\indfun{\inv{\pb{j}{j'}}}{\indfun{\pb{j'}{j}}{B}}\subtype B'$}
    \subtypLine
    \UnaryInfC{$\indfun{\inv{\pb{j}{j'}}}{\indfun{\pb{j'}{j}}{B}}\indvdash I B'$}
\UnaryInfC{$\indfun{\init}{A}\withm{\init}{\pb{j}{j'}}\indfun{\pb{j'}{j}}{B}\indvdash I B'$}
    \BinaryInfC{$A\withm{i}{j}B\indvdash {I} A'\withm{i'}{j'} B'$}
    \DisplayProof \]
  }
  \longVOnly{
     \begin{align*}
       \bottomAlignProof
       \AxiomC{}
       \UnaryInfC{$\varm{f}{X}\subtype \varm{f}{X}$}
       \subtypLine \UnaryInfC{$\varm{f}{X}\indvdash I \varm{f}{X}$}
       \DisplayProof
       \raiseRel{:=}
       \bottomAlignProof
       \AxiomC{$I\vdash {\varm{f}{X}}\ \mathtt{def}$}
       \UnaryInfC{${\varm{f}{X}}\indvdash I {\varm{f}{X}}$}
       \DisplayProof &&
       \bottomAlignProof
       \AxiomC{}
       \UnaryInfC{$\1\subtype \1$}
       \subtypLine
       \UnaryInfC{$\indvdash I \1$}
       \DisplayProof
       \raiseRel{:=}
       \bottomAlignProof
       \AxiomC{\vphantom{A}}
       \UnaryInfC{$\indvdash I \1$}
       \DisplayProof \vspace{0.7em}\\
       \bottomAlignProof
       \AxiomC{$A\subtype A'$}
       \AxiomC{$B\subtype B'$}
       \BinaryInfC{$A\otimes B\subtype A\otimes B$}
       \subtypLine
       \UnaryInfC{$A\otimes B\indvdash I A\otimes B$}
       \DisplayProof
       \raiseRel{:=}
       \bottomAlignProof
       \AxiomC{$A\subtype A'$}
       \subtypLine
       \UnaryInfC{$A\indvdash I A'$}
       \AxiomC{$B\subtype B'$}
       \subtypLine
       \UnaryInfC{$B\indvdash I B'$}
       \BinaryInfC{$A,B\indvdash I A\otimes B$}
       \UnaryInfC{$A\otimes B\indvdash I A\otimes B$}
       \DisplayProof &&
       \bottomAlignProof
       \AxiomC{}
       \UnaryInfC{$\top\subtype \top$}
       \subtypLine
       \UnaryInfC{$\top\indvdash \emptyset \top$}
       \DisplayProof
       \raiseRel{:=}
       \bottomAlignProof
       \AxiomC{$\emptyset\vdash\top\ \mathtt{def}$}
       \UnaryInfC{$\top\indvdash \emptyset \top$}
       \DisplayProof
     \end{align*}
  \begin{align*}
    \shortVOnly{&}
    \bottomAlignProof
    \AxiomC{$i\Bot j'\quad j\Bot i'$}
    \AxiomC{$\indfun{\pb{i'}{i}}{A}\subtype \indfun{\pb{i}{i'}}{A'}$}
    \AxiomC{$\indfun{\pb{j'}{j}}{B}\subtype \indfun{\pb{j}{j'}}{B'}$}
    \TrinaryInfC{$A\withm{i}{j}B\subtype A'\withm{i'}{j'} B'$}
    \subtypLine
    \UnaryInfC{$A\withm{i}{j}B\indvdash {I} A'\withm{i'}{j'} B'$}
    \DisplayProof
    \shortVOnly{\\}
    &\raiseRel{:=}
    \bottomAlignProof
    \AxiomC{$\indfun{\pb{i'}{i}}{A}\subtype \indfun{\pb{i}{i'}}{A'}$}
    \dashedLine
    \UnaryInfC{$\indfun{\inv{\pb{i}{i'}}}{\indfun{\pb{i'}{i}}{A}}\subtype A'$}
    \subtypLine
    \UnaryInfC{$\indfun{\inv{\pb{i}{i'}}}{\indfun{\pb{i'}{i}}{A}}\indvdash I A'$}
\UnaryInfC{$\indfun{\pb{i'}{i}}{A}\withm{\pb{i}{i'}}{\init} \indfun{\init}{B}\indvdash I A'$}
    \AxiomC{$\indfun{\pb{j'}{j}}{B}\subtype \indfun{\pb{j}{j'}}{B'}$}
    \dashedLine
    \UnaryInfC{$\indfun{\inv{\pb{j}{j'}}}{\indfun{\pb{j'}{j}}{B}}\subtype B'$}
    \subtypLine
    \UnaryInfC{$\indfun{\inv{\pb{j}{j'}}}{\indfun{\pb{j'}{j}}{B}}\indvdash I B'$}
\UnaryInfC{$\indfun{\init}{A}\withm{\init}{\pb{j}{j'}}\indfun{\pb{j'}{j}}{B}\indvdash I B'$}
    \BinaryInfC{$A\withm{i}{j}B\indvdash {I} A'\withm{i'}{j'} B'$}
    \DisplayProof
  \end{align*}
  }

  \longVOnly{
  In the fixedpoint case, we are actively using a coinductive definition,
  which is safe since we are producing a $\mu^\infty$-elimination:
  \[
    \bottomAlignProof
    \AxiomC{$\indfun{f}{A[\mu X.A]}\subtype \indfun{g}{A'[\mu Y.A']}$}
    \productiveLine
    \UnaryInfC{$\mum fX.A\subtype \mum {g}Y.A'$}
    \subtypLine
    \UnaryInfC{$\mum fX.A \indvdash I \mum {g}Y.A'$}
    \DisplayProof
    \raiseRel{:=}
    \bottomAlignProof
    \AxiomC{$\indfun{f}{A[\mu X.A/X]} \subtype \indfun{g}{A'[\mu Y.A'/X]}$}
    \subtypLine
    \UnaryInfC{$\indfun{f}{A[\mu X.A/X]} \indvdash I \indfun{g}{A'[\mu Y.A'/X]}$}
    \UnaryInfC{$\indfun{f}{A[\mu X.A/X]} \indvdash I \mum{g}Y.A'$}
    \productiveLine
    \UnaryInfC{$ \mum fX.A\indvdash I \mum{g}Y.A'$}
    \DisplayProof
    \]
    }
    
    Consider now the case
    \[  \raiseRel{\rho =} \bottomAlignProof
  \AxiomC{$\indfun{g}{A}\subtype A'$}
  \UnaryInfC{$\bangm uA\subtype\bangm{g;u}A'$}
  \DisplayProof \]
  In the following, we use $f$ defined as the only function such that $f;\pb{u}{(g;u)}=\id$ and $f;\pb{(g;u)}{u}=g$ which exists since $g;u=id;(g;u).$ we then define
  \[ 
  \raiseRel{\pi_\rho :=}\hspace{-1.5em}
  \bottomAlignProof
  \AxiomC{$\indfun{g}{A}\subtype A'$}
  \AxiomC{}
  \UnaryInfC{$f;\pb{(g;u)}{u}=g$}
  \dashedLine
  \BinaryInfC{$\indfun{f}{\indfun{\pb{(g;u)}{u}}{A}}\subtype  A'$}
  \subtypLine
  \UnaryInfC{$\indfun{f}{\indfun{\pb{(g;u)}{u}}{A}}\indvdash J A'$}
  \AxiomC{\shortVOnly{\hspace{-3.1em}}$u{\circ} f=\id$}
  \BinaryInfC{$\bangm{\pb{u}{(g;u)}}\indfun{\pb{(g;u)}{u}}{A}\indvdash J A'$}
  \AxiomC{\shortVOnly{\hspace{-2.1em}}$(g;u): J{\rightarrow} I$}
  \BinaryInfC{$\bangm uA\indvdash I \bangm{g;u}A'$}
  \DisplayProof
  \]
\end{proof}

In what follows, we shall often write $\AxiomC{$\pi'$}
  \UnaryInfC{$A\subtype A'$}
  \DisplayProof $ for a proof coming from a subtyping derivation. We will also sometimes just write $ A \subtype A'  ,$ omitting the proof specification.

\begin{proposition}[Up-to identities]\label{MZ_stype_subtyping}
  The subtyping proofs are the proofs collapsing to identities:
  \[
  \AxiomC{$\underline \pi$\vphantom{$\id_{\underline A}$}}
  \UnaryInfC{$\underline A\vdash \underline{A'}$}
  \DisplayProof
  = 
  \AxiomC{$\id_{\underline{A}}$}
  \UnaryInfC{$\underline{A}\vdash \underline{A}$}
  \DisplayProof
  \quad \iff\quad
  \exists\rho.\ 
  \AxiomC{$\pi$}
  \UnaryInfC{$A\indvdash I A'$}
  \DisplayProof
  =
  \AxiomC{$\rho$}
  \UnaryInfC{$A\subtype A'$}
  \subtypLine
  \UnaryInfC{$A\indvdash I A'$}
  \DisplayProof.
  \]
\end{proposition}

We shall now prove that the standard subtyping rule is derivable in our system. In order to do so, we also need to extend the base change operation to proofs. The following theorem, and its proof (detailed in the long version) really express how  base change and subtyping are deeply connected.

\begin{lemma}[Base changes over proofs]\label{lemma:BS_and_ST_rules}
  The following rules are derivable:
  \[
  \AxiomC{${A_1},...,{A_n}\indvdash J {B_1}, \dots, {B_m}$}
  \AxiomC{$f : I\rightarrow J$}
  \functorialLine
  \BinaryInfC{$\indfun{f}{A_1},...,\indfun{f}{A_n}\indvdash I \indfun{f}{B_1}, \dots, \indfun{f}{B_m}$}
  \DisplayProof
  \quad
  \AxiomC{$\Gamma\indvdash J B, \Delta$}
  \AxiomC{$B\subtype B'$}
  \subtypLine
  \BinaryInfC{$\Gamma\indvdash J B', \Delta$}
  \DisplayProof
  \]
  By abuse of notation, the hypothesis $f : I\rightarrow J$ and even $B\subtype B'$ may be sometimes omitted, even though they are proof-relevant (i.e., the resulting proof is different for a different choice).
  \longVOnly{\\
  In addition, the resulting proofs have the exact same structures, i.e, if:
  \[\pi'_1=
  \AxiomC{$\pi_1$}
  \UnaryInfC{${A_1},...,{A_n}\indvdash J {B_1}, \dots, {B_m}$}
  \AxiomC{$f : I\rightarrow J$}
  \functorialLine
  \BinaryInfC{$\indfun{f}{A_1},...,\indfun{f}{A_n}\indvdash I \indfun{f}{B_1}, \dots, \indfun{f}{B_m}$}
  \DisplayProof
  \quad
  \pi'_2=
  \AxiomC{$\pi_2$}
  \UnaryInfC{$\Gamma\indvdash J B, \Delta$}
  \AxiomC{$B\subtype B'$}
  \subtypLine
  \BinaryInfC{$\Gamma\indvdash J B', \Delta$}
  \DisplayProof
  \]
  Then $\underline{\pi_1}=\underline{\pi_1'}$ and $\underline{\pi_2}=\underline{\pi_2'}$.}
\end{lemma}
\begin{proof}
  We are proceeding by induction on the proof $\underline \pi$ of the main hypothesis.
  \shortVOnly{
   For lack of space, we are only presenting a few cases. Dereliction:\\
   \AxiomC{$\Gamma\indvdash J g(B)$}
   \AxiomC{}
   \UnaryInfC{$u\circ g=\id$}
   \BinaryInfC{$\Gamma\indvdash J \wnm{u}B$}
   \AxiomC{$f$}
   \functorialLine
   \BinaryInfC{$f(\Gamma)\indvdash I \wnm{\pb{u}{f}}\pb{f}{u}(B)$}
   \DisplayProof$\quad :=$
   \begin{flushright}
     \AxiomC{$\Gamma\indvdash J g(B)$}
     \AxiomC{$f$}
     \functorialLine
     \BinaryInfC{$f(\Gamma)\indvdash I f(g(B))$}
     \AxiomC{}
     \UnaryInfC{$f;g=h;\pb{f}{u}$}
     \UnaryInfC{$f(g(B))\eqtype h;\pb{f}{u}(B)$}
     \subtypLine
     \BinaryInfC{$f(\Gamma)\indvdash I h(\pb{f}{u}(B))$}
     \AxiomC{}
     \UnaryInfC{$\pb{u}{f}\circ h = \id$}
     \BinaryInfC{$f(\Gamma)\indvdash I \wnm{\pb{u}{f}}\pb{f}{u}(B)$}
     \DisplayProof
   \end{flushright}

    where $h$ is the only function s.t. $\pb{f}{u}\circ h= \pb{f}{g}$ and $\pb{u}{f}\circ h= g_{f}$ which exists since $(f;g);u=f$.

    \AxiomC{$\indfun{v}{\bangm{w}A}\indvdash K B$}
\UnaryInfC{$\bangm{w}A\indvdash J \bangm vB$}
    \AxiomC{$f$}
    \functorialLine
    \BinaryInfC{$\indfun{f}{\bangm{w}A}\indvdash I \indfun{f}{\bangm{v}B}$}
    \DisplayProof\quad$:=$\vspace{-.5em}
    \begin{flushright}
        \AxiomC{$\indfun{v}{\bangm{w}A}\indvdash K B$}
        \AxiomC{$\pb{f}{v}$}
        \functorialLine
        \BinaryInfC{$\indfun{\pb{f}{v}}{\indfun{v}{\bangm{w}A}}\indvdash {I \times_{J} K} \indfun{\pb{f}{v}}{B}$}
        \AxiomC{}
        \UnaryInfC{$\pb{v}{f};f=\pb{f}{v};v$}
        \dashedLine
        \UnaryInfC{$\indfun{\pb{v}{f}}{\indfun{f}{\bangm{w}A}}\eqtype \indfun{\pb{f}{v}}{\indfun{v}{\bangm{w}A}}$}
        \subtypLine
        \BinaryInfC{$\indfun{\pb{v}{f}}{\bangm{\pb{w}{f}}\indfun{\pb{f}{w}}{A}}\indvdash {I \times_{J} K} \pb{f}{v}(B)$}
        \UnaryInfC{$\indfun{f}{\bangm{w}A}\indvdash I \indfun{f}{\bangm{v}B}$}
        \DisplayProof
    \end{flushright}
\AxiomC{$\indfun{v}{\bangm{w}A_1},\dots,\indfun{v}{\bangm{w}A_n}\indvdash J B$}
    \UnaryInfC{$\bangm{w_1}A_1,\dots,\bangm{w_n}A_n\indvdash I \bangm{v}B$}
    \AxiomC{$\indfun{g}{B}\subtype B'$}
    \UnaryInfC{$\bangm{v}B \subtype \bangm{v\circ g}B'$}
    \subtypLine
    \BinaryInfC{$\bangm{w_1}A_1,\dots,\bangm{w_n}A_n\indvdash I \bangm{v\circ g}B'$}
    \DisplayProof\quad$:=$
    \begin{flushright}
      \AxiomC{$\indfun{v}{\bangm{w}A_1},\!...,\indfun{v}{\bangm{w}A_n}\indvdash J B$}
\functorialLine
      \UnaryInfC{$\indfun{g}{\indfun{v}{\bangm{w_{1}}A_1}},\!...,\indfun{g}{\indfun{v}{\bangm{w_{n}}A_n}}\indvdash {J'} \indfun{g}{B}$}
      \AxiomC{\hspace{-2em}$\indfun{g}{B}\subtype B'$}
      \subtypLine
      \BinaryInfC{$\indfun{g}{\indfun{v}{\bangm{w_1}A_1}},\!...,\indfun{g}{\indfun{v}{\bangm{w_{n}}A_n}}\indvdash {J'} B'$}
      \AxiomC{$g;v = g;v$}
      \subtypLine
      \BinaryInfC{$\indfun{(g;v)}{\bangm{w_{1}}A_1},\!...,\indfun{(g;v)}{\bangm{w_{n}}A_n}\indvdash {J'} B'$}
      \UnaryInfC{$\bangm{w_1}A_1,\!...,\bangm{w_n}A_n\indvdash I \bangm{g;v}B'$}
      \DisplayProof\vspace{-1.5em}
    \end{flushright}
  }
  
  \longVOnly{
    \begin{gather*}
      \AxiomC{}
      \UnaryInfC{$\varm{g}{X}\indvdash J \varm{g}{X}$}
      \AxiomC{$f$}
      \functorialLine
      \BinaryInfC{$\varm{(f;g)}{X}\indvdash I \varm{(f;g)}{X}$}
      \DisplayProof
       := 
      \AxiomC{}
      \UnaryInfC{$\varm{(f;g)}{X}\indvdash I \varm{(f;g)}{X}$}
      \DisplayProof\\
      \AxiomC{$\Gamma\indvdash J A$}
      \AxiomC{$A\indvdash J \Delta$}
\BinaryInfC{$\Gamma\indvdash J \Delta$}
      \AxiomC{$f$}
      \functorialLine
      \BinaryInfC{$\indfun{f}{\Gamma}\indvdash I \indfun{f}{\Delta}$}
      \DisplayProof
      := 
      \AxiomC{$\Gamma\indvdash J A$}
      \AxiomC{$f$}
      \functorialLine
      \BinaryInfC{$\indfun{f}{\Gamma}\indvdash I \indfun{f}{A}$}
      \AxiomC{$A\indvdash I \Delta$}
      \AxiomC{$f$}
      \functorialLine
      \BinaryInfC{$\indfun{f}{A}\indvdash I \indfun{f}{\Delta}$}
\BinaryInfC{$\indfun{f}{\Gamma}\indvdash I \indfun{f}{\Delta}$}
      \DisplayProof\\[1em]
      \AxiomC{\vphantom{A}}
      \UnaryInfC{$\indvdash J \1$}
      \AxiomC{$f$}
      \functorialLine
      \BinaryInfC{$\indvdash I \1$}
      \DisplayProof
      := 
      \AxiomC{\vphantom{A}}
      \UnaryInfC{$\indvdash I \1$}
      \DisplayProof\\[1em]
      \AxiomC{$\Gamma\indvdash J A$}
      \AxiomC{$\Delta\indvdash J B$}
      \BinaryInfC{$\Gamma,\Delta\indvdash J A\otimes B$}
      \AxiomC{$f$}
      \functorialLine
      \BinaryInfC{$\indfun{f}{\Gamma},\indfun{f}{\Delta}\indvdash I \indfun{f}{A}\otimes \indfun{f}{B}$}
      \DisplayProof
      := 
      \AxiomC{$\Gamma\indvdash J A$}
      \AxiomC{$f$}
      \functorialLine
      \BinaryInfC{$\indfun{f}{\Gamma}\indvdash I \indfun{f}{A}$}
      \AxiomC{$\Delta\indvdash J B$}
      \AxiomC{$f$}
      \functorialLine
      \BinaryInfC{$\indfun{f}{\Delta}\indvdash I \indfun{f}{B}$}
      \BinaryInfC{$\indfun{f}{\Gamma},\indfun{f}{\Delta}\indvdash I \indfun{f}{A}\otimes \indfun{f}{B}$}
      \DisplayProof\\[1em]
      \AxiomC{$\Gamma\indvdash J A,B$}
      \UnaryInfC{$\Gamma\indvdash J A\parr B$}
      \AxiomC{$f$}
      \functorialLine
      \BinaryInfC{$\indfun{f}{\Gamma}\indvdash I \indfun{f}{A}\parr \indfun{f}{B}$}
      \DisplayProof
      := 
      \AxiomC{$\Gamma\indvdash J A,B$}
      \AxiomC{$f$}
      \functorialLine
      \BinaryInfC{$\indfun{f}{\Gamma}\indvdash I \indfun{f}{A},\indfun{f}{B}$}
      \UnaryInfC{$\indfun{f}{\Gamma}\indvdash I \indfun{f}{A}\parr \indfun{f}{B}$}
      \DisplayProof\\[1em]
      \AxiomC{\optional{$\emptyset\vdash\Gamma\ \mathtt{def}$}}
      \UnaryInfC{$\Gamma\indvdash \emptyset \top$}
      \AxiomC{$f$}
      \functorialLine
      \BinaryInfC{$\indfun{f}{\Gamma}\indvdash \emptyset \top$}
      \DisplayProof
       :=  
      \AxiomC{\optional{$\emptyset\vdash \indfun{f}{\Gamma}\ \mathtt{def}$}}
      \UnaryInfC{$\indfun{f}{\Gamma}\indvdash \emptyset \top$}
      \DisplayProof   \quad   \text{ since the only function f targeting $\emptyset$ is the identity. }\\
\scalebox{0.8}{\parbox{1.05\linewidth}{ \AxiomC{$\indfun{i}\Gamma\indvdash {J_1} A$}
      \AxiomC{$\indfun{j}{\Gamma}\indvdash {J_2} B$}
      \BinaryInfC{$\Gamma\indvdash {K} A\withm{i}{j} B$}
      \AxiomC{$f$}
      \functorialLine
      \BinaryInfC{$\indfun{f}{\Gamma}\indvdash {I} \pb{f}{i}(A)\withm{\pb{i}{f}}{\pb{j}{f}} \pb{f}{j}(B)$}
      \DisplayProof
      := 
      \AxiomC{$\indfun{i}{\Gamma}\indvdash {J_1} A$}
      \AxiomC{$\pb{f}{i}$}
      \functorialLine
      \BinaryInfC{$\indfun{\pb{f}{i}}{\indfun{i}{\Gamma}}\indvdash {f\wedge i} \indfun{\pb{f}{i}}A$}
      \AxiomC{}
      \UnaryInfC{$\pb{i}{f}; f = \pb{f}{i}; i$}
      \UnaryInfC{$\indfun{\pb{i}{f}}{\indfun{f}{\Gamma}}\eqtype \indfun{\pb{f}{i}}{\indfun{i}{\Gamma}}$}
      \subtypLine
      \BinaryInfC{$\indfun{\pb{i}{f}}{\indfun{f}{\Gamma}}\indvdash {f\wedge i} \indfun{\pb{f}{i}}{A}$}
      \AxiomC{$j(\Gamma)\indvdash {J_2} B$}
      \AxiomC{$\pb{f}{j}$}
      \functorialLine
      \BinaryInfC{$\indfun{\pb{f}{j}}{\indfun{j}{\Gamma}}\indvdash {f\wedge j} \indfun{\pb{f}{j}}{B}$}
      \AxiomC{}
      \UnaryInfC{$\pb{j}{f}; f = \pb{f}{j}; j$}
      \UnaryInfC{$\indfun{\pb{j}{f}}{\indfun{f}{\Gamma}}\eqtype \indfun{\pb{f}{j}}{\indfun{j}{\Gamma}}$}
      \subtypLine
      \BinaryInfC{$\pb{j}{f}(\indfun{f}{\Gamma})\indvdash {f\wedge j} \pb{f}{j}(B)$}
      \BinaryInfC{$\indfun{f}{\Gamma}\indvdash {I} \indfun{\pb{f}{i}}{A}\withm{\pb{i}{f}}{\pb{j}{f}} \indfun{\pb{f}{j}}{B}$}
      \DisplayProof}}\\[1em]
\AxiomC{$\Gamma\indvdash J \indfun{\inv i}A$}
      \UnaryInfC{$\Gamma\indvdash J A\oplusm{i}{\init} B$}
      \AxiomC{$f$}
      \functorialLine
      \BinaryInfC{$\indfun{f}{\Gamma}\indvdash I \indfun{\pb{f}{i}}{A}\oplusm{\pb{i}{f}}{\init} \indfun{\id}{B}$}
      \DisplayProof
      := 
      \AxiomC{$\Gamma\indvdash J \indfun{\inv i}A$}
      \AxiomC{$f$}
      \functorialLine
      \BinaryInfC{$\indfun{f}{\Gamma}\indvdash I \indfun{f}{\indfun{\inv i}{A}}$}
      \AxiomC{}
      \UnaryInfC{$f;\inv i = \inv{\pb if};\pb{f}{i}$}
      \UnaryInfC{$\indfun{f}{\indfun{\inv i}{A}}\eqtype \indfun{\inv{\pb if}}{\indfun{\pb{f}{i}}{A}}$}
      \subtypLine
      \BinaryInfC{$\indfun{f}{\Gamma}\indvdash I \indfun{\inv{\pb if}}{\indfun{\pb{f}{i}}{A}}$}
      \UnaryInfC{$\indfun{f}{\Gamma}\indvdash I \indfun{\pb{f}{i}}{A}\oplusm{\pb{i}{f}}{\init} \indfun{\id}{B}$}
      \DisplayProof
    \end{gather*}

    \begin{gather*}
      \AxiomC{$\Gamma\indvdash J $}
      \UnaryInfC{$\Gamma\indvdash J \wnm{u}B$}
      \AxiomC{$f$}
      \functorialLine
      \BinaryInfC{$f(\Gamma)\indvdash I \wnm{\pb{u}{f}}\pb{f}{u}(B)$}
      \DisplayProof
      := 
      \AxiomC{$\Gamma\indvdash J $}
      \AxiomC{$f$}
      \functorialLine
      \BinaryInfC{$f(\Gamma)\indvdash I $}
      \UnaryInfC{$f(\Gamma)\indvdash I \wnm{\pb{u}{f}}\pb{f}{u}(B)$}
      \DisplayProof\\[1em]
\AxiomC{$\Gamma\indvdash J g(B)$}
      \AxiomC{}
      \UnaryInfC{$u\circ g=\id$}
      \BinaryInfC{$\Gamma\indvdash J \wnm{u}B$}
      \AxiomC{$f$}
      \functorialLine
      \BinaryInfC{$f(\Gamma)\indvdash I \wnm{\pb{u}{f}}\pb{f}{u}(B)$}
      \DisplayProof
      := 
      \AxiomC{$\Gamma\indvdash J g(B)$}
      \AxiomC{$f$}
      \functorialLine
      \BinaryInfC{$f(\Gamma)\indvdash I f(g(B))$}
      \AxiomC{}
      \UnaryInfC{$f;g=h;\pb{f}{u}$}
      \UnaryInfC{$f(g(B))\eqtype h;\pb{f}{u}(B)$}
      \subtypLine
      \BinaryInfC{$f(\Gamma)\indvdash I h(\pb{f}{u}(B))$}
      \AxiomC{}
      \UnaryInfC{$\pb{u}{f}\circ h = \id$}
      \BinaryInfC{$f(\Gamma)\indvdash I \wnm{\pb{u}{f}}\pb{f}{u}(B)$}
      \DisplayProof
      \\
      \text{where } h \text{ is the only function s.t. } \pb{f}{u}\circ h= \pb{f}{g} \text{ and } \pb{u}{f}\circ h= g_{f} \text{ which exists since } (f;g);u=f\\[1em]
      \AxiomC{$\Gamma\indvdash J \wnm{w}A,\wnm{w}A$}
\UnaryInfC{$\Gamma\indvdash J \wnm wA$}
      \AxiomC{$f$}
      \functorialLine
      \BinaryInfC{$f(\Gamma)\indvdash I \wnm{\pb{w}{ f}} \pb{f}{w}(A)$}
      \DisplayProof
       := 
      \AxiomC{$\Gamma\indvdash J \wnm{w}A,\wnm{w}A$}
      \AxiomC{$f$}
      \functorialLine
      \BinaryInfC{$f(\Gamma)\indvdash I \wnm{\pb{w}{f}}\pb{f}{w}(A),\wnm{\pb{w}{f}}\pb{f}{w}(A)$}
\UnaryInfC{$f(\Gamma)\indvdash I \wnm{\pb{w}{f}} \pb{f}{w}(A)$}
      \DisplayProof
    \end{gather*}

  \hspace*{-3cm}

      \begin{align*}
      &\AxiomC{$\indfun{v}{\bangm{w}A}\indvdash K B$}
      \AxiomC{$\hspace*{-1em}v: J\leftarrow K$}
      \BinaryInfC{$\bangm{w}A\indvdash J \bangm vB$}
      \AxiomC{$f$}
      \functorialLine
      \BinaryInfC{$\indfun{f}{\bangm{w}A}\indvdash I \indfun{f}{\bangm{v}B}$}
      \DisplayProof\\
      &:= 
      \AxiomC{$\indfun{v}{\bangm{w}A}\indvdash K B$}
      \AxiomC{$\pb{f}{v}$}
      \functorialLine
      \BinaryInfC{$\indfun{\pb{f}{v}}{\indfun{v}{\bangm{w}A}}\indvdash {I \times_{J} K} \indfun{\pb{f}{v}}{B}$}
      \AxiomC{}
      \UnaryInfC{$\pb{v}{f};f=\pb{f}{v};v$}
      \dashedLine
      \UnaryInfC{$\indfun{\pb{v}{f}}{\indfun{f}{\bangm{w}A}}\eqtype \indfun{\pb{f}{v}}{\indfun{v}{\bangm{w}A}}$}
      \subtypLine
      \BinaryInfC{$\indfun{\pb{v}{f}}{\bangm{\pb{w}{f}}\indfun{\pb{f}{w}}{A}}\indvdash {I \times_{J} K} \pb{f}{v}(B)$}
      \UnaryInfC{$\indfun{f}{\bangm{w}A}\indvdash I \indfun{f}{\bangm{v}B}$}
      \DisplayProof\\
      &\qquad\text{where } e : (L\times_JI)\times_I(I\times_J K) \simeq (I\times_J K)\times_{K}(K\times_J L)
      \text{ is the iso of double pullback}
    \end{align*}

  \smallskip
  
  \hspace*{-1cm}
  For the second transformation, we also need to unfold the structure of~${B\subtype B'}$ :\\
if $i\Bot j'$ and $i$ bijective, then $j'=\init$, thus\\
    \begin{align*}
      \AxiomC{$\Gamma\indvdash I \indfun{\inv{i}}{A}$}
      \UnaryInfC{$\Gamma\indvdash {I} A\oplusm{i}{\init} B$}
      \AxiomC{$\indfun{\pb{i'}{i}}{A}\subtype \indfun{\pb{i}{i'}}{A'}$}
      \AxiomC{$\indfun{\init}B\subtype \indfun{\init}{B'}$}
      \BinaryInfC{$A\oplusm{i}{\init} B \subtype A'\oplusm{i'}{\init} B'$}
      \subtypLine
      \BinaryInfC{$\Gamma\indvdash I A'\oplusm{i'}{\init} B'$}
      \DisplayProof
      &:=
      \AxiomC{$\Gamma\indvdash I \indfun{\inv{i}}{A}$}
      \AxiomC{$\indfun{\pb{i'}{i}}{A}\subtype \indfun{\pb{i}{i'}}{A'}$}
      \dashedLine
      \UnaryInfC{$\indfun{i';\inv{\pb{i}{i'}}}{\indfun{\pb{i'}{i}}{A}}\subtype \indfun{i';\inv{\pb{i}{i'}}}{\indfun{\pb{i}{i'}}{A'}}$}
      \dashedLine
      \UnaryInfC{$\indfun{\inv{i}}{A}\subtype \indfun{\inv{{i'}}}{A'}$}
      \subtypLine
      \BinaryInfC{$\Gamma\indvdash I \indfun{\inv{{i'}}}{A'}$}
      \UnaryInfC{$\Gamma\indvdash I A'\oplusm{i'}{\init} B'$}
      \DisplayProof
    \end{align*}

      \begin{align*}
      &\text{assuming that } i\Bot j',\ j\Bot i',\ \text{ and }i_{i'}, i'_{i},j_{j'},j'_{i} \text{ bijections :}\\
      &\AxiomC{$i(\Gamma)\indvdash I A$}
      \AxiomC{$\hspace*{-1em}\indfun{j}{\Gamma}\indvdash J B$}
      \BinaryInfC{$\Gamma\indvdash {K} A\withm{i}{j} B$}
\AxiomC{$\indfun{\pb{i'}{i}}{A}\subtype \indfun{\pb{i}{i'}}{A'}$}
      \AxiomC{$\indfun{\pb{j'}{j}}{B}\subtype \indfun{\pb{j}{j'}}{B'}$}
      \BinaryInfC{$A\withm{i}{j} B \subtype A'\withm{i'}{j'} B'$}
      \subtypLine
      \BinaryInfC{$\Gamma\indvdash K A'\withm{i'}{j'} B'$}
      \DisplayProof  :=\\
      & \scalebox{0.8}{\parbox{1.05\linewidth}{\AxiomC{$\indfun{i}{\Gamma}\indvdash I A$}
      \UnaryInfC{$\indfun{i'_i}{\indfun{i}{\Gamma}}\indvdash {(I{\times_K} I')} \indfun{\pb{i'}{i}}{A}$}
      \AxiomC{$\indfun{\pb{i'}{i}}{A}\subtype \indfun{\pb{i}{i'}}{A'}$}
      \subtypLine
      \BinaryInfC{$\indfun{\pb{i'}{i}}{\indfun{i}{\Gamma}}\indvdash {(I{\times_K} I')} \indfun{\pb{i}{i'}}{A'}$}
      \functorialLine
      \UnaryInfC{$\indfun{\inv{\pb{i}{i'}}}{\indfun{i'_i}{\indfun{i}{\Gamma}}}\indvdash {I'} \indfun{\inv{\pb{i}{i'}}}{\indfun{\pb{i}{i'}}{A'}}$}
      \AxiomC{\hspace{-6em}}
      \AxiomC{}
      \UnaryInfC{$\indfun{\inv{\pb{i}{i'}}}{\indfun{i'_i}{\indfun{i}{\Gamma}}}\eqtype \indfun{i'}{\Gamma}$}
      \subtypLine
      \TrinaryInfC{$\indfun{i'}{\Gamma}\indvdash {I'} \indfun{\inv{\pb{i}{i'}}}{\indfun{\pb{i}{i'}}{A'}}$}
      \AxiomC{}
      \UnaryInfC{$\indfun{\inv{\pb{i}{i'}}}{\indfun{\pb{i}{i'}}{A'}} \eqtype A'$}
      \subtypLine
      \BinaryInfC{$\indfun{i'}{\Gamma}\indvdash {I'} A'$}
      \AxiomC{\hspace{-6.5em}}
      \AxiomC{$\indfun{j}{\Gamma}\indvdash J B$}
      \UnaryInfC{$\indfun{j'_j}{\indfun{j}{\Gamma}}\indvdash {(J{\times_K} J')} \indfun{\pb{j'}{j}}{B}$}
      \AxiomC{$\indfun{\pb{j'}{j}}{B}\subtype \indfun{\pb{j}{j'}}{B'}$}
      \subtypLine
      \BinaryInfC{$\indfun{\pb{j'}{j}}{\indfun{j}{\Gamma}}\indvdash {(J{\times_K} J')} \indfun{\pb{j}{j'}}{B'}$}
      \functorialLine
      \UnaryInfC{$\indfun{\inv{\pb{j}{j'}}}{\indfun{\pb{j'}{j}}{\indfun{j}{\Gamma}}}\indvdash {J'} \indfun{\inv{\pb{j}{j'}}}{\indfun{\pb{j}{j'}}{B'}}$}
      \AxiomC{\hspace{-6em}}
      \AxiomC{}
      \UnaryInfC{$\indfun{\inv{\pb{j}{j'}}}{\indfun{\pb{j'}{j}}{\indfun{j}{\Gamma}}}\eqtype \indfun{j'}{\Gamma}$}
      \subtypLine
      \TrinaryInfC{$\indfun{j'}{\Gamma}\indvdash {J'} \indfun{\inv{\pb{j}{j'}}}{\indfun{\pb{j}{j'}}{B'}}$}
      \AxiomC{}
      \UnaryInfC{$\indfun{\inv{\pb{j}{j'}}}{\indfun{\pb{j}{j'}}{B'}} \eqtype B'$}
      \subtypLine
      \BinaryInfC{$\indfun{j'}{\Gamma}\indvdash {J'} B'$}
      \TrinaryInfC{$\Gamma\indvdash K A'\withm{i'}{j'} B'$}
      \DisplayProof}}
      \end{align*}

    \begin{gather*}
      \text{if, in addition, }j=\init\text{ then } j'=init \text{ and }\\
\AxiomC{$\inv{i}(\Gamma)\indvdash {I'} A$}
      \AxiomC{\hspace*{-1.2em}$\emptyset\vdash B\ \mathtt{def}$}
      \BinaryInfC{$\Gamma\indvdash I A\oplusm{i}{\init} B$}
      \AxiomC{$\pb{i'}{i}(A)\subtype \pb{i}{i'}(A')$}
      \AxiomC{$\id(B)\subtype \id(B')$}
      \BinaryInfC{$A\oplusm{i}{\init} B \subtype A'\oplusm{i'}{\init} B'$}
      \subtypLine
      \BinaryInfC{$\Gamma\indvdash I A'\oplusm{i'}{\init} B'$}
      \DisplayProof :=
      \AxiomC{$\inv{i}(\Gamma)\indvdash {I'} A$}
      \AxiomC{$\pb{i'}{i}(A)\subtype \pb{i}{i'}(A')$}
      \subtypLine
      \BinaryInfC{$\inv{i'}(\Gamma)\indvdash {I'} A'$}
      \AxiomC{\hspace*{-1.2em}$\emptyset\vdash B'\ \mathtt{def}$}
      \BinaryInfC{$\Gamma\indvdash I A'\oplusm{i'}{\init} B'$}
      \DisplayProof\\[1em]
\AxiomC{$\Gamma\indvdash I \wnm{(u\circ g)}A,\wnm{(u\circ g)}A$}
      \UnaryInfC{$\Gamma\indvdash I \wnm{u\circ g}A$}
      \AxiomC{$A\subtype g(A')$}
      \UnaryInfC{$\wnm{u\circ g}A \subtype \wnm{u}A'$}
      \subtypLine
      \BinaryInfC{$\Gamma\indvdash I \wnm{u}A'$}
      \DisplayProof :=
      \AxiomC{$\Gamma\indvdash I \wnm{u\circ g}A,\wnm{u\circ g}A$}
      \AxiomC{$A\subtype g(A')$}
      \UnaryInfC{$\wnm{u\circ g}A \subtype \wnm{u}A'$}
      \subtypLine
      \BinaryInfC{$\Gamma\indvdash I \wnm{u\circ g}A_{J},\wnm{u}A'$}
      \AxiomC{$A\subtype g(A')$}
      \UnaryInfC{$\wnm{u\circ g}A \subtype \wnm{u}A'$}
      \subtypLine
      \BinaryInfC{$\Gamma\indvdash I \wnm{u}A',\wnm{u}A'$}
      \UnaryInfC{$\Gamma\indvdash I \wnm{u}A'$}
      \DisplayProof\\[1em]
\AxiomC{$\Gamma\indvdash I $}
      \UnaryInfC{$\Gamma\indvdash I \wnm{u\circ g}A$}
      \AxiomC{$A\subtype g(A')$}
      \UnaryInfC{$\wnm{u\circ g}A \subtype \wnm{u}A'$}
      \subtypLine
      \BinaryInfC{$\Gamma\indvdash I \wnm{u}A'$}
      \DisplayProof :=
      \AxiomC{$\Gamma\indvdash I $}
      \UnaryInfC{$\Gamma\indvdash I \wnm{u}A'$}
      \DisplayProof\\[1em]
\AxiomC{$\Gamma\indvdash I \indfun{f}{A}$}
      \AxiomC{$u\circ g \circ f = \id$}
      \BinaryInfC{$\Gamma\indvdash I \wnm{u\circ g}A$}
      \AxiomC{$A\subtype \indfun{g}{A'}$}
      \UnaryInfC{$\wnm{u\circ g}A \subtype \wnm{u}A'$}
      \subtypLine
      \BinaryInfC{$\Gamma\indvdash I \wnm{u}A'$}
      \DisplayProof :=
      \AxiomC{$\Gamma\indvdash I \indfun{f}{A}$}
      \AxiomC{$A\subtype \indfun{g}{A'}$}
      \UnaryInfC{$\indfun{f}{A} \subtype \indfun{g\circ f}{A'}$}
      \subtypLine
      \BinaryInfC{$\Gamma\indvdash I \indfun{(g\circ f)}{A'}$}      
      \AxiomC{$u\circ g \circ f = \id$}
      \BinaryInfC{$\Gamma\indvdash I \wnm{u}A'$}
      \DisplayProof
    \end{gather*}

      \begin{align*}
      &
      \AxiomC{$\indfun{v}{\bangm{w}A_1},\dots,\indfun{v}{\bangm{w}A_n}\indvdash J B$}
      \UnaryInfC{$\bangm{w_1}A_1,\dots,\bangm{w_n}A_n\indvdash I \bangm{v}B$}
      \AxiomC{$\indfun{g}{B}\subtype B'$}
      \UnaryInfC{$\bangm{v}B \subtype \bangm{v\circ g}B'$}
      \subtypLine
      \BinaryInfC{$\bangm{w_1}A_1,\dots,\bangm{w_n}A_n\indvdash I \bangm{v\circ g}B'$}
      \DisplayProof\\
      &:=
      \AxiomC{$\indfun{v}{\bangm{w}A_1},\dots,\indfun{v}{\bangm{w}A_n}\indvdash J B$}
      \AxiomC{$g$}
      \functorialLine
      \BinaryInfC{$\indfun{g}{\indfun{v}{\bangm{w_{1}}A_1}},\dots,\indfun{g}{\indfun{v}{\bangm{w_{n}}A_n}}\indvdash {J'} \indfun{g}{B}$}
      \AxiomC{$\indfun{g}{B}\subtype B'$}
      \subtypLine
      \BinaryInfC{$\indfun{g}{\indfun{v}{\bangm{w_1}A_1}},\dots,\indfun{g}{\indfun{v}{\bangm{w_{n}}A_n}}\indvdash {J'} B'$}
      \AxiomC{$v\circ g = g;v$}
      \UnaryInfC{$\indfun{(v\circ g)}{\bangm{w_{i}}A_i} \subtype \indfun{g}{\indfun{v}{\bangm{w_i}A_i}}$}
      \subtypLine
      \BinaryInfC{$\indfun{(v\circ g)}{\bangm{w_{1}}A_1},\dots,\indfun{(v\circ g)}{\bangm{w_{n}}A_n}\indvdash {J'} B'$}
      \UnaryInfC{$\bangm{w_1}A_1,\dots,\bangm{w_n}A_n\indvdash I \bangm{v\circ g}B'$}
      \DisplayProof\\[1em]
      &
      \AxiomC{$\bangm{\pb{(w_1\circ g)}{v}}\pb{v}{(w_1\circ g)}(A_1),\bangm{\pb{{w_2}}{v}}\pb{v}{w_2}(A_2),\dots,\bangm{\pb{{w_n}}{v}}\pb{v}{w_n}(A_n)\indvdash J B$}
      \UnaryInfC{$\bangm{w_1\circ g}A_1,\bangm{w_2}A_2,\dots,\bangm{w_n}A_n\indvdash I \bangm vB$}
      \AxiomC{$g(A'_1)\subtype A_1$}
      \UnaryInfC{$\bangm{w_1}A_1' \subtype \bangm{w_1\circ g}A_1$}
      \subtypLine
      \BinaryInfC{$\bangm{w_1}A'_1,\bangm{w_2}A_2,\dots,\bangm{w_n}A_n\indvdash I \bangm vB$}
      \DisplayProof\\
      &:=
      \AxiomC{$\bangm{\pb{(w_1\circ g)}{v}}\pb{v}{(w_1\circ g)}(A_1),\bangm{\pb{{w_2}}{v}}\pb{v}{w_2}(A_2),\dots,\bangm{\pb{{w_n}}{v}}\pb{v}{w_n}(A_n)\indvdash J B$}
      \AxiomC{$h;\pb{v}{w_1} = \pb{v}{(w_1\circ g)};g$}
      \AxiomC{$g(A'_1)\subtype A_1$}
      \BinaryInfC{$h(\pb{v}{w_1}(A'_1)) \subtype \pb{v}{(w_1\circ g)}(A_1)$}
      \UnaryInfC{$\bangm{\pb{{w_1}}{ v}}\pb{v}{w_1}(A'_1) \subtype \bangm{\pb{(w_1\circ g)}{v}}\pb{v}{(w_1\circ g)}(A_1)$}
      \subtypLine
      \BinaryInfC{$\bangm{\pb{{w_1}}{ v}}\pb{v}{w_1}(A'_1),\bangm{\pb{{w_2}}{v}}\pb{v}{w_2}(A_2),\dots,\bangm{\pb{{w_n}}{ v}}\pb{v}{w_n}(A_n)\indvdash J B$}
      \UnaryInfC{$\bangm{w_1}A'_1,\bangm{w_2}A_2,\dots,\bangm{w_n}A_n\indvdash I \bangm vB$}
      \DisplayProof\\
      &\text{where }h\text{ is the only morphism s.t. } h;\pb{v}{w_1} = \pb{v}{(w_1\circ g)};g \text{ and } h;\pb{{w_1}}{ v}=\pb{(w_1\circ g)}{v}\\
    \end{align*}

      \begin{align*}
        \AxiomC{$\Gamma\indvdash I \indfun{f}{A[\mum{\id}X.A/X]}$}
        \UnaryInfC{$\Gamma\indvdash I\mum fX.A$}
        \AxiomC{$\indfun{f}{A[\mu X.A/X]}\subtype \indfun{g}{A'[\mu Y.A'/Y]}$}
        \productiveLine
        \UnaryInfC{$\mum{f}X.A \subtype \mum{g}Y.A'$}
        \subtypLine
        \BinaryInfC{$\Gamma\indvdash I \mum{g}Y.A'$}
        \DisplayProof
        &:=
        \AxiomC{$\Gamma\indvdash I \indfun{f}{A[\mum{\id}X.A/X]}$}
        \AxiomC{$\indfun{f}{A[\mu X.A/X]}\subtype \indfun{g}{A'[\mu Y.A'/Y]}$}
        \subtypLine
        \BinaryInfC{$\Gamma\indvdash I \indfun{g}{A'[\mum{\id}Y.A'/Y]}$}
        \UnaryInfC{$\Gamma\indvdash I \mum{g}Y.A'$}
        \DisplayProof\\[1em]
\AxiomC{$\Gamma\indvdash I \indfun{g}{A[\num{\id}X.A/X]}$}
        \productiveLine
        \UnaryInfC{$\Gamma\indvdash I \num{g}X.A$}
        \AxiomC{$\indfun{g}{A[\nu X.A/X]}\subtype \indfun{f}{A'[\nu Y.A'/Y]}$}
        \productiveLine
        \UnaryInfC{$\num{g}X.A \subtype \num fY.A'$}
        \subtypLine
        \BinaryInfC{$\Gamma\indvdash I \num fY.A'$}
        \DisplayProof
        &:=
        \AxiomC{$\Gamma\indvdash I \indfun{g}{A[\num{\id}X.A/X]}$}
        \AxiomC{$\indfun{g}{A[\nu X.A/X]}\subtype \indfun{f}{A'[\nu Y.A'/Y]}$}
        \subtypLine
        \BinaryInfC{$\Gamma\indvdash I \indfun{f}{A'[\num \id Y.A'/Y]}$}
        \productiveLine
        \UnaryInfC{$\Gamma\indvdash I \num fY.A'$}
        \DisplayProof
\end{align*}

  The remaining cases follow directly from the IH.
 }
\end{proof}

It is important to understand the constructive nature of this lemma, in particular for subtyping rule. Indeed, the resulting proof of base change depends on the choice of $f$ (there could be $f'$ such that $\indfun{f}{A_i}=\indfun{f'}{A_i}$). It is even worst for the proof of the subtyping, which relies on the chosen proof of $B\subtype B'$. The differences are subtle and depend on the choice of locus, but they exist. In Section~\ref{sec:cut} we will quotient proofs so that this difference disappear, but one have to understand this subtlety to understand the quotient.

Another important point is that those operations are invisible if you remove indexes. If you start from a proof $\pi$ and perform any of those two transformations, the proof $\pi'$ you obtain is so that $\underline{\pi}=\underline{\pi'}$. This is because the first operation is recursively applying a base change in the indexes appearing in the proof, while the second operation starts by introducing a cut between $ \pi$ and the proof of $B\indvdash I B'$ obtained from $B\subtype B'$, then it proceeds to eliminate the resulting cut.

\begin{remark}
  For readability, we will use the notation
  \[\indfun{f}{\pi}:=
  \AxiomC{$\pi$}
  \UnaryInfC{$A_1,...,A_n\indvdash I B_1, \dots, B_m$}
  \AxiomC{$f$}
  \functorialLine
  \BinaryInfC{$\indfun{f}{A_1},...,\indfun{f}{A_n}\indvdash I \indfun{f}{B_1}, \dots, \indfun{f}{B_m}$}
  \DisplayProof
  \]
\end{remark}

\longVOnly{
\begin{lemma}[Substitution]
The following proof is derivable :  
\[
\AxiomC{$A\indvdash I A'$}
\AxiomC{$B\indvdash J B'$}
\dashedLine
\BinaryInfC{$A[B/X]\indvdash I A'[B'/X]$}
\DisplayProof
\]
\end{lemma}

\begin{lemma}\label{Park}
  Park's $\nu$-introduction has its derivable counterpart which is defined coinductively :
  \[
  \bottomAlignProof
  \AxiomC{$B\indvdash J A[B/X]$}
  \dasheddefLine
  \UnaryInfC{$\indfun{f}{B}\indvdash I\num fX.A$}
  \DisplayProof\vspace{0.7em}
  \raiseRel{:=}
  \bottomAlignProof
  \AxiomC{$B\indvdash J A[B/X]$}
  \AxiomC{}
  \UnaryInfC{$A\indvdash J A$}
  \AxiomC{$B\indvdash J A[B/X]$}
  \dasheddefLine
  \UnaryInfC{$B\indvdash J \num{\id}X.A$}
  \dashedLine
  \BinaryInfC{$A[B/X]\indvdash J A[\num{\id}X.A/X] $}
\BinaryInfC{$B\indvdash J A[\num{\id}X.A/X]$}
  \functorialLine
  \UnaryInfC{$\indfun{f}{B}\indvdash I \indfun{f}{A[\num{\id}X.A/X]}$}
  \productiveLine
  \UnaryInfC{$\indfun{f}{B}\indvdash I \num fX.A$}
  \DisplayProof\vspace{0.7em}
  \]
\end{lemma}

}

\subsection{Examples and intersection types}
We now show that, as in linear logic, the tensor distributes over the coproduct, and there is an appropriate Seely isomorphism. In addition, the proofs of these derivations project into their equivalent in LL. We use $(\eqvdash)$ to denotes equiprovability.

\begin{proposition}
  The following distribution is admissible:
  \begin{center}
    \AxiomC{}
    \dashedLine
    \UnaryInfC{$A\otimes(B\oplusm{i}{j}C)\ \eqvdash I\ (\indfun{i}{A}\otimes B)\oplusm{i}{j}(\indfun{j}{A}\otimes C)$}
    \DisplayProof
  \end{center}
\end{proposition}

\begin{proposition}[Seely]
  The Seely isomorphisms are admissible:
  \begin{center}
    \AxiomC{}
    \dashedLine
    \UnaryInfC{$\bangm u A\otimes \bangm v B\ \eqvdash I\ \bangm{[ u,v]}(A\withm{\iota_1}{\iota_2} B)$}
    \DisplayProof \end{center}
  where $[u,v]:J\uplus K\rightarrow I$ denotes the copairing of $u$ and $v$.
\end{proposition}
\begin{proof}
    Notice that $\pb{{\iota_i}}{\pi_j}$ and $\pb{{\pi_j}} {\iota_i}$ are the identity if $i=j$ and, respectively, $\term$ and $\init$ if $i\neq j$, thus if $h_1$ and $h_2$ are the only maps such that $h_1;\pb{u}{[u,v]}=\iota_1$, $h_2;\pb{v}{[u,v]}=\iota_2$, $h_1;\pb{[u,v]}{u}=\id$ and $h_2;\pb{[u,v]}{v}=\id$.
    \longVOnly{ \\ \longVOnly{
  }
  \begin{minipage}{\linewidth}
    \begin{small}
      \begin{center}
      \longVOnly{\AxiomC{}
        \dashedLine
        \UnaryInfC{$\id(A)\eqtype A$}
        \UnaryInfC{$\id(A)\indvdash {J} A$}
        \UnaryInfC{$\bangm{\id} \id(A)\indvdash {J} A$}
        \UnaryInfC{$\bangm{\id} \id(A), \bangm{\init}\init(B)\indvdash {J} A$}
        \AxiomC{}
        \dashedLine
        \UnaryInfC{$\id(B)\eqtype B$}
        \UnaryInfC{$\id(B)\indvdash {K} B$}
        \UnaryInfC{$\bangm{\id} \id(B)\indvdash {K} B$}
        \UnaryInfC{$\bangm{\init} \init(A), \bangm{\id} \id(B)\indvdash {K} B$}
        \BinaryInfC{$\bangm{\iota_1} A, \bangm{\iota_2} B\indvdash {K} A\withm{\iota_1}{\iota_2} B$}
        \AxiomC{}
        \dashedLine
        \UnaryInfC{$ h_2(\pb{[u,v]}{v}(B))\subtype \pb{[u,v]}{v}(B)$}
        \UnaryInfC{$\bangm{\pb{v}{[u,v]}} \pb{[u,v]}{v}(A)\subtype \bangm{\iota_2} B$}
        \subtypLine
        \BinaryInfC{$\bangm{\iota_1} A, \bangm{\pb{u}[u,v]} \pb{[u,v]}{v}(B)\indvdash {K} A\withm{\iota_1}{\iota_2} B$}
        \AxiomC{}
        \dashedLine
        \UnaryInfC{$ h_1(\pb{[u,v]}{u}(A))\subtype \pb{[u,v]}{u}(A)$}
        \UnaryInfC{$\bangm{\pb{u}{[u,v]}} \pb{[u,v]}{u}(A)\subtype \bangm{\iota_1} A$}
        \subtypLine
        \BinaryInfC{$\bangm{\pb{u}{[u,v]}} \pb{[u,v]}{u}(A), \bangm{\pb{u}[u,v]} \pb{[u,v]}{v}(B)\indvdash {K} A\withm{\iota_1}{\iota_2} B$}
\UnaryInfC{$\bangm u A, \bangm v B\indvdash I \bangm{[u,v]}(A\withm{\iota_1}{\iota_2} B)$}
        \UnaryInfC{$\bangm u A\otimes \bangm v B\indvdash I \bangm{[u,v]}(A\withm{\iota_1}{\iota_2} B)$}
        \DisplayProof\\[.5em]
        }
\AxiomC{}
        \dashedLine
        \UnaryInfC{$\indfun{\inv{(\pb{\id}{\id})}}{\indfun{\pb{\id}{\id}}{A}} \eqtype A$}
        \subtypLine
        \UnaryInfC{$\indfun{\inv{(\pb{\id}{\id})}}{\indfun{\pb{\id}{\id}}{A}} \indvdash J A$}
\UnaryInfC{$\indfun{\pb{\id}{\id}}{A}\withm{\pb{\id}{\id}}{\init} \indfun{\init}{B} \indvdash J A$}
        \UnaryInfC{$\bangm{\id}(A\withm{\id}{\init} \indfun{\init}{B}) \indvdash J A$}
        \AxiomC{}
        \dashedLine
        \UnaryInfC{$\indfun{\iota_1}{A\withm{\iota_1}{\iota_2} B} \subtype A\withm{\id}{\init} \indfun{\init}{B}$}
        \dashedLine
        \UnaryInfC{$\indfun{h_1}{\indfun{\pb{u}{[u,v]}}{A\withm{\iota_1}{\iota_2} B}} \subtype A\withm{\id}{\init} \indfun{\init}{B}$}
        \UnaryInfC{$\indfun{u}{\bangm{[u,v]}(A\withm{\iota_1}{\iota_2} B)} \subtype \bangm{\id}(A\withm{\id}{\init} \indfun{\init}{B})$}
        \subtypLine
        \BinaryInfC{$\indfun{u}{\bangm{[u,v]}(A\withm{\iota_1}{\iota_2} B)} \indvdash J A$}
        \UnaryInfC{$\bangm{[u,v]}(A\withm{\iota_1}{\iota_2} B) \indvdash I \bangm u A$}
        \AxiomC{}
        \dashedLine
        \UnaryInfC{$\indfun{\inv{(\pb{\id}{\id})}}{\indfun{\pb{\id}{\id}}{B}} \eqtype B$}
        \subtypLine
        \UnaryInfC{$\indfun{\inv{(\pb{\id}{\id})}}{\indfun{\pb{\id}{\id}}{B}} \indvdash K B$}
\UnaryInfC{$\indfun{\init}{A}\withm{\init}{\pb{\id}{\id}} \indfun{\pb{\id}{\id}}{B} \indvdash K B$}
        \UnaryInfC{$\bangm \id(\init(A)\withm{\init}{\id} B) \indvdash K B$}
        \AxiomC{}
        \dashedLine
        \UnaryInfC{$\indfun{\iota_2}{A\withm{\iota_1}{\iota_2} B} \subtype \indfun{\init}{A}\withm{\init}{\id} B$}
        \dashedLine
        \UnaryInfC{$\indfun{h_2}{\indfun{\pb{v}{[u,v]}}{A\withm{\iota_1}{\iota_2} B}} \subtype \indfun{\init}{A}\withm{\init}{\id} B$}
        \UnaryInfC{$\indfun{v}{\bangm{[u,v]}(A\withm{\iota_1}{\iota_2} B)} \subtype \bangm{\id}(\indfun{\init}{A}\withm{\init}{\id} B)$}
        \subtypLine
        \BinaryInfC{$\indfun{v}{\bangm{[u,v]}(A\withm{\iota_1}{\iota_2} B)} \indvdash J B$}
        \UnaryInfC{$\bangm{[u,v]}(A\withm{\iota_1}{\iota_2} B) \indvdash I \bangm v B$}
        \BinaryInfC{$\bangm{[u,v]}(A\withm{\iota_1}{\iota_2} B),\bangm{[u,v]}(A\withm{\iota_1}{\iota_2} B) \indvdash I \bangm u A\otimes \bangm v B$}
        \UnaryInfC{$\bangm{[u,v]}(A\withm{\iota_1}{\iota_2} B) \indvdash I \bangm u A\otimes \bangm v B$}
        \DisplayProof
      \end{center}
    \end{small}
  \end{minipage}

 }
    \shortVOnly{The second derivation is displayed in Figure~\ref{fig:Seely}}
\end{proof}

We then show that the correspondence between a fragment of $\IndLL$ and intersection types (Theorem~\ref{th:ITformulae}) can be push through proofs. However, we will need to relax the statement a bit, because working up-to equivalence in the world of proofs is technically challenging (we will have to do it for $\IndLL$ in the next section).

Before the main theorem, we need a lemma which is akin to Lemma~\ref{formulaeintersection}. This lemma morally states that any two proofs of $\IndLL$ with the same underlying LL proof can be `merged' as one. Not only are we able to define the intersection of formulae, but we can do the same also for proofs. From the `intersection type' point of view, this is just the intersection introduction rule, but from the logical point of view it is a powerful and novel operation on proofs.

\begin{lemma}\label{proofintersection}
  Let $I$ be a set and let $\left[\AxiomC{$\pi_x$}\UnaryInfC{$\indvdash {J_x} \Gamma_x$}\DisplayProof\right]_{x\in I}$ be an indexed set of derivations such that $\underline{\pi_x}=\underline{\pi_y}$ for each $x,y\in I$.\\
  Then\longVOnly{, if the $\pi_x$ are fixedpoint-free,} there is \AxiomC{$\bigwedge_x\pi_x$}\UnaryInfC{$\indvdash{\coprod_{x\in I}J_x} \bigwedge_x\Gamma_x$}\DisplayProof such that $\indfun{\mathtt{inj}_{J_x}}{\bigwedge_x\pi_x}\eqtype \pi_x$.
\end{lemma}
\longVOnly{\todo[inline]{To do}}

\begin{theorem}\label{th:ITproofs}
  The embedding of Definition~\ref{def:emb_it} is coherent with proofs, in the following sense:\\
  Let $t$ be any close lambda term seen as a LL proof derivation $t$ of a formula $\tau$ encoding a simple type. Then there is a refinement \AxiomC{$\pi$}\UnaryInfC{$\indvdash\1 A$}\DisplayProof of $t$ (i.e., s.t. $\underline \pi=t$) iff there is an intersection type derivation of $\vdash t:\llbracket A\rrbracket$.
\end{theorem}
\begin{proof}
  \shortVOnly{
    We generalise with free variables in the context and proceed by induction on $t$. The abstraction case is immediate. The $(\Rightarrow)$ case of the variable uses the fact the any function $f:1\rightarrow I$ is some $\mathtt{cst}_x$. The $(\Leftarrow)$ case of the variable uses the fact that for pairs of functions $f:1\rightarrow I$ and $u:I\rightarrow 1$, $u\circ f=\id$. The $(\Rightarrow)$ case of the application uses the base change $[\mathtt{cst}_x(\pi_2)\mid x\in\dom(u)]$. The $(\Leftarrow)$ case of the application uses Lemma~\ref{proofintersection}.
  }
  \longVOnly{\todo[inline]{To do}
    Let $t$ be any lambda term with free variable $z_1...z_n$ seen as a LL proof derivation $t$ of a sequent $\sigma_1...\sigma_n\vdash\tau$ encoding a simple type. Then we show by induction on $t$ that there is a refinement \AxiomC{$\pi$}\UnaryInfC{$\bangm{u_1}B_1,...,\bangm{u_n}B_n\indvdash\1 A$}\DisplayProof of $t$ (i.e., s.t. $\underline \pi=t$) iff there is an intersection type derivation of
    $$z_1:[\llbracket \indfun{f}{B_1}\rrbracket\mid f:1{\rightarrow}\dom(u_1)],...,x_n:[\llbracket \indfun{f}{B_n}\rrbracket\mid f:1{\rightarrow}\dom(u_n)]\vdash t:\llbracket A\rrbracket\ .$$
    \begin{itemize}
    \item if $t=z_i$, then
      \begin{itemize}
      \item if $\pi =$\AxiomC{$u_i\circ g=\id$}\dashedLine\UnaryInfC{$\bangm{u_1}B_1,...,\bangm{u_n}B_n\indvdash\1 \indfun g{B_i}$}\DisplayProof, for some $g:1\rightarrow \dom(u_i)$, then
        \AxiomC{}\UnaryInfC{$z_1:[\llbracket \indfun{f}{B_1}\rrbracket\mid f:1{\rightarrow}\dom(u_1)],...,x_n:[\llbracket \indfun{f}{B_n}\rrbracket\mid f:1{\rightarrow}\dom(u_n)]\vdash t:\llbracket g{B_i}\rrbracket$}\DisplayProof is an axiom.
      \item Conversely, if $z_1:[\llbracket \indfun{f}{B_1}\rrbracket\mid f:1{\rightarrow}\dom(u_1)],...,x_n:[\llbracket \indfun{f}{B_n}\rrbracket\mid f:1{\rightarrow}\dom(u_n)]\vdash t:\llbracket g{B_i}\rrbracket$
        for $g:1\rightarrow \dom(u_i)$, then $u_i\circ g : 1\rightarrow 1$ can only be the identity, thus
        \AxiomC{}\UnaryInfC{$\bangm{u_1}B_1,...,\bangm{u_n}B_n\indvdash\1 \indfun g{B_i}$}\DisplayProof is correct.
      \end{itemize}
    \item The case $t=\lambda z_{n+1}.t'$ is immediate.
    \item The case $t=t_1t_2$ is using the base change $[\mathtt{cst}_x(\pi_2)\mid x\in\dom(u)]$ in one direction, and Lemma~\ref{proofintersection} in the other. \todo{to develop}
    \end{itemize}
  }
\end{proof}

This equivalence will be completed, in Section~\ref{sec:sem}, by the interpretation of $\IndLL$ in the Scott semantics of linear logic. However, it would be interesting to see how much one can push this correspondence at the syntactical level. The long version investigates how fixedpoint behave, but many other perspective could be addressed.

\section{Vertical equivalence and cut elimination}\label{sec:cut}
\shortVOnly{
  \begin{figure*}
    \def\ScoreOverhang{0pt}
    \longVOnly{
  }
  \begin{minipage}{\linewidth}
    \begin{small}
      \begin{center}
      \longVOnly{\AxiomC{}
        \dashedLine
        \UnaryInfC{$\id(A)\eqtype A$}
        \UnaryInfC{$\id(A)\indvdash {J} A$}
        \UnaryInfC{$\bangm{\id} \id(A)\indvdash {J} A$}
        \UnaryInfC{$\bangm{\id} \id(A), \bangm{\init}\init(B)\indvdash {J} A$}
        \AxiomC{}
        \dashedLine
        \UnaryInfC{$\id(B)\eqtype B$}
        \UnaryInfC{$\id(B)\indvdash {K} B$}
        \UnaryInfC{$\bangm{\id} \id(B)\indvdash {K} B$}
        \UnaryInfC{$\bangm{\init} \init(A), \bangm{\id} \id(B)\indvdash {K} B$}
        \BinaryInfC{$\bangm{\iota_1} A, \bangm{\iota_2} B\indvdash {K} A\withm{\iota_1}{\iota_2} B$}
        \AxiomC{}
        \dashedLine
        \UnaryInfC{$ h_2(\pb{[u,v]}{v}(B))\subtype \pb{[u,v]}{v}(B)$}
        \UnaryInfC{$\bangm{\pb{v}{[u,v]}} \pb{[u,v]}{v}(A)\subtype \bangm{\iota_2} B$}
        \subtypLine
        \BinaryInfC{$\bangm{\iota_1} A, \bangm{\pb{u}[u,v]} \pb{[u,v]}{v}(B)\indvdash {K} A\withm{\iota_1}{\iota_2} B$}
        \AxiomC{}
        \dashedLine
        \UnaryInfC{$ h_1(\pb{[u,v]}{u}(A))\subtype \pb{[u,v]}{u}(A)$}
        \UnaryInfC{$\bangm{\pb{u}{[u,v]}} \pb{[u,v]}{u}(A)\subtype \bangm{\iota_1} A$}
        \subtypLine
        \BinaryInfC{$\bangm{\pb{u}{[u,v]}} \pb{[u,v]}{u}(A), \bangm{\pb{u}[u,v]} \pb{[u,v]}{v}(B)\indvdash {K} A\withm{\iota_1}{\iota_2} B$}
\UnaryInfC{$\bangm u A, \bangm v B\indvdash I \bangm{[u,v]}(A\withm{\iota_1}{\iota_2} B)$}
        \UnaryInfC{$\bangm u A\otimes \bangm v B\indvdash I \bangm{[u,v]}(A\withm{\iota_1}{\iota_2} B)$}
        \DisplayProof\\[.5em]
        }
\AxiomC{}
        \dashedLine
        \UnaryInfC{$\indfun{\inv{(\pb{\id}{\id})}}{\indfun{\pb{\id}{\id}}{A}} \eqtype A$}
        \subtypLine
        \UnaryInfC{$\indfun{\inv{(\pb{\id}{\id})}}{\indfun{\pb{\id}{\id}}{A}} \indvdash J A$}
\UnaryInfC{$\indfun{\pb{\id}{\id}}{A}\withm{\pb{\id}{\id}}{\init} \indfun{\init}{B} \indvdash J A$}
        \UnaryInfC{$\bangm{\id}(A\withm{\id}{\init} \indfun{\init}{B}) \indvdash J A$}
        \AxiomC{}
        \dashedLine
        \UnaryInfC{$\indfun{\iota_1}{A\withm{\iota_1}{\iota_2} B} \subtype A\withm{\id}{\init} \indfun{\init}{B}$}
        \dashedLine
        \UnaryInfC{$\indfun{h_1}{\indfun{\pb{u}{[u,v]}}{A\withm{\iota_1}{\iota_2} B}} \subtype A\withm{\id}{\init} \indfun{\init}{B}$}
        \UnaryInfC{$\indfun{u}{\bangm{[u,v]}(A\withm{\iota_1}{\iota_2} B)} \subtype \bangm{\id}(A\withm{\id}{\init} \indfun{\init}{B})$}
        \subtypLine
        \BinaryInfC{$\indfun{u}{\bangm{[u,v]}(A\withm{\iota_1}{\iota_2} B)} \indvdash J A$}
        \UnaryInfC{$\bangm{[u,v]}(A\withm{\iota_1}{\iota_2} B) \indvdash I \bangm u A$}
        \AxiomC{}
        \dashedLine
        \UnaryInfC{$\indfun{\inv{(\pb{\id}{\id})}}{\indfun{\pb{\id}{\id}}{B}} \eqtype B$}
        \subtypLine
        \UnaryInfC{$\indfun{\inv{(\pb{\id}{\id})}}{\indfun{\pb{\id}{\id}}{B}} \indvdash K B$}
\UnaryInfC{$\indfun{\init}{A}\withm{\init}{\pb{\id}{\id}} \indfun{\pb{\id}{\id}}{B} \indvdash K B$}
        \UnaryInfC{$\bangm \id(\init(A)\withm{\init}{\id} B) \indvdash K B$}
        \AxiomC{}
        \dashedLine
        \UnaryInfC{$\indfun{\iota_2}{A\withm{\iota_1}{\iota_2} B} \subtype \indfun{\init}{A}\withm{\init}{\id} B$}
        \dashedLine
        \UnaryInfC{$\indfun{h_2}{\indfun{\pb{v}{[u,v]}}{A\withm{\iota_1}{\iota_2} B}} \subtype \indfun{\init}{A}\withm{\init}{\id} B$}
        \UnaryInfC{$\indfun{v}{\bangm{[u,v]}(A\withm{\iota_1}{\iota_2} B)} \subtype \bangm{\id}(\indfun{\init}{A}\withm{\init}{\id} B)$}
        \subtypLine
        \BinaryInfC{$\indfun{v}{\bangm{[u,v]}(A\withm{\iota_1}{\iota_2} B)} \indvdash J B$}
        \UnaryInfC{$\bangm{[u,v]}(A\withm{\iota_1}{\iota_2} B) \indvdash I \bangm v B$}
        \BinaryInfC{$\bangm{[u,v]}(A\withm{\iota_1}{\iota_2} B),\bangm{[u,v]}(A\withm{\iota_1}{\iota_2} B) \indvdash I \bangm u A\otimes \bangm v B$}
        \UnaryInfC{$\bangm{[u,v]}(A\withm{\iota_1}{\iota_2} B) \indvdash I \bangm u A\otimes \bangm v B$}
        \DisplayProof
      \end{center}
    \end{small}
  \end{minipage}

     \def\ScoreOverhang{4pt}
    \caption{Seely proof}\label{fig:Seely}
    \longVOnly{

  Exponential basic rules:\\
}
    \begin{small}
\AxiomC{$\pi_1$}
        \UnaryInfC{$\indfun{w}{\bangm{u_1}A_1},\dots,\indfun{w}{\bangm{u_n}A_n}\indvdash {J} B$}
        \UnaryInfC{$\bangm{u_1}A_1,\dots,\bangm{u_n}A_n\indvdash I \bangm w B$}
        \AxiomC{$\pi_2$}
        \UnaryInfC{$\bangm wB,\bangm wB\indvdash I\Gamma$}
        \UnaryInfC{$\bangm wB\indvdash I\Gamma$}
        \cutLine
        \BinaryInfC{$\bangm{u_1}A_1,\dots,\bangm{u_n}A_n\indvdash I\Gamma$}
        \DisplayProof
         $\quad\reduce\quad$
        \AxiomC{$\pi_1$}
        \UnaryInfC{$\indfun{w}{\bangm{u_1}A_1},\dots,\indfun{w}{\bangm{u_n}A_n}\indvdash {J} B$}
        \UnaryInfC{$\bangm{u_1}A_1,\dots,\bangm{u_n}A_n\indvdash I \bangm w B$}
        \AxiomC{$\pi_1$}
        \UnaryInfC{$\indfun{w}{\bangm{u_1}A_1},\dots,\indfun{w}{\bangm{u_n}A_n}\indvdash {J} B$}
        \UnaryInfC{$\bangm{u_1}A_1,\dots,\bangm{u_n}A_n\indvdash I \bangm w B$}
        \AxiomC{$\pi_2$}
        \UnaryInfC{$\bangm wB,\bangm wB\indvdash I\Gamma$}
        \cutLine
        \BinaryInfC{$\bangm{u_1}A_1,\dots,\bangm{u_n}A_n,\bangm wB\indvdash I\Gamma$}
        \cutLine
        \BinaryInfC{$\bangm{u_1}A_1,\bangm{u_1}A_1,\dots,\bangm{u_n}A_n,\bangm{u_n}A_n\indvdash I\Gamma$}
\UnaryInfC{$\svdots$}
        \UnaryInfC{$\bangm{u_1}A_1,\dots,\bangm{u_n}A_n\indvdash I\Gamma$}
        \DisplayProof
\end{small}
  
      \begin{small}
\AxiomC{$\pi_1$}
          \UnaryInfC{$\indfun{w}{\bangm{u_1}A_1},\dots,\indfun{w}{\bangm{u_n}A_n}\indvdash {J} B$}
          \UnaryInfC{$\bangm{u_1}A_1,\dots,\bangm{u_n}A_n\indvdash I \bangm w B$}
          \AxiomC{$\pi_2$}
          \UnaryInfC{$\indvdash I \Gamma$}
          \UnaryInfC{$\bangm{w}B\indvdash I \Gamma$}
          \cutLine
          \BinaryInfC{$\bangm{u_1}A_1,\dots,\bangm{u_n}A_n\indvdash I \Gamma$}
          \DisplayProof\qquad
$\reduce$
          \AxiomC{$\pi_2$}
          \UnaryInfC{$\indvdash I \Gamma$}
          \longVOnly{\UnaryInfC{$\bangm{u_1}A_1\indvdash I \Gamma$}}
          \UnaryInfC{$\svdots$}
          \UnaryInfC{$\bangm{u_1}A_1,\dots,\bangm{u_n}A_n\indvdash I \Gamma$}
          \DisplayProof\\
          \AxiomC{$\pi_1$}
          \UnaryInfC{$\indfun{w}{\bangm{u_1}A_1},\dots,\indfun{w}{\bangm{u_n}A_n}\indvdash {J} B$}
          \UnaryInfC{$\bangm{u_1}A_1,\dots,\bangm{u_n}A_n\indvdash I \bangm w B$}
          \AxiomC{$\pi_2$}
          \UnaryInfC{$\indfun{f}{B} \indvdash I \Gamma$}
          \AxiomC{$w\circ f = \id$}
          \BinaryInfC{$\bangm{w}B\indvdash I \Gamma$}
          \cutLine
          \BinaryInfC{$\bangm{u_1}A_1,\dots,\bangm{u_n}A_n\indvdash I \Gamma$}
          \DisplayProof\quad
$\reduce$\quad
          \AxiomC{$\pi_1$}
          \UnaryInfC{$\indfun{w}{\bangm{u_1}A_1},\dots,\indfun{w}{\bangm{u_n}A_n}\indvdash {J} B$}
          \functorialLine
          \UnaryInfC{$\indfun{f}{\indfun{w}{\bangm{u_1}A_1}},\dots,\indfun{f}{\indfun{w}{\bangm{u_n}A_n}}\indvdash I \indfun{f}{B}$}
          \AxiomC{$w\circ f = \id$}
          \subtypLine
          \BinaryInfC{$\bangm{u_1}A_1,\dots,\bangm{u_n}A_n\indvdash I \indfun{f}{B}$}
          \AxiomC{$\pi_2$}
          \UnaryInfC{$\indfun{f}{B} \indvdash I \Gamma$}
          \cutLine
          \BinaryInfC{$\bangm{u_1}A_1,\dots,\bangm{u_n}A_n\indvdash I \Gamma$}
          \DisplayProof
\AxiomC{$\pi_1$}
          \UnaryInfC{$\indfun{w}{\bangm{u_1}A_1},\dots,\indfun{w}{\bangm{u_n}A_n}\indvdash {J} B$}
          \UnaryInfC{$\bangm{u_1}A_1,\dots,\bangm{u_n}A_n\indvdash I \bangm w B$}
          \AxiomC{$\pi_2$}
          \UnaryInfC{$\indfun{v}{\bangm{w}B},\indfun{v}{\bangm{u'_1}A'_1},\dots,\indfun{v}{\bangm{u'_m}A'_m}\indvdash {K} C$}
          \UnaryInfC{$\bangm wB,\bangm{u'_1}A'_1,\dots,\bangm{u'_m}A'_m\indvdash I \bangm v C$}
          \cutLine
          \BinaryInfC{$\bangm{u_1}A_1,\dots,\bangm{u_n}A_n,\bangm{u'_1}A'_1,\dots,\bangm{u'_m}A'_m \indvdash I \bangm{v}C$}
          \DisplayProof\longVOnly{\\}\shortVOnly{\\[-3em]}
\begin{flushright}
          $\reduce$\qquad
          \AxiomC{$\pi_1$}
          \UnaryInfC{$\indfun{w}{\bangm{u_1}A_1},\dots,\indfun{w}{\bangm{u_n}A_n}\indvdash {J} B$}
          \UnaryInfC{$\bangm{u_1}A_1,\dots,\bangm{u_n}A_n\indvdash I \bangm w B$}
          \functorialLine
          \UnaryInfC{$\indfun{v}{\bangm{u_1}A_1}\dots,\indfun{v}{\bangm{u_n}A_n}\indvdash K \indfun{v}{\bangm w B}$}
          \AxiomC{$\pi_2$}
          \UnaryInfC{$\indfun{v}{\bangm{w}B},\indfun{v}{\bangm{u'_1}A'_1},\dots,\indfun{v}{\bangm{u'_m}A'_m}\indvdash K C$}
          \cutLine
          \BinaryInfC{$\indfun{v}{\bangm{u_1}A_1},\dots,\indfun{v}{\bangm{u_n}A_n},\indfun{v}{\bangm{u'_1}A'_1},\dots,\indfun{v}{\bangm{u'_m}A'_m}\indvdash K C$}
          \UnaryInfC{$\bangm{u_1}A_1,\dots,\bangm{u_n}A_n,\bangm{u'_1}A'_1,\dots,\bangm{u'_m}A'_m \indvdash I \bangm{v}C$}
          \DisplayProof
          \end{flushright}
\end{small}
    \shortVOnly{\vspace*{.5em}}

    \longVOnly{
    \medskip  
    Fixpoints basic rule:\\
      \begin{small}
        \begin{align*}
          &\AxiomC{$\pi_1$}
          \UnaryInfC{$\Gamma\indvdash I \indfun{f}{A[\mum{\id}X.A/X]}$}
          \UnaryInfC{$\Gamma\indvdash I\mum fX.A$}
          \AxiomC{$\pi_2$}
          \UnaryInfC{$\indfun{f}{A[\mum{\id}X.A/X]} \indvdash I \indfun{g}{B}$}
          \UnaryInfC{$\mum fX.A \indvdash I \indfun{(f; g)}{B}$}
          \cutLine
          \BinaryInfC{$\Gamma \indvdash I \indfun{(f; g)}{B}$}
          \DisplayProof
          & \quad\reduce\quad
          \AxiomC{$\pi_1$}
          \UnaryInfC{$\Gamma\indvdash I \indfun{f}{A[\mum{\id}X.A/X]}$}
          \AxiomC{$\pi_2$}
          \UnaryInfC{$\indfun{f}{A[\mum{\id}X.A/X]} \indvdash I \indfun{g}{B}$}
          \cutLine
          \BinaryInfC{$\Gamma \indvdash I \indfun{(f; g)}{B}$}
          \DisplayProof\\
          \end{align*}
      \end{small}
      
    Multiplicative and additive basic rules:\\
    }
      \begin{small}
      \longVOnly{
\AxiomC{$\pi_1$}
          \UnaryInfC{$\Gamma\indvdash I A$}
          \AxiomC{$\pi_2$}
          \UnaryInfC{$\Delta\indvdash I B$}
          \BinaryInfC{$\Gamma,\Delta\indvdash I A\otimes B$}
          \AxiomC{$\pi_3$}
          \UnaryInfC{$A,B\indvdash I \Xi$}
          \UnaryInfC{$A\otimes B\indvdash I \Xi$}
          \cutLine
          \BinaryInfC{$\Gamma,\Delta\indvdash I \Xi$}
          \DisplayProof
$\quad\reduce\quad$
          \AxiomC{$\pi_1$}
          \UnaryInfC{$\Gamma\indvdash I A$}
          \AxiomC{$\pi_2$}
          \UnaryInfC{$\Delta\indvdash I B$}
          \AxiomC{$\pi_3$}
          \UnaryInfC{$A,B\indvdash I \Xi$}
          \cutLine
          \BinaryInfC{$A,\Delta\indvdash I \Xi$}
          \cutLine
          \BinaryInfC{$\Gamma,\Delta\indvdash I \Xi$}
          \DisplayProof\\
          }
          \AxiomC{$\pi_1$}
          \UnaryInfC{$\Xi\indvdash {I} \indfun{\inv{i}}{A}$}
          \AxiomC{$\emptyset\vdash B\ \mathtt{def}$}
          \BinaryInfC{$\Xi\indvdash {I} A\oplusm{i}{\init} B$}
          \AxiomC{$\pi_2$}
          \UnaryInfC{$A\indvdash {I'} \indfun{i}{\Gamma}$}
          \AxiomC{$\pi_3$}
          \UnaryInfC{$B \indvdash \emptyset \indfun{\init}{\Gamma}$}
          \BinaryInfC{$A\oplusm{\id}{\init} B\indvdash {I} \Gamma$}
          \cutLine
          \BinaryInfC{$\Xi\indvdash {I} \Gamma$}
          \DisplayProof
          $\quad\reduce\quad$
          \AxiomC{$\pi_1$}
          \UnaryInfC{$\Xi\indvdash {I} \indfun{\inv{i}}{A}$}
          \AxiomC{$\pi_2$}
          \UnaryInfC{$A\indvdash {I'} \indfun{i}{\Gamma}$}
          \functorialLine
          \UnaryInfC{$\indfun{\inv{i}}{A}\indvdash I \indfun{\inv{i}}{\indfun{i}{\Gamma}}$}
          \subtypLine
          \UnaryInfC{$\indfun{\inv{i}}{A}\indvdash I \Gamma$}
          \cutLine
          \BinaryInfC{$\Xi\indvdash {I} \Gamma$}
          \DisplayProof
          \shortVOnly{\vspace*{1em}}
          \longVOnly{\\
          \AxiomC{$\pi_1$}
          \UnaryInfC{$\Xi\indvdash {J} B$}
          \AxiomC{$\emptyset\vdash A\ \mathtt{def}$}
          \BinaryInfC{$\Xi\indvdash {J} A\oplusm{\init}{\id} B$}
          \AxiomC{$\pi_2$}
          \UnaryInfC{$A\indvdash \emptyset \indfun{\init}{\Gamma}$}
          \AxiomC{$\pi_3$}
          \UnaryInfC{$B \indvdash J \Gamma$}
          \BinaryInfC{$A\oplusm{\init}{\id} B\indvdash {J} \Gamma$}
          \cutLine
          \BinaryInfC{$\Xi\indvdash {J} \Gamma$}
          \DisplayProof
          $ \quad\reduce\quad$
          \AxiomC{$\pi_1$}
          \UnaryInfC{$\Xi\indvdash {J} B$}
          \AxiomC{$\pi_3$}
          \UnaryInfC{$B\indvdash {J} \Gamma$}
          \cutLine
          \BinaryInfC{$\Xi\indvdash {J} \Gamma$}
          \DisplayProof
          }
\end{small}
    \longVOnly{
    \bigskip
    
    Additives non-trivial commutation rules:\\
    }
      \begin{small}
\longVOnly{
          \AxiomC{$\emptyset\vdash\Gamma,A\ \mathtt{def}$}
          \UnaryInfC{$\0\indvdash \emptyset \Gamma,A$}
          \AxiomC{$\pi_1$}
          \UnaryInfC{$A\indvdash \emptyset \Xi$}
          \cutLine
          \BinaryInfC{$\0\indvdash \emptyset \Xi,\Gamma$}
          \DisplayProof
$\quad\reduce\quad$
          \AxiomC{$\emptyset\vdash\Xi,\Gamma\ \mathtt{def}$}
          \UnaryInfC{$\0\indvdash \emptyset \Xi,\Gamma$}
          \DisplayProof\quad\quad
          viable since $A\indvdash \emptyset \Xi$ implies $\emptyset\vdash\Xi\ \mathtt{def}$\\
          }
          \AxiomC{$\pi_1$}
          \UnaryInfC{$B\indvdash I \indfun{i}{\Gamma},\indfun{i}{A}$}
          \AxiomC{$\pi_2$}
          \UnaryInfC{$C\indvdash J \indfun{j}{\Gamma},\indfun{j}{A}$}
          \BinaryInfC{$B\oplusm{i}{j} C\indvdash {W} \Gamma,A$}
          \AxiomC{$\pi_3$}
          \UnaryInfC{$A\indvdash {W} \Xi$}
          \cutLine
          \BinaryInfC{$B\oplusm{i}{j} C\indvdash {W} \Gamma,\Xi$}
          \DisplayProof
\quad$\reduce$\quad
          \AxiomC{$\pi_1$}
          \UnaryInfC{$B\indvdash I \indfun{i}{\Gamma},\indfun{i}{A}$}
          \AxiomC{$\pi_3$}
          \UnaryInfC{$A\indvdash W \Xi$}
          \functorialLine
          \UnaryInfC{$\indfun{i}{A}\indvdash I \indfun{i}{\Xi}$}
          \cutLine
          \BinaryInfC{$B\indvdash I \indfun{i}{\Gamma},\indfun{i}{\Xi}$}
          \AxiomC{$\pi_2$}
          \UnaryInfC{$C\indvdash I \indfun{j}{\Gamma},\indfun{j}{A}$}
          \AxiomC{$\pi_3$}
          \UnaryInfC{$ A\indvdash W \Xi$}
          \functorialLine
          \UnaryInfC{$ \indfun{j}{A}\indvdash J \indfun{j}{\Xi}$}
          \cutLine
          \BinaryInfC{$C\indvdash J \indfun{j}{\Gamma},\indfun{j}{j}{\Xi}$}
          \BinaryInfC{$B\oplusm{i}{j} C\indvdash {W} \Gamma,\Xi$}
          \DisplayProof    
\end{small}

    \longVOnly{
    Exponential trivial commutation rules:\fnote{todo}\\

    Multiplicative and additives trivial commutation rules:
    \begin{small}
      \begin{align*}
        \AxiomC{$\pi_1$}
        \UnaryInfC{$\Xi\indvdash I a$}
        \AxiomC{$a\in\mathtt{variable}(I)$}
        \UnaryInfC{$a\indvdash I a$}
        \cutLine
        \BinaryInfC{$\Xi\indvdash I a$}
        \DisplayProof
        & \reduce
        \AxiomC{$\pi_1$}
        \UnaryInfC{$\Xi\indvdash I a$}
        \DisplayProof\\
        \AxiomC{$\pi_1$}
        \UnaryInfC{$\Xi\indvdash I \bot$}
        \AxiomC{\vphantom{A}}
        \UnaryInfC{$\bot\indvdash I $}
        \cutLine
        \BinaryInfC{$\Xi\indvdash I $}
        \DisplayProof
        & \reduce
        \AxiomC{$\pi_1$}
        \UnaryInfC{$\Xi\indvdash I $}
        \DisplayProof\\
        \AxiomC{$\pi_1$}
        \UnaryInfC{$\Xi\indvdash I A$}
        \AxiomC{$\pi_2$}
        \UnaryInfC{$\Gamma,A\indvdash I B$}
        \AxiomC{$\pi_3$}
        \UnaryInfC{$\Delta\indvdash I C$}
        \BinaryInfC{$\Gamma,A,\Delta\indvdash I B\otimes C$}
        \cutLine
        \BinaryInfC{$\Xi,\Gamma,\Delta\indvdash I B\otimes C$}
        \DisplayProof
        & \reduce
        \AxiomC{$\pi_1$}
        \UnaryInfC{$\Xi\indvdash I A$}
        \AxiomC{$\pi_2$}
        \UnaryInfC{$\Gamma,A\indvdash I B$}
        \cutLine
        \BinaryInfC{$\Xi,\Gamma\indvdash I B$}
        \AxiomC{$\pi_3$}
        \UnaryInfC{$\Delta\indvdash I C$}
        \BinaryInfC{$\Xi,\Gamma,\Delta\indvdash I B\otimes C$}
        \DisplayProof\\
        \AxiomC{$\pi_1$}
        \UnaryInfC{$\Xi\indvdash I A$}
        \AxiomC{$\pi_2$}
        \UnaryInfC{$\Gamma\indvdash I B$}
        \AxiomC{$\pi_3$}
        \UnaryInfC{$\Delta,A\indvdash I C$}
        \BinaryInfC{$\Gamma,\Delta,A\indvdash I B\otimes C$}
        \cutLine
        \BinaryInfC{$\Xi,\Gamma,\Delta\indvdash I B\otimes C$}
        \DisplayProof
        & \reduce
        \AxiomC{$\pi_2$}
        \UnaryInfC{$\Gamma\indvdash I B$}
        \AxiomC{$\pi_1$}
        \UnaryInfC{$\Xi\indvdash I A$}
        \AxiomC{$\pi_3$}
        \UnaryInfC{$\Delta,A\indvdash I C$}
        \cutLine
        \BinaryInfC{$\Xi,\Delta\indvdash I C$}
        \BinaryInfC{$\Xi,\Gamma,\Delta\indvdash I B\otimes C$}
        \DisplayProof\\
        \AxiomC{$\pi_1$}
        \UnaryInfC{$\Xi\indvdash I A$}
        \AxiomC{$\pi_2$}
        \UnaryInfC{$\Gamma,A\indvdash I B,C$}
        \UnaryInfC{$\Gamma,A\indvdash I B\parr C$}
        \cutLine
        \BinaryInfC{$\Xi,\Gamma\indvdash I B\parr C$}
        \DisplayProof
        & \reduce
        \AxiomC{$\pi_1$}
        \UnaryInfC{$\Xi\indvdash I A$}
        \AxiomC{$\pi_2$}
        \UnaryInfC{$\Gamma,A\indvdash I B,C$}
        \cutLine
        \BinaryInfC{$\Xi,\Gamma,A\indvdash I B,C$}
        \UnaryInfC{$\Xi,\Gamma\indvdash I B\parr C$}
        \DisplayProof\\
        \AxiomC{$\pi_1$}
        \UnaryInfC{$\Xi\indvdash I A$}
        \AxiomC{$\pi_2$}
        \UnaryInfC{$\Gamma,A\indvdash I B$}
        \AxiomC{$0\vdash C\ \mathtt{def}$}
        \BinaryInfC{$\Gamma,A\indvdash I B\oplusm{\id}{\init} C$}
        \cutLine
        \BinaryInfC{$\Gamma,\Xi\indvdash I B\oplusm{\id}{\init} C$}
        \DisplayProof
        & \reduce
        \AxiomC{$\pi_1$}
        \UnaryInfC{$\Xi\indvdash I A$}
        \AxiomC{$\pi_2$}
        \UnaryInfC{$\Gamma,A\indvdash I B$}
        \cutLine
        \BinaryInfC{$\Gamma,\Xi\indvdash I B$}
        \AxiomC{$0\vdash C\ \mathtt{def}$}
        \BinaryInfC{$\Gamma,\Xi\indvdash I B\oplusm\id\init C$}
        \DisplayProof\\
        \AxiomC{$\pi_1$}
        \UnaryInfC{$\Xi\indvdash I A$}
        \AxiomC{$\pi_2$}
        \UnaryInfC{$\Gamma,A\indvdash I C$}
        \AxiomC{$0\vdash B\ \mathtt{def}$}
        \BinaryInfC{$\Gamma,A\indvdash I B\oplusm\init\id C$}
        \cutLine
        \BinaryInfC{$\Gamma,\Xi\indvdash I B\oplusm\init\id C$}
        \DisplayProof
        & \reduce
        \AxiomC{$\pi_1$}
        \UnaryInfC{$\Xi\indvdash I A$}
        \AxiomC{$\pi_2$}
        \UnaryInfC{$\Gamma,A\indvdash I C$}
        \cutLine
        \BinaryInfC{$\Gamma,\Xi\indvdash I C$}
        \AxiomC{$0\vdash B\ \mathtt{def}$}
        \BinaryInfC{$\Gamma,\Xi\indvdash I B\oplusm\init\id C$}
        \DisplayProof
      \end{align*}
    \end{small}
}
     \caption{Cut-elimination procedure. For space reasons, we only present  rules concerning exponentials and additives. In the additive case the base change plays an active role. (see the long version for a full presentation).}\label{fig:cut-elim}
  \end{figure*}
}

(Multi)set-based intersection types are known not to have a cut elimination procedure, but just a subject reduction theorem. This means that there can exist several ways to correctly type the $\beta$-reduct of a intersection-typed term, without canonical choice. This phenomenon depends on the strict commutativity of the intersection constructor. In order to recover a deterministic procedure, one has to add additional information such as rigid types~\cite{MazzaPV18}, ``threads''~\cite{Vial17}, explicit permutations~\cite{ol:intdist} or keyed maps~\cite{BreuvartL18}.

$\IndLL$ adds enough information through its indexes, determining a proper locally-deterministic\footnote{by this we mean deterministic once the cut has been selected.} cut elimination procedure. However, the resulting procedure is not confluent anymore because of the choice locally made during cut elimination. This is due to the interaction between cut elimination and proof-relevant subtyping. In order to recover confluence, we shall need to quotient proofs up to an appropriate congruence induced by the subtyping.

\subsection{Vertical equivalence}

We will present a reasonable equivalence on proofs that will make us recover the confluence of cut elimination. It is an open question to know which equivalence is the smallest sufficient one. But there is one particular equivalence relation that is sufficient and very canonical: the one that consists of making subtyping proof-irrelevant.

Indeed, the arbitrary choices made along the reduction concern either the choice of a subtyping proof, or the choice of base naming. Both of which will be emptied of their proof-relevance content once the following equivalence is considered. We call it \emph{vertical equivalence} in reference to that appearing in~\cite{MellZeil13} which is of similar nature.

\begin{definition}[Proof-irrelevance]\label{def:proof-i}
  We quotient proofs of the same sequent by the congruence freely generated by the following rules.
  \def\ScoreOverhang{0pt}
  \begin{align*}
    \bottomAlignProof
    \AxiomC{$\pi$\vphantom{A}}
    \UnaryInfC{$\Gamma\indvdash I A$}
    \AxiomC{$\pi_1$\vphantom{A}}
    \UnaryInfC{$A\subtype B$}
    \subtypLine
    \BinaryInfC{$\Gamma\indvdash I B$}
    \DisplayProof
    &\raiseRel{\eqtype}
    \bottomAlignProof
    \AxiomC{$\pi$\vphantom{A}}
    \UnaryInfC{$\Gamma\indvdash I A$}
    \AxiomC{$\pi_2$\vphantom{A}}
    \UnaryInfC{$A\subtype B$}
    \subtypLine
    \BinaryInfC{$\Gamma\indvdash I B$}
    \DisplayProof
    &\hspace{3em}
    \bottomAlignProof
    \AxiomC{$\pi$\vphantom{A}}
    \UnaryInfC{$\Gamma\indvdash I A$}
    \AxiomC{\vphantom{A}}
    \UnaryInfC{$A\eqtype A$}
    \subtypLine
    \BinaryInfC{$\Gamma\indvdash I A$}
    \DisplayProof
    &\raiseRel{\eqtype}
    \bottomAlignProof
    \AxiomC{$\pi$\vphantom{A}}
    \UnaryInfC{$\Gamma\indvdash I A$}
    \DisplayProof
  \end{align*}
  \def\ScoreOverhang{4pt}
  \begin{align*}
    \bottomAlignProof
    \AxiomC{$\pi$\vphantom{A}}
    \UnaryInfC{$\Gamma\indvdash I A$}
    \AxiomC{$\pi_1$\vphantom{A}}
    \UnaryInfC{$A\subtype B$}
    \subtypLine
    \BinaryInfC{$\Gamma\indvdash I B$}
    \AxiomC{$\pi_2$\vphantom{A}}
    \UnaryInfC{$B\subtype C$}
    \subtypLine
    \BinaryInfC{$\Gamma\indvdash I C$}
    \DisplayProof
    &\raiseRel{\eqtype}
    \bottomAlignProof
    \AxiomC{$\pi$\vphantom{A}}
    \UnaryInfC{$\Gamma\indvdash I A$}
    \AxiomC{$\pi_1$\vphantom{A}}
    \UnaryInfC{$A\subtype B$}
    \AxiomC{$\pi_2$\vphantom{A}}
    \UnaryInfC{$B\subtype C$}
    \dashedLine
    \BinaryInfC{$A\subtype C$}
    \subtypLine
    \BinaryInfC{$\Gamma\indvdash I C$}
    \DisplayProof
  \end{align*}
  \begin{align*}
    \bottomAlignProof
    \AxiomC{$\pi_1$\vphantom{A}}
    \UnaryInfC{$\Gamma\indvdash I A$}
    \AxiomC{$\pi_1'$\vphantom{A}}
    \UnaryInfC{$A\subtype B$}
    \subtypLine
    \BinaryInfC{$\Gamma\indvdash I B$}
    \AxiomC{$\pi_2$\vphantom{A}}
    \UnaryInfC{$B\indvdash I \Delta$}
\BinaryInfC{$\Gamma\indvdash I \Delta$}
    \DisplayProof
    &\raiseRel{\eqtype}
    \bottomAlignProof
    \AxiomC{$\pi_1$\vphantom{A}}
    \UnaryInfC{$\Gamma\indvdash I A$}
    \AxiomC{$\pi_2$\vphantom{A}}
    \UnaryInfC{$B\indvdash I \Delta$}
    \AxiomC{$\pi_1'$\vphantom{A}}
    \UnaryInfC{$A\subtype B$}
    \subtypLine
    \BinaryInfC{$A\indvdash I \Delta$}
\BinaryInfC{$\Gamma\indvdash I \Delta$}
    \DisplayProof
  \end{align*}
\end{definition}

Remark: only the second and third rules of Definition \ref{def:proof-i} are absolutely necessary. The last one is important for the confluence proof. The first one make proofs much easier. Proving some of the results that follow without it would be much more technically challenging.

\begin{lemma}\label{lemma:proof_psedo_funct}
  For any proof $\pi$ we can infer the following equivalences, up to the congruence of Definition \ref{def:proof-i}.
 \[
  \bottomAlignProof
  \AxiomC{$\indfun{f_1}{\!\!\cdots\! \indfun{f_n}{\pi}}$\vphantom{A}}
  \functorialLine
  \UnaryInfC{$\indvdash J \indfun{f_1}{\!\!\cdots\! \indfun{f_n}{\Gamma}}$}
  \AxiomC{$f_1;\!...;f_n=g_1;\!...;g_k$}
  \dashedLine
  \UnaryInfC{$\indfun{f_1}{\!\!\cdots\! \indfun{f_n}{\Gamma}}\subtype \indfun{g_1}{\!\!\cdots\! \indfun{g_k}{\Gamma}}$}
  \subtypLine
  \BinaryInfC{$\indvdash J \indfun{g_1}{\!\!\cdots\! \indfun{g_k}{\Gamma}}$}
  \DisplayProof
  \raiseRel{\eqtype}
  \bottomAlignProof
  \AxiomC{$\indfun{g_1}{\!\!\cdots\! \indfun{g_k}{\pi}}$\vphantom{A}}
  \functorialLine
  \UnaryInfC{$\indvdash J \indfun{g_1}{\!\!\cdots\! \indfun{g_k}{\Gamma}}$}
  \DisplayProof
  \]
  \[
  \bottomAlignProof
  \AxiomC{$\pi$\vphantom{A}}
  \UnaryInfC{$\Gamma\indvdash I B$}
  \AxiomC{$\rho$\vphantom{A}}
  \UnaryInfC{$B\subtype B'$}
  \subtypLine
  \BinaryInfC{$\Gamma\indvdash I B'$}
  \functorialLine
  \UnaryInfC{$\indfun{f}{\Gamma}\indvdash J \indfun{f}{B'}$}
  \DisplayProof
  \raiseRel{\eqtype}
  \bottomAlignProof
  \AxiomC{$\pi$\vphantom{A}}
  \UnaryInfC{$\Gamma\indvdash I B$}
  \functorialLine
  \UnaryInfC{$\indfun{f}{\Gamma}\indvdash J \indfun{f}{B}$}
  \AxiomC{$\rho$\vphantom{A}}
  \UnaryInfC{$B\subtype B'$}
  \dashedLine
  \UnaryInfC{$\indfun{f}{B}\subtype \indfun{f}{B'}$}
  \subtypLine
  \BinaryInfC{$\indfun{f}{\Gamma}\indvdash J \indfun{f}{B'}$}
  \DisplayProof
  \]
\end{lemma}
\begin{proof}
 By induction on the skeleton $\underline\pi$ of~$\pi$. The proof is a lengthy application of the IH by exploiting our proof equivalence to infer the appropriate shape of the derivations. \\
\longVOnly{   By induction on the skeleton $\underline\pi$ of~$\pi$.\\
  Many cases are trivial inductions. The difficult cases are those for which both transformations are nontrivial such as the dereliction, the promotion (we treat the one premise version of it, wlog) or the introduction of the product.
  
  Due to the number of cases, we won't present them all, but they are either trivial induction or follow the same pattern.

  Let $u$ and $v$ be functions, both of codomain $I$. We define $e_v$ to be the isomorphism between the iterated pullback of $v$ along each $f_i$ and the iterated pullback of $v$ along each $g_i$.
  
   \scalebox{0.5}{\parbox{1.05\linewidth}{  \begin{align*}
        &
        \AxiomC{$\pi$\vphantom{A}}
        \UnaryInfC{$\indfun{v}{\bangm{u}A}\indvdash K B$}
        \UnaryInfC{$\bangm{u}A\indvdash I \bangm{v}B$}
        \functorialLine
        \UnaryInfC{$\vdots$}
        \functorialLine
        \UnaryInfC{$\indfun{f_1}{\cdots \indfun{f_n}{\bangm{u}A}} \indvdash J \indfun{f_1}{\cdots \indfun{f_n}{\bangm{v}B}}$}
        \AxiomC{$g_1;\dots,g_k=f_1;\dots;f_n$}
        \dashedLine
        \UnaryInfC{$\indfun{g_1}{\cdots \indfun{g_k}{\bangm{u}A}}\eqtype \indfun{f_1}{\cdots \indfun{f_n}{\bangm{u}A}}$}
        \subtypLine
        \BinaryInfC{$\indfun{g_1}{\cdots \indfun{g_k}{\bangm{u}A}} \indvdash J \indfun{f_1}{\cdots \indfun{f_n}{\bangm{v}B}}$}
        \AxiomC{$f_1;\dots;f_n=g_1;\dots,g_k$}
        \dashedLine
        \UnaryInfC{$\indfun{f_1}{\cdots \indfun{f_n}{\bangm{v}B}}\eqtype \indfun{g_1}{\cdots \indfun{g_k}{\bangm{v}B}}$}
        \subtypLine
        \BinaryInfC{$\indfun{g_1}{\cdots \indfun{g_k}{\bangm{u}A}}\indvdash J \indfun{g_1}{\cdots \indfun{g_k}{\bangm{v}B}}$}
        \DisplayProof\\
        &=
        \AxiomC{$\pi$\vphantom{A}}
        \UnaryInfC{$\indfun{v}{\bangm{u}A}\indvdash K B$}
        \functorialLine
        \UnaryInfC{$\vdots$}
        \functorialLine
        \UnaryInfC{$\indfun{\pb{{f_1}}{\tpb{v}{f_n}{\cdots}{f_2}}}{\cdots \indfun{\pb{{f_n}}{v}}{\indfun{v}{\bangm{u}A}}} \indvdash L \indfun{\pb{{f_1}}{\tpb{v}{f_n}{\cdots}{f_2}}}{\cdots \indfun{\pb{{f_n}}{v}}{B}}$}
        \AxiomC{$\tpb{v}{f_n}{\cdots}{f_1};f_1\dots;f_n=\pb{{f_1}}{\tpb{v}{f_n}{\cdots}{f_2}};\dots;\pb{{f_n}}{v};v$}
        \subtypLine
        \BinaryInfC{$\indfun{\tpb{v}{f_n}{\cdots}{f_1}}{\indfun{f_1}{\cdots \indfun{f_n}{\bangm{u}A}}}\indvdash L \indfun{\pb{{f_1}}{\tpb{v}{f_n}{\cdots}{f_2}}}{\cdots \indfun{\pb{{f_n}}{v}}{B}}$}
        \UnaryInfC{$\indfun{f_1}{\cdots \indfun{f_n}{\bangm{u}A}} \indvdash J \indfun{f_1}{\cdots \indfun{f_n}{\bangm{v}B}}$}
        \AxiomC{$g_1;\dots,g_k=f_1;\dots;f_n$}
        \dashedLine
        \UnaryInfC{$\indfun{g_1}{\cdots \indfun{g_k}{\bangm{u}A}}\eqtype \indfun{f_1}{\cdots \indfun{f_n}{\bangm{u}A}}$}
        \subtypLine
        \BinaryInfC{$\indfun{g_1}{\cdots \indfun{g_k}{\bangm{u}A}}\indvdash J \indfun{f_1}{\cdots \indfun{f_n}{\bangm{v}B}} $}
        \AxiomC{$f_1;\dots;f_n=g_1;\dots,g_k$}
        \dashedLine
        \UnaryInfC{$\indfun{f_1}{\cdots \indfun{f_n}{\bangm{v}B}}\subtype \indfun{g_1}{\cdots \indfun{g_k}{\bangm{v}B}}$}
        \subtypLine
        \BinaryInfC{$\indfun{g_1}{\cdots \indfun{g_k}{\bangm{u}A}}\indvdash J \indfun{g_1}{\cdots \indfun{g_k}{\bangm{v}B}}$}
        \DisplayProof\\
        &=
        \AxiomC{$\pi$\vphantom{A}}
        \UnaryInfC{$\indfun{v}{\bangm{u}A}\indvdash K B$}
        \functorialLine
        \UnaryInfC{$\vdots$}
        \functorialLine
        \UnaryInfC{$\indfun{\pb{{f_1}}{\tpb{v}{f_n}{\cdots}{f_2}}}{\cdots \indfun{\pb{{f_n}}{v}}{\indfun{v}{\bangm{u}A}}} \indvdash L \indfun{\pb{{f_1}}{\tpb{v}{f_n}{\cdots}{f_2}}}{\cdots \indfun{\pb{{f_n}}{v}}{B}}$}
        \AxiomC{$\tpb{v}{f_n}{\cdots}{f_1};f_1\dots;f_n=\pb{{f_1}}{\tpb{v}{f_n}{\cdots}{f_2}};\dots;\pb{{f_n}}{v};v$}
        \subtypLine
        \BinaryInfC{$\indfun{\tpb{v}{f_n}{\cdots}{f_1}}{\indfun{f_1}{\cdots \indfun{f_n}{\bangm{u}A}}}\indvdash L \indfun{\pb{{f_1}}{\tpb{v}{f_n}{\cdots}{f_2}}}{\cdots \indfun{\pb{{f_n}}{v}}{B}}$}
        \AxiomC{$g_1;\dots,g_k=f_1;\dots;f_n$}
        \dashedLine
        \UnaryInfC{$\indfun{g_1}{\cdots \indfun{g_k}{\bangm{u}A}}\eqtype \indfun{f_1}{\cdots \indfun{f_n}{\bangm{u}A}}$}
        \dashedLine
        \UnaryInfC{$\indfun{\tpb{v}{f_n}{\cdots}{f_1}}{\indfun{g_1}{\cdots \indfun{g_k}{\bangm{u}A}}}\eqtype \indfun{\tpb{v}{f_n}{\cdots}{f_1}}{\indfun{f_1}{\cdots \indfun{f_n}{\bangm{u}A}}}$}
        \subtypLine
        \BinaryInfC{$\indfun{\tpb{v}{f_n}{\cdots}{f_1}}{\indfun{g_1}{\cdots \indfun{g_k}{\bangm{u}A}}}\indvdash L \indfun{\pb{{f_1}}{\tpb{v}{f_n}{\cdots}{f_2}}}{\cdots \indfun{\pb{{f_n}}{v}}{B}}$}        
        \UnaryInfC{$\indfun{g_1}{\cdots \indfun{g_k}{\bangm{u}A}} \indvdash J \indfun{f_1}{\cdots \indfun{f_n}{\bangm{v}B}}$}
        \AxiomC{$f_1;\dots;f_n=g_1;\dots,g_k$}
        \dashedLine
        \UnaryInfC{$\indfun{f_1}{\cdots \indfun{f_n}{\bangm{v}B}}\subtype \indfun{g_1}{\cdots \indfun{g_k}{\bangm{v}B}}$}
        \subtypLine
        \BinaryInfC{$\indfun{g_1}{\cdots \indfun{g_k}{\bangm{u}A}}\indvdash J \indfun{g_1}{\cdots \indfun{g_k}{\bangm{v}B}}$}
        \DisplayProof\\
           &\eqtype
        \AxiomC{$\pi$\vphantom{A}}
        \UnaryInfC{$\indfun{v}{\bangm{u}A}\indvdash K B$}
        \functorialLine
        \UnaryInfC{$\vdots$}
        \functorialLine
        \UnaryInfC{$\indfun{\pb{{f_1}}{\tpb{v}{f_n}{\cdots}{f_2}}}{\cdots \indfun{\pb{{f_n}}{v}}{\indfun{v}{\bangm{u}A}}} \indvdash L \indfun{\pb{{f_1}}{\tpb{v}{f_n}{\cdots}{f_2}}}{\cdots \indfun{\pb{{f_n}}{v}}{B}}$}
        \AxiomC{$\tpb{v}{f_n}{\cdots}{f_1};g_1\dots;g_k=\pb{{f_1}}{\tpb{v}{f_n}{\cdots}{f_2}};\dots;\pb{{f_n}}{v};v$}
        \subtypLine
        \BinaryInfC{$\indfun{\tpb{v}{f_n}{\cdots}{f_1}}{\indfun{g_1}{\cdots \indfun{g_k}{\bangm{u}A}}}\indvdash L \indfun{\pb{{f_1}}{\tpb{v}{f_n}{\cdots}{f_2}}}{\cdots \indfun{\pb{{f_n}}{v}}{B}}$}        
        \UnaryInfC{$\indfun{g_1}{\cdots \indfun{g_k}{\bangm{u}A}} \indvdash J \indfun{f_1}{\cdots \indfun{f_n}{\bangm{v}B}}$}
        \AxiomC{$f_1;\dots;f_n=g_1;\dots,g_k$}
        \dashedLine
        \UnaryInfC{$\indfun{f_1}{\cdots \indfun{f_n}{\bangm{v}B}}\subtype \indfun{g_1}{\cdots \indfun{g_k}{\bangm{v}B}}$}
        \subtypLine
        \BinaryInfC{$\indfun{g_1}{\cdots \indfun{g_k}{\bangm{u}A}}\indvdash J \indfun{g_1}{\cdots \indfun{g_k}{\bangm{v}B}}$}
        \DisplayProof\\
        &=
        \AxiomC{$\pi$\vphantom{A}}
        \UnaryInfC{$\indfun{v}{\bangm{u}A}\indvdash K B$}
        \functorialLine
        \UnaryInfC{$\vdots$}
        \functorialLine
        \UnaryInfC{$\indfun{\pb{{f_1}}{\tpb{v}{f_n}{\cdots}{f_2}}}{\cdots \indfun{\pb{{f_n}}{v}}{\indfun{v}{\bangm{u}A}}} \indvdash L \indfun{\pb{{f_1}}{\tpb{v}{f_n}{\cdots}{f_2}}}{\cdots \indfun{\pb{{f_n}}{v}}{B}}$}
        \AxiomC{$\tpb{v}{f_n}{\cdots}{f_1};g_1\dots;g_k=\pb{{f_1}}{\tpb{v}{f_n}{\cdots}{f_2}};\dots;\pb{{f_n}}{v};v$}
        \subtypLine
        \BinaryInfC{$\indfun{\tpb{v}{f_n}{\cdots}{f_1}}{\indfun{g_1}{\cdots \indfun{g_k}{\bangm{u}A}}}\indvdash L \indfun{\pb{{f_1}}{\tpb{v}{f_n}{\cdots}{f_2}}}{\cdots \indfun{\pb{{f_n}}{v}}{B}}$}        
        \functorialLine
        \UnaryInfC{$\indfun{e_v}{\indfun{\tpb{v}{f_n}{\cdots}{f_1}}{\indfun{g_1}{\cdots \indfun{g_k}{\bangm{u}A}}}}\indvdash {L'} \indfun{e_v}{\indfun{\pb{{f_1}}{\tpb{v}{f_n}{\cdots}{f_2}}}{\cdots \indfun{\pb{{f_n}}{v}}{B}}}$}        
        \AxiomC{$e_v;\pb{{f_1}}{\tpb{v}{f_n}{\cdots}{f_2}};\dots;\pb{{f_n}}{v}=\pb{{g_1}}{\tpb{v}{g_k}{\cdots}{g_2}};\dots,\pb{{g_k}}{v}$}
        \subtypLine
        \BinaryInfC{$\indfun{e_v}{\indfun{\tpb{v}{f_n}{\cdots}{f_1}}{\indfun{g_1}{\cdots \indfun{g_k}{\bangm{u}A}}}}\indvdash {L'} \indfun{\pb{{g_1}}{\tpb{v}{g_k}{\cdots}{g_2}}}{\cdots \indfun{\pb{{g_k}}{v}}{B}}$}
        \AxiomC{$\tpb{v}{g_k}{\cdots}{g_1}=e_v;\tpb{v}{f_n}{\cdots}{f_1}$}
        \subtypLine
        \BinaryInfC{$\indfun{\tpb{v}{g_k}{\cdots}{g_1}}{\indfun{g_1}{\cdots \indfun{g_k}{\bangm{u}A}}}\indvdash {L'} \indfun{\pb{{g_1}}{\tpb{v}{g_k}{\cdots}{g_2}}}{\cdots \indfun{\pb{{g_k}}{v}}{B}}$}
        \UnaryInfC{$\indfun{g_1}{\cdots \indfun{g_k}{\bangm{u}A}}\indvdash J \indfun{g_1}{\cdots \indfun{g_k}{\bangm{v}B}}$}
        \DisplayProof\\
      \end{align*} 
       \begin{align*}
        &\eqtype
        \AxiomC{$\pi$\vphantom{A}}
        \UnaryInfC{$\indfun{v}{\bangm{u}A}\indvdash K B$}
        \functorialLine
        \UnaryInfC{$\vdots$}
        \functorialLine
        \UnaryInfC{$\indfun{\pb{{f_1}}{\tpb{v}{f_n}{\cdots}{f_2}}}{\cdots \indfun{\pb{{f_n}}{v}}{\indfun{v}{\bangm{u}A}}} \indvdash L \indfun{\pb{{f_1}}{\tpb{v}{f_n}{\cdots}{f_2}}}{\cdots \indfun{\pb{{f_n}}{v}}{B}}$}
        \AxiomC{$\tpb{v}{f_n}{\cdots}{f_1};g_1\dots;g_k=\pb{{f_1}}{\tpb{v}{f_n}{\cdots}{f_2}};\dots;\pb{{f_n}}{v};v$}
        \subtypLine
        \BinaryInfC{$\indfun{\tpb{v}{f_n}{\cdots}{f_1}}{\indfun{g_1}{\cdots \indfun{g_k}{\bangm{u}A}}}\indvdash L \indfun{\pb{{f_1}}{\tpb{v}{f_n}{\cdots}{f_2}}}{\cdots \indfun{\pb{{f_n}}{v}}{B}}$}        
        \functorialLine
        \UnaryInfC{$\indfun{e_v}{\indfun{\tpb{v}{f_n}{\cdots}{f_1}}{\indfun{g_1}{\cdots \indfun{g_k}{\bangm{u}A}}}}\indvdash {L'} \indfun{e_v}{\indfun{\pb{{f_1}}{\tpb{v}{f_n}{\cdots}{f_2}}}{\cdots \indfun{\pb{{f_n}}{v}}{B}}}$}        
        \AxiomC{$e_v;\pb{{f_1}}{\tpb{v}{f_n}{\cdots}{f_2}};\dots;\pb{{f_n}}{v}=\pb{{g_1}}{\tpb{v}{g_k}{\cdots}{g_2}};\dots,\pb{{g_k}}{v}$}
        \subtypLine
        \BinaryInfC{$\indfun{e_v}{\indfun{\tpb{v}{f_n}{\cdots}{f_1}}{\indfun{g_1}{\cdots \indfun{g_k}{\bangm{u}A}}}}\indvdash {L'} \indfun{\pb{{g_1}}{\tpb{v}{g_k}{\cdots}{g_2}}}{\cdots \indfun{\pb{{g_k}}{v}}{B}}$}
        \AxiomC{$\tpb{v}{g_k}{\cdots}{g_1}=e_v;\tpb{v}{f_n}{\cdots}{f_1}$}
        \subtypLine
        \BinaryInfC{$\indfun{\tpb{v}{g_k}{\cdots}{g_1}}{\indfun{g_1}{\cdots \indfun{g_k}{\bangm{u}A}}}\indvdash {L'} \indfun{\pb{{g_1}}{\tpb{v}{g_k}{\cdots}{g_2}}}{\cdots \indfun{\pb{{g_k}}{v}}{B}}$}
        \UnaryInfC{$\indfun{g_1}{\cdots \indfun{g_k}{\bangm{u}A}}\indvdash J \indfun{g_1}{\cdots \indfun{g_k}{\bangm{v}B}}$}
        \DisplayProof\\
        &\eqtype
        \AxiomC{$\pi$\vphantom{A}}
        \UnaryInfC{$\indfun{v}{\bangm{u}A}\indvdash K B$}
        \functorialLine
        \UnaryInfC{$\vdots$}
        \functorialLine
        \UnaryInfC{$\indfun{\pb{{f_1}}{\tpb{v}{f_n}{\cdots}{f_2}}}{\cdots \indfun{\pb{{f_n}}{v}}{\indfun{v}{\bangm{u}A}}} \indvdash L \indfun{\pb{{f_1}}{\tpb{v}{f_n}{\cdots}{f_2}}}{\cdots \indfun{\pb{{f_n}}{v}}{B}}$}
        \functorialLine
        \UnaryInfC{$\indfun{e_v}{\indfun{\pb{{f_1}}{\tpb{v}{f_n}{\cdots}{f_2}}}{\cdots \indfun{\pb{{f_n}}{v}}{\indfun{v}{\bangm{u}A}}}} \indvdash {L'} \indfun{e_v}{\indfun{\pb{{f_1}}{\tpb{v}{f_n}{\cdots}{f_2}}}{\cdots \indfun{\pb{{f_n}}{v}}{B}}}$}
        \AxiomC{$\tpb{v}{f_n}{\cdots}{f_1};g_1\dots;g_k=\pb{{f_1}}{\tpb{v}{f_n}{\cdots}{f_2}};\dots;\pb{{f_n}}{v};v$}
        \dashedLine
        \UnaryInfC{$\indfun{\tpb{v}{f_n}{\cdots}{f_1}}{\indfun{g_1}{\cdots \indfun{g_k}{\bangm{u}A}}}\eqtype \indfun{\pb{{f_1}}{\tpb{v}{f_n}{\cdots}{f_2}}}{\cdots \indfun{\pb{{f_n}}{v}}{\indfun{v}{\bangm{u}A}}}$}
        \dashedLine
        \UnaryInfC{$\indfun{e_v}{\indfun{\tpb{v}{f_n}{\cdots}{f_1}}{\indfun{g_1}{\cdots \indfun{g_k}{\bangm{u}A}}}} \eqtype \indfun{e_v}{\indfun{\pb{{f_1}}{\tpb{v}{f_n}{\cdots}{f_2}}}{\cdots \indfun{\pb{{f_n}}{v}}{\indfun{v}{\bangm{u}A}}}}$}
        \subtypLine
        \BinaryInfC{$\indfun{e_v}{\indfun{\tpb{v}{f_n}{\cdots}{f_1}}{\indfun{g_1}{\cdots \indfun{g_k}{\bangm{u}A}}}}\indvdash {L'} \indfun{e_v}{\indfun{\pb{{f_1}}{\tpb{v}{f_n}{\cdots}{f_2}}}{\cdots \indfun{\pb{{f_n}}{v}}{B}}}$}    
        \AxiomC{$e_v;\pb{{f_1}}{\tpb{v}{f_n}{\cdots}{f_2}};\dots;\pb{{f_n}}{v}=\pb{{g_1}}{\tpb{v}{g_k}{\cdots}{g_2}};\dots,\pb{{g_k}}{v}$}
        \subtypLine
        \BinaryInfC{$\indfun{e_v}{\indfun{\tpb{v}{f_n}{\cdots}{f_1}}{\indfun{g_1}{\cdots \indfun{g_k}{\bangm{u}A}}}}\indvdash {L'} \indfun{\pb{{g_1}}{\tpb{v}{g_k}{\cdots}{g_2}}}{\cdots \indfun{\pb{{g_k}}{v}}{B}}$}
        \AxiomC{$\tpb{v}{g_k}{\cdots}{g_1}=e_v;\tpb{v}{f_n}{\cdots}{f_1}$}
        \subtypLine
        \BinaryInfC{$\indfun{\tpb{v}{g_k}{\cdots}{g_1}}{\indfun{g_1}{\cdots \indfun{g_k}{\bangm{u}A}}}\indvdash {L'} \indfun{\pb{{g_1}}{\tpb{v}{g_k}{\cdots}{g_2}}}{\cdots \indfun{\pb{{g_k}}{v}}{B}}$}
        \UnaryInfC{$\indfun{g_1}{\cdots \indfun{g_k}{\bangm{u}A}}\indvdash J \indfun{g_1}{\cdots \indfun{g_k}{\bangm{v}B}}$}
        \DisplayProof\\
        &\eqtype
        \AxiomC{$\pi$\vphantom{A}}
        \UnaryInfC{$\indfun{v}{\bangm{u}A}\indvdash K B$}
        \functorialLine
        \UnaryInfC{$\vdots$}
        \functorialLine
        \UnaryInfC{$\indfun{\pb{{f_1}}{\tpb{v}{f_n}{\cdots}{f_2}}}{\cdots \indfun{\pb{{f_n}}{v}}{\indfun{v}{\bangm{u}A}}} \indvdash L \indfun{\pb{{f_1}}{\tpb{v}{f_n}{\cdots}{f_2}}}{\cdots \indfun{\pb{{f_n}}{v}}{B}}$}
        \functorialLine
        \UnaryInfC{$\indfun{e_v}{\indfun{\pb{{f_1}}{\tpb{v}{f_n}{\cdots}{f_2}}}{\cdots \indfun{\pb{{f_n}}{v}}{\indfun{v}{\bangm{u}A}}}} \indvdash {L'} \indfun{e_v}{\indfun{\pb{{f_1}}{\tpb{v}{f_n}{\cdots}{f_2}}}{\cdots \indfun{\pb{{f_n}}{v}}{B}}}$}
        \AxiomC{$e_v;\pb{{f_1}}{\tpb{v}{f_k}{\cdots}{f_2}};\dots;\pb{{f_n}}{v}=\pb{{g_1}}{\tpb{v}{g_k}{\cdots}{g_2}};\dots;\pb{{g_k}}{v}$}
        \subtypLine        
        \BinaryInfC{$\indfun{\pb{{g_1}}{\tpb{v}{g_k}{\cdots}{g_2}}}{\cdots \indfun{\pb{{g_k}}{v}}{\indfun{v}{\bangm{u}A}}} \indvdash {{L'}'} \indfun{\pb{{g_1}}{\tpb{v}{g_k}{\cdots}{g_2}}}{\cdots \indfun{\pb{{g_k}}{v}}{B}}$}
        \AxiomC{$\tpb{v}{g_k}{\cdots}{g_1};g_1\dots;g_n=\pb{{g_1}}{\tpb{v}{g_k}{\cdots}{g_2}};\dots;\pb{{g_k}}{v};v$}
        \subtypLine
        \BinaryInfC{$\indfun{\tpb{v}{g_n}{\cdots}{g_1}}{\indfun{g_1}{\cdots \indfun{g_k}{\bangm{u}A}}}\indvdash {L'} \indfun{\pb{{g_1}}{\tpb{v}{g_k}{\cdots}{g_2}}}{\cdots \indfun{\pb{{g_k}}{v}}{B}}$}
        \UnaryInfC{$\indfun{g_1}{\cdots \indfun{g_k}{\bangm{u}A}} \indvdash J \indfun{g_1}{\cdots \indfun{g_k}{\bangm{v}B}}$}
        \DisplayProof\\
        &\eqtype
        \AxiomC{$\pi$\vphantom{A}}
        \UnaryInfC{$\indfun{v}{\bangm{u}A}\indvdash {L'} B$}
        \functorialLine
        \UnaryInfC{$\vdots$}
        \functorialLine
        \UnaryInfC{$\indfun{\pb{{g_1}}{\tpb{v}{g_k}{\cdots}{g_2}}}{\cdots \indfun{\pb{{g_k}}{v}}{\indfun{v}{\bangm{u}A}}} \indvdash {L'} \indfun{\pb{{g_1}}{\tpb{v}{g_k}{\cdots}{g_2}}}{\cdots \indfun{\pb{{g_k}}{v}}{B}}$}
        \AxiomC{$\tpb{v}{g_k}{\cdots}{g_1};g_1\dots;g_n=\pb{{g_1}}{\tpb{v}{g_k}{\cdots}{g_2}};\dots;\pb{{g_k}}{v};v$}
        \subtypLine
        \BinaryInfC{$\indfun{\tpb{v}{g_n}{\cdots}{g_1}}{\indfun{g_1}{\cdots \indfun{g_k}{\bangm{u}A}}}\indvdash {L'} \indfun{\pb{{g_1}}{\tpb{v}{g_k}{\cdots}{g_2}}}{\cdots \indfun{\pb{{g_k}}{v}}{B}}$}
        \UnaryInfC{$\indfun{g_1}{\cdots \indfun{g_k}{\bangm{u}A}} \indvdash J \indfun{g_1}{\cdots \indfun{g_k}{\bangm{v}B}}$}
        \DisplayProof\\
        &=
        \AxiomC{$\pi$\vphantom{A}}
        \UnaryInfC{$\indfun{v}{\bangm{u}A}\indvdash {K} B$}
        \UnaryInfC{$\bangm{u}A\indvdash I \bangm{v}B$}
        \functorialLine
        \UnaryInfC{$\vdots$}
        \functorialLine
        \UnaryInfC{$\indfun{g_1}{\cdots \indfun{g_k}{\bangm{u}A}} \indvdash J \indfun{g_1}{\cdots \indfun{g_k}{\bangm{v}B}}$}
        \DisplayProof        
      \end{align*}}}
  
  Let $h$ be the only function such that $\pb{v}{f}\circ h=g\circ \pb{f}{v\circ g}$ and $\pb{v}{f}\circ h=\pb{(v\circ g)}{f}$.
  
     \scalebox{0.5}{\parbox{1.05\linewidth}{  \begin{align*}
        &\bottomAlignProof
        \AxiomC{$\pi$\vphantom{A}}
        \UnaryInfC{$\indfun{v}{\bangm{u}A}\indvdash K B$}
        \UnaryInfC{$\bangm{u}A\indvdash I \bangm{v}B$}
        \AxiomC{$\rho$\vphantom{A}}
        \UnaryInfC{$\indfun{g}{B}\subtype B'$}
        \UnaryInfC{$\bangm{v}B \subtype \bangm{v\circ g}B'$}
        \subtypLine
        \BinaryInfC{$\bangm{u}A\indvdash I \bangm{v\circ g}B'$}
        \functorialLine
        \UnaryInfC{$\indfun{f}{\bangm{u}A}\indvdash J \indfun{f}{\bangm{v\circ g}B'}$}
        \DisplayProof \\
        &=
        \AxiomC{$\pi$\vphantom{A}}
        \UnaryInfC{$\indfun{v}{\bangm{u}A}\indvdash K B$}
\functorialLine
        \UnaryInfC{$\indfun{g}{\indfun{v}{\bangm{u}A}}\indvdash {L} \indfun{g}{B}$}
        \AxiomC{$\rho$\vphantom{A}}
        \UnaryInfC{$\indfun{g}{B}\subtype B'$}
        \subtypLine
        \BinaryInfC{$\indfun{g}{\indfun{v}{\bangm{u}A}}\indvdash {L} B'$}
        \AxiomC{$v\circ g = g;v$}
        \dashedLine
        \UnaryInfC{$\indfun{(v\circ g)}{\bangm{u}A} \subtype \indfun{g}{\indfun{v}{\bangm{u}A}}$}
        \subtypLine
        \BinaryInfC{$\indfun{(v\circ g)}{\bangm{u}A}\indvdash {L} B'$}
        \UnaryInfC{$\bangm{u}A\indvdash I \bangm{v\circ g}B'$}
        \functorialLine
        \UnaryInfC{$\indfun{f}{\bangm{u}A}\indvdash J \indfun{f}{\bangm{v\circ g}B'}$}
        \DisplayProof\\
        &=
        \AxiomC{$\pi$\vphantom{A}}
        \UnaryInfC{$\indfun{v}{\bangm{u}A}\indvdash K B$}
        \AxiomC{$g$}
        \functorialLine
        \BinaryInfC{$\indfun{g}{\indfun{v}{\bangm{u}A}}\indvdash {L} \indfun{g}{B}$}
        \AxiomC{$\rho$\vphantom{A}}
        \UnaryInfC{$\indfun{g}{B}\subtype B'$}
        \subtypLine
        \BinaryInfC{$\indfun{g}{\indfun{v}{\bangm{u}A}}\indvdash {L} B'$}
        \AxiomC{$v\circ g = g;v$}
        \dashedLine
        \UnaryInfC{$\indfun{(v\circ g)}{\bangm{u}A} \subtype \indfun{g}{\indfun{v}{\bangm{u}A}}$}
        \subtypLine
        \BinaryInfC{$\indfun{(v\circ g)}{\bangm{u}A}\indvdash {L} B'$}
\functorialLine
        \UnaryInfC{$\indfun{\pb{f}{v\circ g}}{\indfun{(v\circ g)}{\bangm{u}A}}\indvdash {L {\times_{\!I}} J} \indfun{\pb{f}{v\circ g}}{B'}$}
        \AxiomC{}
        \UnaryInfC{$\pb{(v\circ g)}{f};f=\pb{f}{(v\circ g)};(v\circ g)$}
        \dashedLine
        \UnaryInfC{$\indfun{\pb{(v\circ g)}{f}}{\indfun{f}{\bangm{u}A}}\eqtype \indfun{\pb{f}{(v\circ g)}}{\indfun{(v\circ g)}{\bangm{u}A}}$}
        \subtypLine
        \BinaryInfC{$\indfun{\pb{(v\circ g)}{f}}{\bangm{\pb{u}{f}}\indfun{\pb{f}{u}}{A}}\indvdash {L {\times_{\!I}} J} \indfun{\pb{f}{v\circ g}}{B'}$}
        \UnaryInfC{$\indfun{f}{\bangm{u}A}\indvdash J \indfun{f}{\bangm{v\circ g}B'}$}
        \DisplayProof\\
        &\eqtype
        \AxiomC{$\pi$\vphantom{A}}
        \UnaryInfC{$\indfun{v}{\bangm{u}A}\indvdash K B$}
        \functorialLine
        \UnaryInfC{$\indfun{g}{\indfun{v}{\bangm{u}A}}\indvdash {L} \indfun{g}{B}$}
        \functorialLine
        \UnaryInfC{$\indfun{\pb{f}{v\circ g}}{\indfun{g}{\indfun{v}{\bangm{u}A}}}\indvdash {L {\times_{\!I}} J} \indfun{\pb{f}{v\circ g}}{\indfun{g}{B}}$}
        \AxiomC{$\rho$\vphantom{A}}
        \UnaryInfC{$\indfun{g}{B}\subtype B'$}
        \dashedLine
        \UnaryInfC{$\indfun{\pb{f}{v\circ g}}{\indfun{g}{B}}\subtype \indfun{\pb{f}{v\circ g}}{B'}$}
        \subtypLine
        \BinaryInfC{$\indfun{\pb{f}{v\circ g}}{\indfun{g}{\indfun{v}{\bangm{u}A}}}\indvdash {L {\times_{\!I}} J} \indfun{\pb{f}{v\circ g}}{B'}$}
        \AxiomC{$v\circ g = g;v$}
        \dashedLine
        \UnaryInfC{$\indfun{(v\circ g)}{\bangm{u}A} \subtype \indfun{g}{\indfun{v}{\bangm{u}A}}$}
        \dashedLine
        \UnaryInfC{$\indfun{\pb{f}{v\circ g}}{\indfun{(v\circ g)}{\bangm{u}A}} \subtype \indfun{\pb{f}{v\circ g}}{\indfun{g}{\indfun{v}{\bangm{u}A}}}$}
        \subtypLine
        \BinaryInfC{$\indfun{\pb{f}{v\circ g}}{\indfun{(v\circ g)}{\bangm{u}A}}\indvdash {L {\times_{\!I}} J} \indfun{\pb{f}{v\circ g}}{B'}$}
        \AxiomC{}
        \UnaryInfC{$\pb{(v\circ g)}{f};f=\pb{f}{(v\circ g)};(v\circ g)$}
        \dashedLine
        \UnaryInfC{$\indfun{\pb{(v\circ g)}{f}}{\indfun{f}{\bangm{u}A}}\eqtype \indfun{\pb{f}{(v\circ g)}}{\indfun{(v\circ g)}{\bangm{u}A}}$}
        \subtypLine
        \BinaryInfC{$\indfun{\pb{(v\circ g)}{f}}{\bangm{\pb{u}{f}}\indfun{\pb{f}{u}}{A}}\indvdash {L {\times_{\!I}} J} \indfun{\pb{f}{v\circ g}}{B'}$}
        \UnaryInfC{$\indfun{f}{\bangm{u}A}\indvdash J \indfun{f}{\bangm{v\circ g}B'}$}
        \DisplayProof\\
        &\eqtype
        \AxiomC{$\pi$\vphantom{A}}
        \UnaryInfC{$\indfun{v}{\bangm{u}A}\indvdash K B$}
        \functorialLine
        \UnaryInfC{$\indfun{\pb{f}{v}}{\indfun{v}{\bangm{u}A}}\indvdash {{K {\times_{\!I}} J}} \indfun{\pb{f}{v}}{B}$}
        \functorialLine
        \UnaryInfC{$\indfun{h}{\indfun{\pb{f}{v}}{\indfun{v}{\bangm{u}A}}}\indvdash {{L {\times_{\!I}} J}} \indfun{h}{\indfun{\pb{f}{v}}{B}}$}
        \AxiomC{}
        \UnaryInfC{$\pb{f}{v\circ g};g = h;\pb{f}{v}$}
        \subtypLine
        \BinaryInfC{$\indfun{\pb{f}{v\circ g}}{\indfun{g}{\indfun{v}{\bangm{u}A}}}\indvdash {{L {\times_{\!I}} J}} \indfun{\pb{f}{v\circ g}}{\indfun{g}{B}}$}
        \AxiomC{$\rho$\vphantom{A}}
        \UnaryInfC{$\indfun{g}{B}\subtype B'$}
        \dashedLine
        \UnaryInfC{$\indfun{\pb{f}{v\circ g}}{\indfun{g}{B}}\subtype \indfun{\pb{f}{v\circ g}}{B'}$}
        \subtypLine
        \BinaryInfC{$\indfun{\pb{f}{v\circ g}}{\indfun{g}{\indfun{v}{\bangm{u}A}}}\indvdash {{L {\times_{\!I}} J}} \indfun{\pb{f}{v\circ g}}{B'}$}
        \AxiomC{$v\circ g = g;v$}
        \dashedLine
        \UnaryInfC{$\indfun{(v\circ g)}{\bangm{u}A} \subtype \indfun{g}{\indfun{v}{\bangm{u}A}}$}
        \dashedLine
        \UnaryInfC{$\indfun{\pb{f}{v\circ g}}{\indfun{(v\circ g)}{\bangm{u}A}} \subtype \indfun{\pb{f}{v\circ g}}{\indfun{g}{\indfun{v}{\bangm{u}A}}}$}
        \subtypLine
        \BinaryInfC{$\indfun{\pb{f}{v\circ g}}{\indfun{(v\circ g)}{\bangm{u}A}}\indvdash {L {\times_{\!I}} J} \indfun{\pb{f}{v\circ g}}{B'}$}
        \AxiomC{}
        \UnaryInfC{$\pb{(v\circ g)}{f};f=\pb{f}{(v\circ g)};(v\circ g)$}
        \dashedLine
        \UnaryInfC{$\indfun{\pb{(v\circ g)}{f}}{\indfun{f}{\bangm{u}A}}\eqtype \indfun{\pb{f}{(v\circ g)}}{\indfun{(v\circ g)}{\bangm{u}A}}$}
        \subtypLine
        \BinaryInfC{$\indfun{\pb{(v\circ g)}{f}}{\bangm{\pb{u}{f}}\indfun{\pb{f}{u}}{A}}\indvdash {L {\times_{\!I}} J} \indfun{\pb{f}{v\circ g}}{B'}$}
        \UnaryInfC{$\indfun{f}{\bangm{u}A}\indvdash J \indfun{f}{\bangm{v\circ g}B'}$}
        \DisplayProof\\
        &\eqtype
        \AxiomC{$\pi$\vphantom{A}}
        \UnaryInfC{$\indfun{v}{\bangm{u}A}\indvdash K B$}
        \functorialLine
        \UnaryInfC{$\indfun{\pb{f}{v}}{\indfun{v}{\bangm{u}A}}\indvdash {K {\times_{\!I}} J} \indfun{\pb{f}{v}}{B}$}
        \functorialLine
        \UnaryInfC{$\indfun{h}{\indfun{\pb{f}{v}}{\indfun{v}{\bangm{u}A}}}\indvdash {L {\times_{\!I}} J} \indfun{h}{\indfun{\pb{f}{v}}{B}}$}
        \AxiomC{}
        \UnaryInfC{$\pb{v}{f};f=\pb{f}{v};v$}
        \dashedLine
        \UnaryInfC{$\indfun{\pb{v}{f}}{\indfun{f}{\bangm{u}A}}\eqtype \indfun{\pb{f}{v}}{\indfun{v}{\bangm{u}A}}$}
        \dashedLine
        \UnaryInfC{$\indfun{h}{\indfun{\pb{v}{f}}{\indfun{f}{\bangm{u}A}}}\eqtype \indfun{h}{\indfun{\pb{f}{v}}{\indfun{v}{\bangm{u}A}}}$}
        \subtypLine
        \BinaryInfC{$\indfun{h}{\indfun{\pb{v}{f}}{\bangm{\pb{u}{f}}\indfun{\pb{f}{u}}{A}}}\indvdash {L {\times_{\!I}} J} \indfun{h}{\indfun{\pb{f}{v}}{B}}$}
        \AxiomC{$\rho$\vphantom{A}}
        \UnaryInfC{$\indfun{g}{B}\subtype B'$}
        \dashedLine
        \UnaryInfC{$\indfun{\pb{f}{v\circ g}}{\indfun{g}{B}} \subtype \indfun{\pb{f}{v\circ g}}{B'}$}
        \AxiomC{}
        \UnaryInfC{$h;\pb{f}{v}=\pb{f}{v\circ g};g$}
        \UnaryInfC{$\indfun{h}{\indfun{\pb{f}{v}}{B}} \equiv \indfun{\pb{f}{v\circ g}}{\indfun{g}{B}}$}
        \dashedLine
        \BinaryInfC{$\indfun{h}{\indfun{\pb{f}{v}}{B}} \subtype \indfun{\pb{f}{v\circ g}}{B'}$}
        \subtypLine
        \BinaryInfC{$\indfun{h}{\indfun{\pb{v}{f}}{\bangm{\pb{u}{f}}\indfun{\pb{f}{u}}{A}}}\indvdash {L {\times_{\!I}} J} \indfun{\pb{f}{v\circ g}}{B'}$}
        \AxiomC{}
        \UnaryInfC{$\pb{(v\circ g)}{f} = h;\pb{v}{f}$}
        \dashedLine
        \UnaryInfC{$\indfun{\pb{(v\circ g)}{f}}{\bangm{\pb{u}{f}}\indfun{\pb{f}{u}}{A}} \subtype \indfun{h}{\indfun{\pb{v}{f}}{\bangm{\pb{u}{f}}\indfun{\pb{f}{u}}{A}}}$}
        \subtypLine
        \BinaryInfC{$\indfun{\pb{(v\circ g)}{f}}{\bangm{\pb{u}{f}}\indfun{\pb{f}{u}}{A}}\indvdash {L {\times_{\!I}} J} \indfun{\pb{f}{v\circ g}}{B'}$}
        \UnaryInfC{$\indfun{f}{\bangm{u}A}\indvdash J \indfun{f}{\bangm{v\circ g}B'}$}
        \DisplayProof\\
        &\eqtype
        \AxiomC{$\pi$\vphantom{A}}
        \UnaryInfC{$\indfun{v}{\bangm{u}A}\indvdash K B$}
        \functorialLine
        \UnaryInfC{$\indfun{\pb{f}{v}}{\indfun{v}{\bangm{u}A}}\indvdash {K {\times_{\!I}} J} \indfun{\pb{f}{v}}{B}$}
        \AxiomC{}
        \UnaryInfC{$\pb{v}{f};f=\pb{f}{v};v$}
        \dashedLine
        \UnaryInfC{$\indfun{\pb{v}{f}}{\indfun{f}{\bangm{u}A}}\eqtype \indfun{\pb{f}{v}}{\indfun{v}{\bangm{u}A}}$}
        \subtypLine
        \BinaryInfC{$\indfun{\pb{v}{f}}{\bangm{\pb{u}{f}}\indfun{\pb{f}{u}}{A}}\indvdash {K {\times_{\!I}} J} \indfun{\pb{f}{v}}{B}$}
        \functorialLine
        \UnaryInfC{$\indfun{h}{\indfun{\pb{v}{f}}{\bangm{\pb{u}{f}}\indfun{\pb{f}{u}}{A}}}\indvdash {L {\times_{\!I}} J} \indfun{h}{\indfun{\pb{f}{v}}{B}}$}
        \AxiomC{$\rho$\vphantom{A}}
        \UnaryInfC{$\indfun{g}{B}\subtype B'$}
        \dashedLine
        \UnaryInfC{$\indfun{\pb{f}{v\circ g}}{\indfun{g}{B}} \subtype \indfun{\pb{f}{v\circ g}}{B'}$}
        \AxiomC{}
        \UnaryInfC{$h;\pb{f}{v}=\pb{f}{v\circ g};g$}
        \UnaryInfC{$\indfun{h}{\indfun{\pb{f}{v}}{B}} \equiv \indfun{\pb{f}{v\circ g}}{\indfun{g}{B}}$}
        \dashedLine
        \BinaryInfC{$\indfun{h}{\indfun{\pb{f}{v}}{B}} \subtype \indfun{\pb{f}{v\circ g}}{B'}$}
        \subtypLine
        \BinaryInfC{$\indfun{h}{\indfun{\pb{v}{f}}{\bangm{\pb{u}{f}}\indfun{\pb{f}{u}}{A}}}\indvdash {L {\times_{\!I}} J} \indfun{\pb{f}{v\circ g}}{B'}$}
        \AxiomC{}
        \UnaryInfC{$\pb{(v\circ g)}{f} = h;\pb{v}{f}$}
        \dashedLine
        \UnaryInfC{$\indfun{\pb{(v\circ g)}{f}}{\bangm{\pb{u}{f}}\indfun{\pb{f}{u}}{A}} \subtype \indfun{h}{\indfun{\pb{v}{f}}{\bangm{\pb{u}{f}}\indfun{\pb{f}{u}}{A}}}$}
        \subtypLine
        \BinaryInfC{$\indfun{\pb{(v\circ g)}{f}}{\bangm{\pb{u}{f}}\indfun{\pb{f}{u}}{A}}\indvdash {L {\times_{\!I}} J} \indfun{\pb{f}{v\circ g}}{B'}$}
        \UnaryInfC{$\indfun{f}{\bangm{u}A}\indvdash J \indfun{f}{\bangm{v\circ g}B'}$}
        \DisplayProof\\
        &=
        \AxiomC{$\pi$\vphantom{A}}
        \UnaryInfC{$\indfun{v}{\bangm{u}A}\indvdash K B$}
        \functorialLine
        \UnaryInfC{$\indfun{\pb{f}{v}}{\indfun{v}{\bangm{u}A}}\indvdash {K {\times_{\!I}} J} \indfun{\pb{f}{v}}{B}$}
        \AxiomC{}
        \UnaryInfC{$\pb{v}{f};f=\pb{f}{v};v$}
        \dashedLine
        \UnaryInfC{$\indfun{\pb{v}{f}}{\indfun{f}{\bangm{u}A}}\eqtype \indfun{\pb{f}{v}}{\indfun{v}{\bangm{u}A}}$}
        \subtypLine
        \BinaryInfC{$\indfun{\pb{v}{f}}{\bangm{\pb{u}{f}}\indfun{\pb{f}{u}}{A}}\indvdash {K {\times_{\!I}} J} \indfun{\pb{f}{v}}{B}$}
        \UnaryInfC{$\indfun{f}{\bangm{u}A}\indvdash J \indfun{f}{\bangm{v}B}$}
        \AxiomC{$\rho$\vphantom{A}}
        \UnaryInfC{$\indfun{g}{B}\subtype B'$}
        \dashedLine
        \UnaryInfC{$\indfun{\pb{f}{v\circ g}}{\indfun{g}{B}} \subtype \indfun{\pb{f}{v\circ g}}{B'}$}
        \AxiomC{}
        \UnaryInfC{$h;\pb{f}{v}=\pb{f}{v\circ g};g$}
        \UnaryInfC{$\indfun{h}{\indfun{\pb{f}{v}}{B}} \equiv \indfun{\pb{f}{v\circ g}}{\indfun{g}{B}}$}
        \dashedLine
        \BinaryInfC{$\indfun{h}{\indfun{\pb{f}{v}}{B}} \subtype \indfun{\pb{f}{v\circ g}}{B'}$}
        \UnaryInfC{$\bangm{\pb{v}{f}}\indfun{\pb{f}{v}}{B} \subtype \bangm{\pb{v}{f}\circ h}\indfun{\pb{f}{v\circ g}}{B'}$}
        \subtypLine
        \BinaryInfC{$\indfun{f}{\bangm{u}A}\indvdash J \indfun{f}{\bangm{v\circ g}B'}$}
        \DisplayProof\\
        &=
        \AxiomC{$\pi$\vphantom{A}}
        \UnaryInfC{$\indfun{v}{\bangm{u}A}\indvdash K B$}
        \functorialLine
        \UnaryInfC{$\indfun{\pb{f}{v}}{\indfun{v}{\bangm{u}A}}\indvdash {I \times_{J} K} \indfun{\pb{f}{v}}{B}$}
        \AxiomC{}
        \UnaryInfC{$\pb{v}{f};f=\pb{f}{v};v$}
        \dashedLine
        \UnaryInfC{$\indfun{\pb{v}{f}}{\indfun{f}{\bangm{u}A}}\eqtype \indfun{\pb{f}{v}}{\indfun{v}{\bangm{u}A}}$}
        \subtypLine
        \BinaryInfC{$\indfun{\pb{v}{f}}{\bangm{\pb{u}{f}}\indfun{\pb{f}{u}}{A}}\indvdash {I \times_{J} K} \indfun{\pb{f}{v}}{B}$}
        \UnaryInfC{$\indfun{f}{\bangm{u}A}\indvdash J \indfun{f}{\bangm{v}B}$}
        \AxiomC{$\rho$\vphantom{A}}
        \UnaryInfC{$\indfun{g}{B}\subtype B'$}
        \UnaryInfC{$\bangm{v}B \subtype \bangm{v\circ g}B'$}
        \dashedLine
        \UnaryInfC{$\indfun{f}{\bangm{v}B} \subtype \indfun{f}{\bangm{v\circ g}B'}$}
        \subtypLine
        \BinaryInfC{$\indfun{f}{\bangm{u}A}\indvdash J \indfun{f}{\bangm{v\circ g}B'}$}
        \DisplayProof\\
        &=
        \AxiomC{$\pi$\vphantom{A}}
        \UnaryInfC{$\indfun{v}{\bangm{u}A}\indvdash K B$}
        \UnaryInfC{$\bangm{u}A\indvdash I \bangm{v}B$}
        \functorialLine
        \UnaryInfC{$\indfun{f}{\bangm{u}A}\indvdash J \indfun{f}{\bangm{v}B}$}
        \AxiomC{$\rho$\vphantom{A}}
        \UnaryInfC{$\indfun{g}{B}\subtype B'$}
        \UnaryInfC{$\bangm{v}B \subtype \bangm{v\circ g}B'$}
        \dashedLine
        \UnaryInfC{$\indfun{f}{\bangm{v}B} \subtype \indfun{f}{\bangm{v\circ g}B'}$}
        \subtypLine
        \BinaryInfC{$\indfun{f}{\bangm{u}A}\indvdash J \indfun{f}{\bangm{v\circ g}B'}$}
        \DisplayProof
      \end{align*}}}
 }
\end{proof}

\begin{theorem}\label{th:equiv_base_change}
  For any two proofs
  $$\pi_1\eqtype\pi_2 \implies \indfun{f}{\pi_1}\eqtype\indfun{f}{\pi_2}.$$
\end{theorem}
\begin{proof}
  By an easy induction on the proof $\pi_1\eqtype\pi_2$ using the fact that base changes can be lifted to both sequents and subtyping.
\end{proof}

\subsection{Cut elimination}

The cut elimination of $\IndLL$ is following exactly the cut-elimination of LL, in the sens that if $\pi\reduce \pi'$, then $\underline\pi\reduce\underline \pi'$, i.e., our only contribution is with the management of the indexes.

Concretely, we are using the subtyping and base change derivable rules from Lemma~\ref{lemma:BS_and_ST_rules}. It is important to notice that those derived rules are only managing indexes, in the sens that applying them do not change the structure $\underline\pi$ of the proof.

\begin{definition}[cut-elimination]
 Indexed linear logic has a cut-elimination procedure refining LL cut-elimination\longVOnly{:
    \longVOnly{

  Exponential basic rules:\\
}
    \begin{small}
\AxiomC{$\pi_1$}
        \UnaryInfC{$\indfun{w}{\bangm{u_1}A_1},\dots,\indfun{w}{\bangm{u_n}A_n}\indvdash {J} B$}
        \UnaryInfC{$\bangm{u_1}A_1,\dots,\bangm{u_n}A_n\indvdash I \bangm w B$}
        \AxiomC{$\pi_2$}
        \UnaryInfC{$\bangm wB,\bangm wB\indvdash I\Gamma$}
        \UnaryInfC{$\bangm wB\indvdash I\Gamma$}
        \cutLine
        \BinaryInfC{$\bangm{u_1}A_1,\dots,\bangm{u_n}A_n\indvdash I\Gamma$}
        \DisplayProof
         $\quad\reduce\quad$
        \AxiomC{$\pi_1$}
        \UnaryInfC{$\indfun{w}{\bangm{u_1}A_1},\dots,\indfun{w}{\bangm{u_n}A_n}\indvdash {J} B$}
        \UnaryInfC{$\bangm{u_1}A_1,\dots,\bangm{u_n}A_n\indvdash I \bangm w B$}
        \AxiomC{$\pi_1$}
        \UnaryInfC{$\indfun{w}{\bangm{u_1}A_1},\dots,\indfun{w}{\bangm{u_n}A_n}\indvdash {J} B$}
        \UnaryInfC{$\bangm{u_1}A_1,\dots,\bangm{u_n}A_n\indvdash I \bangm w B$}
        \AxiomC{$\pi_2$}
        \UnaryInfC{$\bangm wB,\bangm wB\indvdash I\Gamma$}
        \cutLine
        \BinaryInfC{$\bangm{u_1}A_1,\dots,\bangm{u_n}A_n,\bangm wB\indvdash I\Gamma$}
        \cutLine
        \BinaryInfC{$\bangm{u_1}A_1,\bangm{u_1}A_1,\dots,\bangm{u_n}A_n,\bangm{u_n}A_n\indvdash I\Gamma$}
\UnaryInfC{$\svdots$}
        \UnaryInfC{$\bangm{u_1}A_1,\dots,\bangm{u_n}A_n\indvdash I\Gamma$}
        \DisplayProof
\end{small}
  
      \begin{small}
\AxiomC{$\pi_1$}
          \UnaryInfC{$\indfun{w}{\bangm{u_1}A_1},\dots,\indfun{w}{\bangm{u_n}A_n}\indvdash {J} B$}
          \UnaryInfC{$\bangm{u_1}A_1,\dots,\bangm{u_n}A_n\indvdash I \bangm w B$}
          \AxiomC{$\pi_2$}
          \UnaryInfC{$\indvdash I \Gamma$}
          \UnaryInfC{$\bangm{w}B\indvdash I \Gamma$}
          \cutLine
          \BinaryInfC{$\bangm{u_1}A_1,\dots,\bangm{u_n}A_n\indvdash I \Gamma$}
          \DisplayProof\qquad
$\reduce$
          \AxiomC{$\pi_2$}
          \UnaryInfC{$\indvdash I \Gamma$}
          \longVOnly{\UnaryInfC{$\bangm{u_1}A_1\indvdash I \Gamma$}}
          \UnaryInfC{$\svdots$}
          \UnaryInfC{$\bangm{u_1}A_1,\dots,\bangm{u_n}A_n\indvdash I \Gamma$}
          \DisplayProof\\
          \AxiomC{$\pi_1$}
          \UnaryInfC{$\indfun{w}{\bangm{u_1}A_1},\dots,\indfun{w}{\bangm{u_n}A_n}\indvdash {J} B$}
          \UnaryInfC{$\bangm{u_1}A_1,\dots,\bangm{u_n}A_n\indvdash I \bangm w B$}
          \AxiomC{$\pi_2$}
          \UnaryInfC{$\indfun{f}{B} \indvdash I \Gamma$}
          \AxiomC{$w\circ f = \id$}
          \BinaryInfC{$\bangm{w}B\indvdash I \Gamma$}
          \cutLine
          \BinaryInfC{$\bangm{u_1}A_1,\dots,\bangm{u_n}A_n\indvdash I \Gamma$}
          \DisplayProof\quad
$\reduce$\quad
          \AxiomC{$\pi_1$}
          \UnaryInfC{$\indfun{w}{\bangm{u_1}A_1},\dots,\indfun{w}{\bangm{u_n}A_n}\indvdash {J} B$}
          \functorialLine
          \UnaryInfC{$\indfun{f}{\indfun{w}{\bangm{u_1}A_1}},\dots,\indfun{f}{\indfun{w}{\bangm{u_n}A_n}}\indvdash I \indfun{f}{B}$}
          \AxiomC{$w\circ f = \id$}
          \subtypLine
          \BinaryInfC{$\bangm{u_1}A_1,\dots,\bangm{u_n}A_n\indvdash I \indfun{f}{B}$}
          \AxiomC{$\pi_2$}
          \UnaryInfC{$\indfun{f}{B} \indvdash I \Gamma$}
          \cutLine
          \BinaryInfC{$\bangm{u_1}A_1,\dots,\bangm{u_n}A_n\indvdash I \Gamma$}
          \DisplayProof
\AxiomC{$\pi_1$}
          \UnaryInfC{$\indfun{w}{\bangm{u_1}A_1},\dots,\indfun{w}{\bangm{u_n}A_n}\indvdash {J} B$}
          \UnaryInfC{$\bangm{u_1}A_1,\dots,\bangm{u_n}A_n\indvdash I \bangm w B$}
          \AxiomC{$\pi_2$}
          \UnaryInfC{$\indfun{v}{\bangm{w}B},\indfun{v}{\bangm{u'_1}A'_1},\dots,\indfun{v}{\bangm{u'_m}A'_m}\indvdash {K} C$}
          \UnaryInfC{$\bangm wB,\bangm{u'_1}A'_1,\dots,\bangm{u'_m}A'_m\indvdash I \bangm v C$}
          \cutLine
          \BinaryInfC{$\bangm{u_1}A_1,\dots,\bangm{u_n}A_n,\bangm{u'_1}A'_1,\dots,\bangm{u'_m}A'_m \indvdash I \bangm{v}C$}
          \DisplayProof\longVOnly{\\}\shortVOnly{\\[-3em]}
\begin{flushright}
          $\reduce$\qquad
          \AxiomC{$\pi_1$}
          \UnaryInfC{$\indfun{w}{\bangm{u_1}A_1},\dots,\indfun{w}{\bangm{u_n}A_n}\indvdash {J} B$}
          \UnaryInfC{$\bangm{u_1}A_1,\dots,\bangm{u_n}A_n\indvdash I \bangm w B$}
          \functorialLine
          \UnaryInfC{$\indfun{v}{\bangm{u_1}A_1}\dots,\indfun{v}{\bangm{u_n}A_n}\indvdash K \indfun{v}{\bangm w B}$}
          \AxiomC{$\pi_2$}
          \UnaryInfC{$\indfun{v}{\bangm{w}B},\indfun{v}{\bangm{u'_1}A'_1},\dots,\indfun{v}{\bangm{u'_m}A'_m}\indvdash K C$}
          \cutLine
          \BinaryInfC{$\indfun{v}{\bangm{u_1}A_1},\dots,\indfun{v}{\bangm{u_n}A_n},\indfun{v}{\bangm{u'_1}A'_1},\dots,\indfun{v}{\bangm{u'_m}A'_m}\indvdash K C$}
          \UnaryInfC{$\bangm{u_1}A_1,\dots,\bangm{u_n}A_n,\bangm{u'_1}A'_1,\dots,\bangm{u'_m}A'_m \indvdash I \bangm{v}C$}
          \DisplayProof
          \end{flushright}
\end{small}
    \shortVOnly{\vspace*{.5em}}

    \longVOnly{
    \medskip  
    Fixpoints basic rule:\\
      \begin{small}
        \begin{align*}
          &\AxiomC{$\pi_1$}
          \UnaryInfC{$\Gamma\indvdash I \indfun{f}{A[\mum{\id}X.A/X]}$}
          \UnaryInfC{$\Gamma\indvdash I\mum fX.A$}
          \AxiomC{$\pi_2$}
          \UnaryInfC{$\indfun{f}{A[\mum{\id}X.A/X]} \indvdash I \indfun{g}{B}$}
          \UnaryInfC{$\mum fX.A \indvdash I \indfun{(f; g)}{B}$}
          \cutLine
          \BinaryInfC{$\Gamma \indvdash I \indfun{(f; g)}{B}$}
          \DisplayProof
          & \quad\reduce\quad
          \AxiomC{$\pi_1$}
          \UnaryInfC{$\Gamma\indvdash I \indfun{f}{A[\mum{\id}X.A/X]}$}
          \AxiomC{$\pi_2$}
          \UnaryInfC{$\indfun{f}{A[\mum{\id}X.A/X]} \indvdash I \indfun{g}{B}$}
          \cutLine
          \BinaryInfC{$\Gamma \indvdash I \indfun{(f; g)}{B}$}
          \DisplayProof\\
          \end{align*}
      \end{small}
      
    Multiplicative and additive basic rules:\\
    }
      \begin{small}
      \longVOnly{
\AxiomC{$\pi_1$}
          \UnaryInfC{$\Gamma\indvdash I A$}
          \AxiomC{$\pi_2$}
          \UnaryInfC{$\Delta\indvdash I B$}
          \BinaryInfC{$\Gamma,\Delta\indvdash I A\otimes B$}
          \AxiomC{$\pi_3$}
          \UnaryInfC{$A,B\indvdash I \Xi$}
          \UnaryInfC{$A\otimes B\indvdash I \Xi$}
          \cutLine
          \BinaryInfC{$\Gamma,\Delta\indvdash I \Xi$}
          \DisplayProof
$\quad\reduce\quad$
          \AxiomC{$\pi_1$}
          \UnaryInfC{$\Gamma\indvdash I A$}
          \AxiomC{$\pi_2$}
          \UnaryInfC{$\Delta\indvdash I B$}
          \AxiomC{$\pi_3$}
          \UnaryInfC{$A,B\indvdash I \Xi$}
          \cutLine
          \BinaryInfC{$A,\Delta\indvdash I \Xi$}
          \cutLine
          \BinaryInfC{$\Gamma,\Delta\indvdash I \Xi$}
          \DisplayProof\\
          }
          \AxiomC{$\pi_1$}
          \UnaryInfC{$\Xi\indvdash {I} \indfun{\inv{i}}{A}$}
          \AxiomC{$\emptyset\vdash B\ \mathtt{def}$}
          \BinaryInfC{$\Xi\indvdash {I} A\oplusm{i}{\init} B$}
          \AxiomC{$\pi_2$}
          \UnaryInfC{$A\indvdash {I'} \indfun{i}{\Gamma}$}
          \AxiomC{$\pi_3$}
          \UnaryInfC{$B \indvdash \emptyset \indfun{\init}{\Gamma}$}
          \BinaryInfC{$A\oplusm{\id}{\init} B\indvdash {I} \Gamma$}
          \cutLine
          \BinaryInfC{$\Xi\indvdash {I} \Gamma$}
          \DisplayProof
          $\quad\reduce\quad$
          \AxiomC{$\pi_1$}
          \UnaryInfC{$\Xi\indvdash {I} \indfun{\inv{i}}{A}$}
          \AxiomC{$\pi_2$}
          \UnaryInfC{$A\indvdash {I'} \indfun{i}{\Gamma}$}
          \functorialLine
          \UnaryInfC{$\indfun{\inv{i}}{A}\indvdash I \indfun{\inv{i}}{\indfun{i}{\Gamma}}$}
          \subtypLine
          \UnaryInfC{$\indfun{\inv{i}}{A}\indvdash I \Gamma$}
          \cutLine
          \BinaryInfC{$\Xi\indvdash {I} \Gamma$}
          \DisplayProof
          \shortVOnly{\vspace*{1em}}
          \longVOnly{\\
          \AxiomC{$\pi_1$}
          \UnaryInfC{$\Xi\indvdash {J} B$}
          \AxiomC{$\emptyset\vdash A\ \mathtt{def}$}
          \BinaryInfC{$\Xi\indvdash {J} A\oplusm{\init}{\id} B$}
          \AxiomC{$\pi_2$}
          \UnaryInfC{$A\indvdash \emptyset \indfun{\init}{\Gamma}$}
          \AxiomC{$\pi_3$}
          \UnaryInfC{$B \indvdash J \Gamma$}
          \BinaryInfC{$A\oplusm{\init}{\id} B\indvdash {J} \Gamma$}
          \cutLine
          \BinaryInfC{$\Xi\indvdash {J} \Gamma$}
          \DisplayProof
          $ \quad\reduce\quad$
          \AxiomC{$\pi_1$}
          \UnaryInfC{$\Xi\indvdash {J} B$}
          \AxiomC{$\pi_3$}
          \UnaryInfC{$B\indvdash {J} \Gamma$}
          \cutLine
          \BinaryInfC{$\Xi\indvdash {J} \Gamma$}
          \DisplayProof
          }
\end{small}
    \longVOnly{
    \bigskip
    
    Additives non-trivial commutation rules:\\
    }
      \begin{small}
\longVOnly{
          \AxiomC{$\emptyset\vdash\Gamma,A\ \mathtt{def}$}
          \UnaryInfC{$\0\indvdash \emptyset \Gamma,A$}
          \AxiomC{$\pi_1$}
          \UnaryInfC{$A\indvdash \emptyset \Xi$}
          \cutLine
          \BinaryInfC{$\0\indvdash \emptyset \Xi,\Gamma$}
          \DisplayProof
$\quad\reduce\quad$
          \AxiomC{$\emptyset\vdash\Xi,\Gamma\ \mathtt{def}$}
          \UnaryInfC{$\0\indvdash \emptyset \Xi,\Gamma$}
          \DisplayProof\quad\quad
          viable since $A\indvdash \emptyset \Xi$ implies $\emptyset\vdash\Xi\ \mathtt{def}$\\
          }
          \AxiomC{$\pi_1$}
          \UnaryInfC{$B\indvdash I \indfun{i}{\Gamma},\indfun{i}{A}$}
          \AxiomC{$\pi_2$}
          \UnaryInfC{$C\indvdash J \indfun{j}{\Gamma},\indfun{j}{A}$}
          \BinaryInfC{$B\oplusm{i}{j} C\indvdash {W} \Gamma,A$}
          \AxiomC{$\pi_3$}
          \UnaryInfC{$A\indvdash {W} \Xi$}
          \cutLine
          \BinaryInfC{$B\oplusm{i}{j} C\indvdash {W} \Gamma,\Xi$}
          \DisplayProof
\quad$\reduce$\quad
          \AxiomC{$\pi_1$}
          \UnaryInfC{$B\indvdash I \indfun{i}{\Gamma},\indfun{i}{A}$}
          \AxiomC{$\pi_3$}
          \UnaryInfC{$A\indvdash W \Xi$}
          \functorialLine
          \UnaryInfC{$\indfun{i}{A}\indvdash I \indfun{i}{\Xi}$}
          \cutLine
          \BinaryInfC{$B\indvdash I \indfun{i}{\Gamma},\indfun{i}{\Xi}$}
          \AxiomC{$\pi_2$}
          \UnaryInfC{$C\indvdash I \indfun{j}{\Gamma},\indfun{j}{A}$}
          \AxiomC{$\pi_3$}
          \UnaryInfC{$ A\indvdash W \Xi$}
          \functorialLine
          \UnaryInfC{$ \indfun{j}{A}\indvdash J \indfun{j}{\Xi}$}
          \cutLine
          \BinaryInfC{$C\indvdash J \indfun{j}{\Gamma},\indfun{j}{j}{\Xi}$}
          \BinaryInfC{$B\oplusm{i}{j} C\indvdash {W} \Gamma,\Xi$}
          \DisplayProof    
\end{small}

    \longVOnly{
    Exponential trivial commutation rules:\fnote{todo}\\

    Multiplicative and additives trivial commutation rules:
    \begin{small}
      \begin{align*}
        \AxiomC{$\pi_1$}
        \UnaryInfC{$\Xi\indvdash I a$}
        \AxiomC{$a\in\mathtt{variable}(I)$}
        \UnaryInfC{$a\indvdash I a$}
        \cutLine
        \BinaryInfC{$\Xi\indvdash I a$}
        \DisplayProof
        & \reduce
        \AxiomC{$\pi_1$}
        \UnaryInfC{$\Xi\indvdash I a$}
        \DisplayProof\\
        \AxiomC{$\pi_1$}
        \UnaryInfC{$\Xi\indvdash I \bot$}
        \AxiomC{\vphantom{A}}
        \UnaryInfC{$\bot\indvdash I $}
        \cutLine
        \BinaryInfC{$\Xi\indvdash I $}
        \DisplayProof
        & \reduce
        \AxiomC{$\pi_1$}
        \UnaryInfC{$\Xi\indvdash I $}
        \DisplayProof\\
        \AxiomC{$\pi_1$}
        \UnaryInfC{$\Xi\indvdash I A$}
        \AxiomC{$\pi_2$}
        \UnaryInfC{$\Gamma,A\indvdash I B$}
        \AxiomC{$\pi_3$}
        \UnaryInfC{$\Delta\indvdash I C$}
        \BinaryInfC{$\Gamma,A,\Delta\indvdash I B\otimes C$}
        \cutLine
        \BinaryInfC{$\Xi,\Gamma,\Delta\indvdash I B\otimes C$}
        \DisplayProof
        & \reduce
        \AxiomC{$\pi_1$}
        \UnaryInfC{$\Xi\indvdash I A$}
        \AxiomC{$\pi_2$}
        \UnaryInfC{$\Gamma,A\indvdash I B$}
        \cutLine
        \BinaryInfC{$\Xi,\Gamma\indvdash I B$}
        \AxiomC{$\pi_3$}
        \UnaryInfC{$\Delta\indvdash I C$}
        \BinaryInfC{$\Xi,\Gamma,\Delta\indvdash I B\otimes C$}
        \DisplayProof\\
        \AxiomC{$\pi_1$}
        \UnaryInfC{$\Xi\indvdash I A$}
        \AxiomC{$\pi_2$}
        \UnaryInfC{$\Gamma\indvdash I B$}
        \AxiomC{$\pi_3$}
        \UnaryInfC{$\Delta,A\indvdash I C$}
        \BinaryInfC{$\Gamma,\Delta,A\indvdash I B\otimes C$}
        \cutLine
        \BinaryInfC{$\Xi,\Gamma,\Delta\indvdash I B\otimes C$}
        \DisplayProof
        & \reduce
        \AxiomC{$\pi_2$}
        \UnaryInfC{$\Gamma\indvdash I B$}
        \AxiomC{$\pi_1$}
        \UnaryInfC{$\Xi\indvdash I A$}
        \AxiomC{$\pi_3$}
        \UnaryInfC{$\Delta,A\indvdash I C$}
        \cutLine
        \BinaryInfC{$\Xi,\Delta\indvdash I C$}
        \BinaryInfC{$\Xi,\Gamma,\Delta\indvdash I B\otimes C$}
        \DisplayProof\\
        \AxiomC{$\pi_1$}
        \UnaryInfC{$\Xi\indvdash I A$}
        \AxiomC{$\pi_2$}
        \UnaryInfC{$\Gamma,A\indvdash I B,C$}
        \UnaryInfC{$\Gamma,A\indvdash I B\parr C$}
        \cutLine
        \BinaryInfC{$\Xi,\Gamma\indvdash I B\parr C$}
        \DisplayProof
        & \reduce
        \AxiomC{$\pi_1$}
        \UnaryInfC{$\Xi\indvdash I A$}
        \AxiomC{$\pi_2$}
        \UnaryInfC{$\Gamma,A\indvdash I B,C$}
        \cutLine
        \BinaryInfC{$\Xi,\Gamma,A\indvdash I B,C$}
        \UnaryInfC{$\Xi,\Gamma\indvdash I B\parr C$}
        \DisplayProof\\
        \AxiomC{$\pi_1$}
        \UnaryInfC{$\Xi\indvdash I A$}
        \AxiomC{$\pi_2$}
        \UnaryInfC{$\Gamma,A\indvdash I B$}
        \AxiomC{$0\vdash C\ \mathtt{def}$}
        \BinaryInfC{$\Gamma,A\indvdash I B\oplusm{\id}{\init} C$}
        \cutLine
        \BinaryInfC{$\Gamma,\Xi\indvdash I B\oplusm{\id}{\init} C$}
        \DisplayProof
        & \reduce
        \AxiomC{$\pi_1$}
        \UnaryInfC{$\Xi\indvdash I A$}
        \AxiomC{$\pi_2$}
        \UnaryInfC{$\Gamma,A\indvdash I B$}
        \cutLine
        \BinaryInfC{$\Gamma,\Xi\indvdash I B$}
        \AxiomC{$0\vdash C\ \mathtt{def}$}
        \BinaryInfC{$\Gamma,\Xi\indvdash I B\oplusm\id\init C$}
        \DisplayProof\\
        \AxiomC{$\pi_1$}
        \UnaryInfC{$\Xi\indvdash I A$}
        \AxiomC{$\pi_2$}
        \UnaryInfC{$\Gamma,A\indvdash I C$}
        \AxiomC{$0\vdash B\ \mathtt{def}$}
        \BinaryInfC{$\Gamma,A\indvdash I B\oplusm\init\id C$}
        \cutLine
        \BinaryInfC{$\Gamma,\Xi\indvdash I B\oplusm\init\id C$}
        \DisplayProof
        & \reduce
        \AxiomC{$\pi_1$}
        \UnaryInfC{$\Xi\indvdash I A$}
        \AxiomC{$\pi_2$}
        \UnaryInfC{$\Gamma,A\indvdash I C$}
        \cutLine
        \BinaryInfC{$\Gamma,\Xi\indvdash I C$}
        \AxiomC{$0\vdash B\ \mathtt{def}$}
        \BinaryInfC{$\Gamma,\Xi\indvdash I B\oplusm\init\id C$}
        \DisplayProof
      \end{align*}
    \end{small}
}
   }\shortVOnly{, as presented in Figure~\ref{fig:cut-elim}, where cuts are highlighted in \nonColorblind{red}\colorblind{orange}.
  }
\end{definition}
Notice that dashed rules, which are application of the two meta-rules of Lemma~\ref{lemma:BS_and_ST_rules}, are are structurally neutral, in that applying them will not change the structure of the proof tree, only the indexes and locci appearing in them. This is important to see that this procedure is strictly a refinement of that of linear logic.

We will need to name the reductions in order to track them through the application of base changes. In order to do so, we will point the corresponding reduction in the underlying LL proof. 

\begin{definition}[Pointed reducution]
  Let $\pi$ be a derivation, and $c$ be a specific cut of $\underline{\pi}$, we denote $\pi\stackrel{c}{\reduce}\pi'$ the elimination of this specific cut in $\pi$.
\end{definition}

The following lemma is critical for cut-elimination and pinpoint exactly where the equivalence relation is used for. It states that cuts can distributes over meta-rules of Lemma~\ref{lemma:BS_and_ST_rules}, but only up-to $(\eqtype)$.
\begin{theorem}[Simulation]\label{th:bisim_equiv}\ \\
  Let $\AxiomC{$\pi_1$}\UnaryInfC{$\Gamma\indvdash I A$}\DisplayProof\stackrel{c}{\reduce}\AxiomC{$\pi_2$}\UnaryInfC{$\Gamma\indvdash I A$}\DisplayProof$be a reduction, then for any $f$ and $\rho$ :
  \begin{itemize}
  \item $\AxiomC{$\pi_1$}\UnaryInfC{$\Gamma\indvdash I A$}\AxiomC{$\rho$}\UnaryInfC{$A\subtype A'$}\subtypLine\BinaryInfC{$\Gamma\indvdash I A'$}\DisplayProof  \stackrel{c}{\reduce}\AxiomC{$\pi'_2$}\UnaryInfC{$\Gamma\indvdash I A$}\DisplayProof\eqtype\AxiomC{$\pi_2$}\UnaryInfC{$\Gamma\indvdash I A$}\AxiomC{$\rho$}\UnaryInfC{$A\subtype A'$}\subtypLine\BinaryInfC{$\Gamma\indvdash I A'$}\DisplayProof$.
  \item $\AxiomC{$\pi_1$}\UnaryInfC{$\Gamma\indvdash I A$}\AxiomC{$g$}\functorialLine\BinaryInfC{$\indfun{g}{\Gamma}\indvdash K \indfun{g}{A}$}\DisplayProof  \stackrel{c}{\reduce}\AxiomC{$\pi''_2$}\UnaryInfC{$\indfun{g}{\Gamma}\indvdash K \indfun{g}{A}$}\DisplayProof\eqtype\AxiomC{$\pi_2$}\UnaryInfC{$\Gamma\indvdash I A$}\AxiomC{$g$}\functorialLine\BinaryInfC{$\indfun{g}{\Gamma}\indvdash K \indfun{g}{A}$}\DisplayProof$.
  \end{itemize}
\end{theorem}
\begin{proof}
  By induction on $\underline\pi$.\\
  The cases where $c$ is not the root are immediate induction on the unfolding of \AxiomC{$\pi_1$}\UnaryInfC{$\Gamma\indvdash I A$}\AxiomC{$\rho$}\UnaryInfC{$A\subtype A'$}\subtypLine\BinaryInfC{$\Gamma\indvdash I A'$}\DisplayProof and \AxiomC{$\pi_1$}\UnaryInfC{$\Gamma\indvdash I A$}\AxiomC{$f$}\functorialLine\BinaryInfC{$\indfun{f}{\Gamma}\indvdash I \indfun{f}{A}$}\DisplayProof.\\
  When $c$ is the root, we perform a case analysis on the each cut.
  \longVOnly{
    We will only develop the digging-dereliction elimination case, which exhibits all the difficulties of other cases, i.e., we assume that\\ $\pi_1=$\AxiomC{$\pi_3$}
    \UnaryInfC{$\indfun{w}{\bangm{u_1}A_1},\dots,\indfun{w}{\bangm{u_n}A_n}\indvdash {J} B$}
    \UnaryInfC{$\bangm{u_1}A_1,\dots,\bangm{u_n}A_n\indvdash I \bangm w B$}
    \AxiomC{$\pi_4$}
    \UnaryInfC{$\indfun{f}{B} \indvdash I \Gamma$}
    \AxiomC{$w\circ f = \id$}
    \BinaryInfC{$\bangm{w}B\indvdash I \Gamma$}
    \cutLine
    \BinaryInfC{$\bangm{u_1}A_1,\dots,\bangm{u_n}A_n\indvdash I \Gamma$}
    \DisplayProof and \\
    $\pi_2=$
    \AxiomC{$\pi_3$}
    \UnaryInfC{$\indfun{w}{\bangm{u_1}A_1},\dots,\indfun{w}{\bangm{u_n}A_n}\indvdash {J} B$}
    \functorialLine
    \UnaryInfC{$\indfun{f}{\indfun{w}{\bangm{u_1}A_1}},\dots,\indfun{f}{\indfun{w}{\bangm{u_n}A_n}}\indvdash I \indfun{f}{B}$}
    \AxiomC{$w\circ f = \id$}
    \subtypLine
    \BinaryInfC{$\bangm{u_1}A_1,\dots,\bangm{u_n}A_n\indvdash I \indfun{f}{B}$}
    \AxiomC{$\pi_4$}
    \UnaryInfC{$\indfun{f}{B} \indvdash I \Gamma$}
    \cutLine
    \BinaryInfC{$\bangm{u_1}A_1,\dots,\bangm{u_n}A_n\indvdash I \Gamma$}
    \DisplayProof

 For the subtyping case, there are two subcases depending on which formula the subtyping is applied to.

 If $\rho$ is subtyping a formula in $\Gamma$ this is immediate,
    If $\rho=$\AxiomC{$\rho'$}\UnaryInfC{$\indfun{g}{A_i}\subtype A'_i$}\UnaryInfC{$\bangm{u_i}A_i\subtype \bangm{u_i\circ g}A'_i$}\DisplayProof with $u_i=u_i\circ g$, then let\\
      $\rho''=$
      \AxiomC{$h;\pb{w}{ u_1} = \pb{w}{ (u_1\circ g)};g$}
      \AxiomC{$\rho'$}
      \UnaryInfC{$\indfun{g}{A_i}\subtype A'_i$}
      \BinaryInfC{$\indfun{h}{\indfun{\pb{w}{u_i}}{A_i}} \subtype \indfun{\pb{w}{(u_i\circ g)}}{A'_i}$}
      \UnaryInfC{$\bangm{\pb{{u_i}}{w}}\indfun{\pb{w}{u_i}}{A_i} \subtype \bangm{\pb{(u_i\circ g)}{w}}\indfun{\pb{w}{(u_i\circ g)}}{A'_i}$}
      \DisplayProof
      so that\\
      $\pi'_2=$
      \AxiomC{$\pi_3$}
      \UnaryInfC{$\indfun{w}{\bangm{u_1}A'_1},\dots,\indfun{w}{\bangm{u_n}A'_n}\indvdash {J} B$}
      \functorialLine
      \UnaryInfC{$\indfun{f}{\indfun{w}{\bangm{u_1}A'_1}},\dots,\indfun{f}{\indfun{w}{\bangm{u_i}A_i}},\dots,\indfun{f}{\indfun{w}{\bangm{u_n}A'_n}}\indvdash I \indfun{f}{B}$}
      \AxiomC{$\rho''$}
      \UnaryInfC{$\indfun{w}{\bangm{u_i}A_i} \subtype \indfun{w}{\bangm{u_i}A'_i}$}
      \UnaryInfC{$\indfun{f}{\indfun{w}{\bangm{u_i}A_i}} \subtype \indfun{f}{\indfun{w}{\bangm{u_i}A'_i}}$}
      \subtypLine
      \BinaryInfC{$\indfun{f}{\indfun{w}{\bangm{u_1}A'_1}},\dots,\indfun{f}{\indfun{w}{\bangm{u_i}A_i}},\dots,\indfun{f}{\indfun{w}{\bangm{u_n}A'_n}}\indvdash I \indfun{f}{B}$}
      \AxiomC{$w\circ f = \id$}
      \subtypLine
      \BinaryInfC{$\bangm{u_1}A'_1,\dots,\bangm{u_i}A_i,\dots,\bangm{u_n}A'_n\indvdash I \indfun{f}{B}$}
      \AxiomC{$\pi_4$}
      \UnaryInfC{$\indfun{f}{B} \indvdash I \Gamma$}
      \cutLine
      \BinaryInfC{$\bangm{u_1}A'_1,\dots,\bangm{u_i}A_i,\dots,\bangm{u_n}A'_n\indvdash I \Gamma$}
      \DisplayProof\\
      while\\
      \AxiomC{$\pi_2$}\UnaryInfC{$\Gamma\indvdash I A$}\AxiomC{$\rho$}\UnaryInfC{$A\subtype A'$}\subtypLine\BinaryInfC{$\Gamma\indvdash I A'$}\DisplayProof$=$
      \AxiomC{$\pi_3$}
      \UnaryInfC{$\indfun{w}{\bangm{u_1}A_1},\dots,\indfun{w}{\bangm{u_n}A_n}\indvdash {J} B$}
      \functorialLine
      \UnaryInfC{$\indfun{f}{\indfun{w}{\bangm{u_1}A_1}},\dots,\indfun{f}{\indfun{w}{\bangm{u_n}A_n}}\indvdash I \indfun{f}{B}$}
      \AxiomC{$w\circ f = \id$}
      \subtypLine
      \BinaryInfC{$\bangm{u_1}A_1,\dots,\bangm{u_n}A_n\indvdash I \indfun{f}{B}$}
      \AxiomC{$\rho'$}
      \subtypLine
      \BinaryInfC{$\bangm{u_1}A_1,\dots,\bangm{u_1}A'_i,\dots,\bangm{u_n}A_n\indvdash I \indfun{f}{B}$}
      \AxiomC{$\pi_4$}
      \UnaryInfC{$\indfun{f}{B} \indvdash I \Gamma$}
      \cutLine
      \BinaryInfC{$\bangm{u_1}A_1,\dots,\bangm{u_1}A'_i,\dots,\bangm{u_n}A_n\indvdash I \Gamma$}
      \DisplayProof\\
      The equivalence between those two is quite immediate as they only differ on the way subtyping is applied.
  
   For the base change case, we assume that there is a given $g:K\rightarrow I$ and\\
    \AxiomC{$\pi_1$}\UnaryInfC{$\Gamma\indvdash I A$}\AxiomC{$g$}\functorialLine\BinaryInfC{$\indfun{g}{\Gamma}\indvdash K \indfun{g}{A}$}\DisplayProof$=$
    \AxiomC{$\pi_3$}
    \UnaryInfC{$\indfun{w}{\bangm{u_1}A_1},\dots,\indfun{w}{\bangm{u_n}A_n}\indvdash {J} B$}
    \functorialLine
    \UnaryInfC{$\indfun{\pb{g}{w}}{\indfun{w}{\bangm{u_1}A_1}},\dots,\indfun{\pb{g}{w}}{\indfun{w}{\bangm{u_n}A_n}}\indvdash {J'} \indfun{\pb{g}{w}}{B}$}
    \subtypLine
    \UnaryInfC{$\indfun{\pb{w}{g}}{\indfun{g}{\bangm{u_1}A_1}},\dots,\indfun{\pb{w}{g}}{\indfun{g}{\bangm{u_n}A_n}}\indvdash {J'} \indfun{\pb{g}{w}}{B}$}
    \UnaryInfC{$\indfun{g}{\bangm{u_1}A_1},\dots,\indfun{g}{\bangm{u_n}A_n}\indvdash K \indfun{g}{\bangm w B}$}
    \AxiomC{$\pi_4$}
    \UnaryInfC{$\indfun{f}{B} \indvdash I \Gamma$}
    \functorialLine
    \UnaryInfC{$\indfun{g}{\indfun{f}{B}}\indvdash I \indfun{g}{\Gamma}$}
    \AxiomC{}
    \UnaryInfC{$g;f=h;\pb{g}{w}$}
    \UnaryInfC{$\indfun{g}{\indfun{f}{B}}\eqtype \indfun{(h;\pb{g}{w})}{B}$}
    \subtypLine
    \BinaryInfC{$\indfun{h}{\indfun{\pb{g}{w}}{B}}\indvdash I \indfun{g}{\Gamma}$}
    \AxiomC{}
    \UnaryInfC{$\pb{w}{g}\circ h = \id$}
    \BinaryInfC{$\bangm{\pb{w}{g}}\indfun{\pb{g}{w}}{B}\indvdash I \indfun{g}{\Gamma}$}
    \cutLine
    \BinaryInfC{$\indfun{g}{\bangm{u_1}A_1},\dots,\indfun{g}{\bangm{u_n}A_n}\indvdash I \indfun{g}{\Gamma}$}
    \DisplayProof\\
    where $h$ is the only function s.t. $\pb{g}{w}\circ h= g;f$ and $\pb{w}{g}\circ h= \id$ which exist since $w\circ f\circ g=g$.
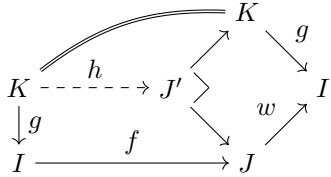
    \begin{center}
      \begin{tikzpicture}
        \node (K) at (0,0) {$K$};
        \node (I) at (0,-1) {$I$};
        \node (J) at (3,-1) {$J$};
        \node (I2) at (4,0) {$I$};
        \node (K2) at (3,1) {$K$};
        \node (J') at (2,0) {$J'$};
\draw[-] (2.3,.2) to (2.5,0);
        \draw[-] (2.3,-.2) to (2.5,0);
\draw[->] (K) to node[auto] {$g$} (I);
        \draw[->] (I) to node[auto] {$f$} (J);
        \draw[->] (J) to node[auto] {$w$} (I2);
        \draw[->] (K2) to node[auto] {$g$} (I2);
        \draw[double] (K) to[bend left =20] (K2);
        \draw[->] (J') to (J);
        \draw[->] (J') to (K2);
        \draw[->,dashed] (K) to node[auto] {$h$} (J');
      \end{tikzpicture}
    \end{center}
    Thus, computing $\pi''_2$ and applying proof-equivalence axioms, we get\\
    \begin{minipage}{\linewidth}
        \small
        \begin{align*}
          \pi''_2 &=
          \AxiomC{$\pi_3$}
          \UnaryInfC{$\indfun{w}{\bangm{u_1}A_1},\dots,\indfun{w}{\bangm{u_n}A_n}\indvdash {J} B$}
          \functorialLine
          \UnaryInfC{$\indfun{\pb{g}{w}}{\indfun{w}{\bangm{u_1}A_1}},\dots,\indfun{\pb{g}{w}}{\indfun{w}{\bangm{u_n}A_n}}\indvdash {J} \indfun{\pb{g}{w}}{B}$}
          \subtypLine
          \UnaryInfC{$\indfun{\pb{w}{g}}{\indfun{g}{\bangm{u_1}A_1}},\dots,\indfun{\pb{w}{g}}{\indfun{g}{\bangm{u_n}A_n}}\indvdash {J} \indfun{\pb{g}{w}}{B}$}
          \functorialLine
          \UnaryInfC{$\indfun{h}{\indfun{\pb{w}{g}}{\indfun{g}{\bangm{u_1}A_1}}},\dots,\indfun{h}{\indfun{\pb{w}{g}}{\indfun{g}{\bangm{u_n}A_n}}}\indvdash I \indfun{h}{\indfun{\pb{g}{w}}{B}}$}
          \AxiomC{$\pb{w}{g}\circ h = \id$}
          \subtypLine
          \BinaryInfC{$\indfun{g}{\bangm{u_1}A_1},\dots,\indfun{g}{\bangm{u_n}A_n}\indvdash K \indfun{h}{\indfun{\pb{g}{w}}{B}}$}
          \AxiomC{$\pi_4$}
          \UnaryInfC{$ \indfun{f}{B}\indvdash I \Gamma$}
          \functorialLine
          \UnaryInfC{$\indfun{g}{\indfun{f}{B}}\indvdash K \indfun{g}{\Gamma}$}
          \AxiomC{}
          \UnaryInfC{$g;f=h;\pb{g}{w}$}
          \UnaryInfC{$\indfun{h}{\indfun{\pb{g}{w}}{B}}\eqtype \indfun{g}{\indfun{f}{B}}$}
          \subtypLine
          \BinaryInfC{$\indfun{h}{\indfun{\pb{g}{w}}{B}} \indvdash K \indfun{g}{\Gamma}$}
          \cutLine
          \BinaryInfC{$\indfun{g}{\bangm{u_1}A_1},\dots,\indfun{g}{\bangm{u_n}A_n}\indvdash K \indfun{g}{\Gamma}$}
          \DisplayProof\\
          &\eqtype
          \AxiomC{$\pi_3$}
          \UnaryInfC{$\indfun{w}{\bangm{u_1}A_1},\dots,\indfun{w}{\bangm{u_n}A_n}\indvdash {J} B$}
          \functorialLine
          \UnaryInfC{$\indfun{\pb{g}{w}}{\indfun{w}{\bangm{u_1}A_1}},\dots,\indfun{\pb{g}{w}}{\indfun{w}{\bangm{u_n}A_n}}\indvdash {J} \indfun{\pb{g}{w}}{B}$}
          \subtypLine
          \UnaryInfC{$\indfun{\pb{w}{g}}{\indfun{g}{\bangm{u_1}A_1}},\dots,\indfun{\pb{w}{g}}{\indfun{g}{\bangm{u_n}A_n}}\indvdash {J} \indfun{\pb{g}{w}}{B}$}
          \functorialLine
          \UnaryInfC{$\indfun{h}{\indfun{\pb{w}{g}}{\indfun{g}{\bangm{u_1}A_1}}},\dots,\indfun{h}{\indfun{\pb{w}{g}}{\indfun{g}{\bangm{u_n}A_n}}}\indvdash I \indfun{h}{\indfun{\pb{g}{w}}{B}}$}
          \AxiomC{$\pb{w}{g}\circ h = \id$}
          \subtypLine
          \BinaryInfC{$\indfun{g}{\bangm{u_1}A_1},\dots,\indfun{g}{\bangm{u_n}A_n}\indvdash K \indfun{h}{\indfun{\pb{g}{w}}{B}}$}
          \AxiomC{}
          \UnaryInfC{$g;f=h;\pb{g}{w}$}
          \UnaryInfC{$\indfun{h}{\indfun{\pb{g}{w}}{B}}\eqtype \indfun{g}{\indfun{f}{B}}$}
          \subtypLine
          \BinaryInfC{$\indfun{g}{\bangm{u_1}A_1},\dots,\indfun{g}{\bangm{u_n}A_n}\indvdash K \indfun{g}{\indfun{f}{B}}$}
          \AxiomC{$\pi_4$}
          \UnaryInfC{$ \indfun{f}{B}\indvdash I \Gamma$}
          \functorialLine
          \UnaryInfC{$\indfun{g}{\indfun{f}{B}}\indvdash K \indfun{g}{\Gamma}$}
          \cutLine
          \BinaryInfC{$\indfun{g}{\bangm{u_1}A_1},\dots,\indfun{g}{\bangm{u_n}A_n}\indvdash K \indfun{g}{\Gamma}$}
          \DisplayProof\\
        \end{align*}
      \end{minipage}\\
      Then, by applying Lemma~\ref{lemma:proof_psedo_funct}.2,\\
      \begin{minipage}{\linewidth}
        \small
        \begin{align*}
          \pi''_2 &\eqtype
          \AxiomC{$\pi_3$}
          \UnaryInfC{$\indfun{w}{\bangm{u_1}A_1},\dots,\indfun{w}{\bangm{u_n}A_n}\indvdash {J} B$}
          \functorialLine
          \UnaryInfC{$\indfun{\pb{g}{w}}{\indfun{w}{\bangm{u_1}A_1}},\dots,\indfun{\pb{g}{w}}{\indfun{w}{\bangm{u_n}A_n}}\indvdash {J} \indfun{\pb{g}{w}}{B}$}
          \functorialLine
          \UnaryInfC{$\indfun{h}{\indfun{\pb{g}{w}}{\indfun{w}{\bangm{u_1}A_1}}},\dots,\indfun{h}{\indfun{\pb{g}{w}}{\indfun{w}{\bangm{u_n}A_n}}}\indvdash {J} \indfun{h}{\indfun{\pb{g}{w}}{B}}$}
          \subtypLine
          \UnaryInfC{$\indfun{h}{\indfun{\pb{w}{g}}{\indfun{g}{\bangm{u_1}A_1}}},\dots,\indfun{h}{\indfun{\pb{w}{g}}{\indfun{g}{\bangm{u_n}A_n}}}\indvdash I \indfun{h}{\indfun{\pb{g}{w}}{B}}$}
          \AxiomC{$\pb{w}{g}\circ h = \id$}
          \subtypLine
          \BinaryInfC{$\indfun{g}{\bangm{u_1}A_1},\dots,\indfun{g}{\bangm{u_n}A_n}\indvdash K \indfun{h}{\indfun{\pb{g}{w}}{B}}$}
          \AxiomC{}
          \UnaryInfC{$g;f=h;\pb{g}{w}$}
          \UnaryInfC{$\indfun{h}{\indfun{\pb{g}{w}}{B}}\eqtype \indfun{g}{\indfun{f}{B}}$}
          \subtypLine
          \BinaryInfC{$\indfun{g}{\bangm{u_1}A_1},\dots,\indfun{g}{\bangm{u_n}A_n}\indvdash K \indfun{g}{\indfun{f}{B}}$}
          \AxiomC{$\pi_4$}
          \UnaryInfC{$ \indfun{f}{B}\indvdash I \Gamma$}
          \functorialLine
          \UnaryInfC{$\indfun{g}{\indfun{f}{B}}\indvdash K \indfun{g}{\Gamma}$}
          \cutLine
          \BinaryInfC{$\indfun{g}{\bangm{u_1}A_1},\dots,\indfun{g}{\bangm{u_n}A_n}\indvdash K \indfun{g}{\Gamma}$}
          \DisplayProof\\
          &\eqtype
          \AxiomC{$\pi_3$}
          \UnaryInfC{$\indfun{w}{\bangm{u_1}A_1},\dots,\indfun{w}{\bangm{u_n}A_n}\indvdash {J} B$}
          \functorialLine
          \UnaryInfC{$\indfun{\pb{g}{w}}{\indfun{w}{\bangm{u_1}A_1}},\dots,\indfun{\pb{g}{w}}{\indfun{w}{\bangm{u_n}A_n}}\indvdash {J} \indfun{\pb{g}{w}}{B}$}
          \functorialLine
          \UnaryInfC{$\indfun{h}{\indfun{\pb{g}{w}}{\indfun{w}{\bangm{u_1}A_1}}},\dots,\indfun{h}{\indfun{\pb{g}{w}}{\indfun{w}{\bangm{u_n}A_n}}}\indvdash {J} \indfun{h}{\indfun{\pb{g}{w}}{B}}$}
          \AxiomC{$w\circ\pb{g}{w}\circ h = g$}
          \subtypLine
          \BinaryInfC{$\indfun{g}{\bangm{u_1}A_1},\dots,\indfun{g}{\bangm{u_n}A_n}\indvdash K \indfun{h}{\indfun{\pb{g}{w}}{B}}$}
          \AxiomC{}
          \UnaryInfC{$g;f=h;\pb{g}{w}$}
          \UnaryInfC{$\indfun{h}{\indfun{\pb{g}{w}}{B}}\eqtype \indfun{g}{\indfun{f}{B}}$}
          \subtypLine
          \BinaryInfC{$\indfun{g}{\bangm{u_1}A_1},\dots,\indfun{g}{\bangm{u_n}A_n}\indvdash K \indfun{g}{\indfun{f}{B}}$}
          \AxiomC{$\pi_4$}
          \UnaryInfC{$ \indfun{f}{B}\indvdash I \Gamma$}
          \functorialLine
          \UnaryInfC{$\indfun{g}{\indfun{f}{B}}\indvdash K \indfun{g}{\Gamma}$}
          \cutLine
          \BinaryInfC{$\indfun{g}{\bangm{u_1}A_1},\dots,\indfun{g}{\bangm{u_n}A_n}\indvdash K \indfun{g}{\Gamma}$}
          \DisplayProof
        \end{align*}
      \end{minipage}
      Then, by applying Lemma~\ref{lemma:proof_psedo_funct}.1,\\
      \begin{minipage}{\linewidth}
        \small
        \begin{align*}
          \pi''_2 &\eqtype
          \AxiomC{$\pi_3$}
          \UnaryInfC{$\indfun{w}{\bangm{u_1}A_1},\dots,\indfun{w}{\bangm{u_n}A_n}\indvdash {J} B$}
          \functorialLine
          \UnaryInfC{$\indfun{f}{\indfun{w}{\bangm{u_1}A_1}},\dots,\indfun{f}{\indfun{w}{\bangm{u_n}A_n}}\indvdash {J} \indfun{f}{B}$}
          \functorialLine
          \UnaryInfC{$\indfun{g}{\indfun{f}{\indfun{w}{\bangm{u_1}A_1}}},\dots,\indfun{g}{\indfun{f}{\indfun{w}{\bangm{u_n}A_n}}}\indvdash {J} \indfun{g}{\indfun{f}{B}}$}
          \subtypLine
          \UnaryInfC{$\indfun{h}{\indfun{\pb{g}{w}}{\indfun{w}{\bangm{u_1}A_1}}},\dots,\indfun{h}{\indfun{\pb{g}{w}}{\indfun{w}{\bangm{u_n}A_n}}}\indvdash {J} \indfun{h}{\indfun{\pb{g}{w}}{B}}$}
          \AxiomC{$w\circ\pb{g}{w}\circ h = g$}
          \subtypLine
          \BinaryInfC{$\indfun{g}{\bangm{u_1}A_1},\dots,\indfun{g}{\bangm{u_n}A_n}\indvdash K \indfun{h}{\indfun{\pb{g}{w}}{B}}$}
          \AxiomC{}
          \UnaryInfC{$g;f=h;\pb{g}{w}$}
          \UnaryInfC{$\indfun{h}{\indfun{\pb{g}{w}}{B}}\eqtype \indfun{g}{\indfun{f}{B}}$}
          \subtypLine
          \BinaryInfC{$\indfun{g}{\bangm{u_1}A_1},\dots,\indfun{g}{\bangm{u_n}A_n}\indvdash K \indfun{g}{\indfun{f}{B}}$}
          \AxiomC{$\pi_4$}
          \UnaryInfC{$ \indfun{f}{B}\indvdash I \Gamma$}
          \functorialLine
          \UnaryInfC{$\indfun{g}{\indfun{f}{B}}\indvdash K \indfun{g}{\Gamma}$}
          \cutLine
          \BinaryInfC{$\indfun{g}{\bangm{u_1}A_1},\dots,\indfun{g}{\bangm{u_n}A_n}\indvdash K \indfun{g}{\Gamma}$}
          \DisplayProof
        \end{align*}
      \end{minipage}
      Similarly, we have :\\
      \begin{minipage}{\linewidth}
        \small
        \begin{align*}
          \hspace*{-2cm}\AxiomC{$\pi_2$}\UnaryInfC{$\Gamma\indvdash I A$}\AxiomC{$g$}\functorialLine\BinaryInfC{$\indfun{g}{\Gamma}\indvdash I \indfun{g}{A}$}\DisplayProof
          &=
          \AxiomC{$\pi_3$}
          \UnaryInfC{$\indfun{w}{\bangm{u_1}A_1},\dots,\indfun{w}{\bangm{u_n}A_n}\indvdash {J} B$}
          \functorialLine
          \UnaryInfC{$\indfun{f}{\indfun{w}{\bangm{u_1}A_1}},\dots,\indfun{f}{\indfun{w}{\bangm{u_n}A_n}}\indvdash I \indfun{f}{B}$}
          \AxiomC{$w\circ f = \id$}
          \subtypLine
          \BinaryInfC{$\bangm{u_1}A_1,\dots,\bangm{u_n}A_n\indvdash I \indfun{f}{B}$}
          \functorialLine
          \UnaryInfC{$\indfun{g}{\bangm{u_1}A_1},\dots,\indfun{g}{\bangm{u_n}A_n}\indvdash K \indfun{g}{\indfun{f}{B}}$}
          \AxiomC{$\pi_4$}
          \UnaryInfC{$\indfun{f}{B} \indvdash I \Gamma$}
          \functorialLine
          \UnaryInfC{$\indfun{g}{\indfun{f}{B}} \indvdash K \indfun{g}{\Gamma}$}
          \cutLine
          \BinaryInfC{$\indfun{g}{\bangm{u_1}A_1},\dots,\indfun{g}{\bangm{u_n}A_n}\indvdash K \indfun{g}{\Gamma}$}
          \DisplayProof\\
          &\eqtype
          \AxiomC{$\pi_3$}
          \UnaryInfC{$\indfun{w}{\bangm{u_1}A_1},\dots,\indfun{w}{\bangm{u_n}A_n}\indvdash {J} B$}
          \functorialLine
          \UnaryInfC{$\indfun{f}{\indfun{w}{\bangm{u_1}A_1}},\dots,\indfun{f}{\indfun{w}{\bangm{u_n}A_n}}\indvdash I \indfun{f}{B}$}
          \functorialLine
          \UnaryInfC{$\indfun{g}{\indfun{f}{\indfun{w}{\bangm{u_1}A_1}}},\dots,\indfun{g}{\indfun{f}{\indfun{w}{\bangm{u_n}A_n}}}\indvdash K \indfun{g}{\indfun{f}{B}}$}
          \AxiomC{$w\circ f = \id$}
          \subtypLine
          \BinaryInfC{$\indfun{g}{\bangm{u_1}A_1},\dots,\indfun{g}{\bangm{u_n}A_n}\indvdash K \indfun{g}{\indfun{f}{B}}$}
          \AxiomC{$\pi_4$}
          \UnaryInfC{$\indfun{f}{B} \indvdash I \Gamma$}
          \functorialLine
          \UnaryInfC{$\indfun{g}{\indfun{f}{B}} \indvdash K \indfun{g}{\Gamma}$}
          \cutLine
          \BinaryInfC{$\indfun{g}{\bangm{u_1}A_1},\dots,\indfun{g}{\bangm{u_n}A_n}\indvdash K \indfun{g}{\Gamma}$}
          \DisplayProof
        \end{align*}
      \end{minipage}
      which is equivalent.
  }
\end{proof}

\subsection{Confluence of the cut-elimination}
It is not difficult, even though quite heavy, to compute the normal form of a proof containing two nested dereliction-promotion cuts. We generally obtain two similar proofs but the internal loci are different, even though equivalent, because the pullback have been taken in a different order. It remains an open question weather it is possible to obtain two equivalent proofs where the loci are not pointwise equivalent.

\begin{definition}[Parallel reduction]\ \\
  The relation $\pi\preduce\pi'$ is the top-to-bottom cut-elimination of an arbitrary number of cuts in $\pi$. It can be defined inductively by the reflexive axiom $\pi\preduce \pi$ and all cut elimination steps composed with a reduction on each subproof.

  We only give one case that serves as example :\\
  if $\pi_1\preduce\pi_1'$, $\pi_2\preduce\pi_2'$, and $\pi_3\preduce\pi_3'$, then\\
  \shortVOnly{\def\ScoreOverhang{1pt}}
        \AxiomC{$\pi_1$}
        \UnaryInfC{$B\indvdash I i(\Gamma),i(A)$}
          \AxiomC{$\pi_2$}
          \UnaryInfC{$C\indvdash J j(\Gamma),j(A)$}
          \BinaryInfC{$B\oplusm{i}{j} C\indvdash {W} \Gamma,A$}
          \AxiomC{$\pi_3$}
          \UnaryInfC{$A\indvdash {W} \Xi$}
          \cutLine
          \BinaryInfC{$B\oplusm{i}{j} C\indvdash {W} \Gamma,\Xi$}
          \DisplayProof \quad $\preduce$\par
    {\raggedleft
          \AxiomC{$\pi'_1$}
          \UnaryInfC{$B\indvdash I i(\Gamma),i(A)$}
          \AxiomC{$\pi'_3$}
          \UnaryInfC{$A\indvdash W \Xi$}
          \functorialLine
          \UnaryInfC{$i(A)\indvdash I i(\Xi)$}
          \cutLine
          \BinaryInfC{$B\indvdash I i(\Gamma),i(\Xi)$}
          \AxiomC{$\pi'_2$}
          \UnaryInfC{$C\indvdash I j(\Gamma),j(A)$}
          \AxiomC{$\pi'_3$}
          \UnaryInfC{$ A\indvdash W \Xi$}
          \functorialLine
          \UnaryInfC{$ j(A)\indvdash Jj(\Xi)$}
          \cutLine
          \BinaryInfC{$C\indvdash J j(\Gamma),j(\Xi)$}
          \BinaryInfC{$B\oplusm{i}{j} C\indvdash {W} \Gamma,\Xi$}
          \DisplayProof\par
      }
  \shortVOnly{\def\ScoreOverhang{4pt}}
\end{definition}

\begin{lemma}
  If $\pi\preduce\sigma$ by reducing more redexes than $\pi\preduce\rho$, then $\rho\preduce \tau$ with $\sigma\eqtype\tau$.
\end{lemma}
\begin{proof}
  By induction on $\underline{\pi}$ with three cases for each cut (no reduction, reducing in $\sigma$ only, reducing in both). All cases are immediate excepts for those where the reduction use a base change (dashed blue rules) or an equivalence (dashed green rules). In those cases we have to use Theorem~\ref{th:bisim_equiv}. As an example, we will detail the case of promotion vs dereliction, which uses both:\\
  If $\pi=$
  \AxiomC{$\pi_1$}
  \UnaryInfC{$w(\bangm{u_1}A_1),\dots,w(\bangm{u_n}A_n)\indvdash {J} B$}
  \UnaryInfC{$\bangm{u_1}A_1,\dots,\bangm{u_n}A_n\indvdash I \bangm w B$}
  \AxiomC{$\pi_2$}
  \UnaryInfC{$f(B) \indvdash I \Gamma$}
  \AxiomC{$f;w = \id$}
  \BinaryInfC{$\bangm{w}B\indvdash I \Gamma$}
  \cutLine
  \BinaryInfC{$\bangm{u_1}A_1,\dots,\bangm{u_n}A_n\indvdash I \Gamma$}
  \DisplayProof
  Then :
  \begin{itemize}[noitemsep,topsep=0pt,parsep=0pt,partopsep=0pt,leftmargin=0em,itemindent=1em]
  \item Either the root cut is reduced by none so that $\sigma$ is the proof\\
    \AxiomC{$\sigma_1$}
    \UnaryInfC{$w(\bangm{u_1}A_1),\dots,w(\bangm{u_n}A_n)\indvdash {J} B$}
    \UnaryInfC{$\bangm{u_1}A_1,\dots,\bangm{u_n}A_n\indvdash I \bangm w B$}
    \AxiomC{$\sigma_2$}
    \UnaryInfC{$f(B) \indvdash I \Gamma$}
    \AxiomC{$f;w = \id$}
    \BinaryInfC{$\bangm{w}B\indvdash I \Gamma$}
    \cutLine
    \BinaryInfC{$\bangm{u_1}A_1,\dots,\bangm{u_n}A_n\indvdash I \Gamma$}
    \DisplayProof\\
    and $\rho$ is the proof\\
    \AxiomC{$\rho_1$}
    \UnaryInfC{$w(\bangm{u_1}A_1),\dots,w(\bangm{u_n}A_n)\indvdash {J} B$}
    \UnaryInfC{$\bangm{u_1}A_1,\dots,\bangm{u_n}A_n\indvdash I \bangm w B$}
    \AxiomC{$\rho_2$}
    \UnaryInfC{$f(B) \indvdash I \Gamma$}
    \AxiomC{$f;w = \id$}
    \BinaryInfC{$\bangm{w}B\indvdash I \Gamma$}
    \cutLine
    \BinaryInfC{$\bangm{u_1}A_1,\dots,\bangm{u_n}A_n\indvdash I \Gamma$}
    \DisplayProof\\
    with $\pi_i\preduce \sigma_i$ and $\pi_i\preduce \rho_i$ and by reducing strictly less cuts in the second. In this case, we can apply the IH to find $\tau_i$ such that $\rho_i\preduce \tau_i$ and $\tau_i\eqtype \sigma_i$, then $\tau$ is the proof\\
    \AxiomC{$\tau_1$}
    \UnaryInfC{$w(\bangm{u_1}A_1),\dots,w(\bangm{u_n}A_n)\indvdash {J} B$}
    \UnaryInfC{$\bangm{u_1}A_1,\dots,\bangm{u_n}A_n\indvdash I \bangm w B$}
    \AxiomC{$\tau_2$}
    \UnaryInfC{$f(B) \indvdash I \Gamma$}
    \AxiomC{$f;w = \id$}
    \BinaryInfC{$\bangm{w}B\indvdash I \Gamma$}
    \cutLine
    \BinaryInfC{$\bangm{u_1}A_1,\dots,\bangm{u_n}A_n\indvdash I \Gamma$}
    \DisplayProof.
  \item or the root cut is reduced by only in $\sigma$, so that $\sigma$ is the proof\\
\AxiomC{$\sigma_1$}
    \UnaryInfC{$w(\bangm{u_1}A_1),\dots,w(\bangm{u_n}A_n)\indvdash {J} B$}
    \functorialLine
    \UnaryInfC{\shortVOnly{\hspace*{-.5em}}$f(w(\bangm{u_1}A_1)),\dots,f(w(\bangm{u_n}A_n))\indvdash I f(B)$}
    \AxiomC{\shortVOnly{\hspace*{-1.2em}}$f;w = \id$\shortVOnly{\hspace*{-.5em}}}
    \subtypLine
    \BinaryInfC{$\bangm{u_1}A_1,\dots,\bangm{u_n}A_n\indvdash I f(B)$}
    \AxiomC{$\sigma_2$}
    \UnaryInfC{\shortVOnly{\hspace*{-.5em}}$f(B) \indvdash I \Gamma$\shortVOnly{\hspace*{-.5em}}}
    \cutLine
    \BinaryInfC{$\bangm{u_1}A_1,\dots,\bangm{u_n}A_n\indvdash I \Gamma$}
    \DisplayProof\\
    and $\rho$ is the proof\\
    \AxiomC{$\rho_1$}
    \UnaryInfC{$w(\bangm{u_1}A_1),\dots,w(\bangm{u_n}A_n)\indvdash {J} B$}
    \UnaryInfC{$\bangm{u_1}A_1,\dots,\bangm{u_n}A_n\indvdash I \bangm w B$}
    \AxiomC{$\rho_2$}
    \UnaryInfC{$f(B) \indvdash I \Gamma$}
    \AxiomC{$f;w = \id$}
    \BinaryInfC{$\bangm{w}B\indvdash I \Gamma$}
    \cutLine
    \BinaryInfC{$\bangm{u_1}A_1,\dots,\bangm{u_n}A_n\indvdash I \Gamma$}
    \DisplayProof\\
    with $\pi_i\preduce \sigma_i$ and $\pi_i\preduce \rho_i$ by reducing strictly less cuts in the second, in which case, we can apply the IH to fine $\tau_i$ such that $\rho_i\preduce \tau_i$ and $\tau_i\eqtype \sigma_i$.\\
    We now fix $\tau$ as the proof\\
\AxiomC{$\tau_1$}
    \UnaryInfC{$w(\bangm{u_1}A_1),\dots,w(\bangm{u_n}A_n)\indvdash {J} B$}
    \functorialLine
    \UnaryInfC{\shortVOnly{\hspace*{-.5em}}$f(w(\bangm{u_1}A_1)),\dots,f(w(\bangm{u_n}A_n))\indvdash I f(B)$}
    \AxiomC{\shortVOnly{\hspace*{-1.2em}}$f;w = \id$\shortVOnly{\hspace*{-.5em}}}
    \subtypLine
    \BinaryInfC{$\bangm{u_1}A_1,\dots,\bangm{u_n}A_n\indvdash I f(B)$}
    \AxiomC{$\tau_2$}
    \UnaryInfC{\shortVOnly{\hspace*{-.5em}}$f(B) \indvdash I \Gamma$\shortVOnly{\hspace*{-.5em}}}
    \cutLine
    \BinaryInfC{$\bangm{u_1}A_1,\dots,\bangm{u_n}A_n\indvdash I \Gamma$}
    \DisplayProof
    \\
    By Theorem~\ref{th:equiv_base_change},\\
    \AxiomC{$\sigma_1$}
    \UnaryInfC{$w(\bangm{u_1}A_1),\dots,w(\bangm{u_n}A_n)\indvdash {J} B$}
    \functorialLine
    \UnaryInfC{$f(w(\bangm{u_1}A_1)),\dots,f(w(\bangm{u_n}A_n))\indvdash I f(B)$}
    \DisplayProof$=f(\sigma_1)$\par
    {\raggedleft
    $\eqtype \ f(\tau_1)=$
    \AxiomC{$\tau_1$}
    \UnaryInfC{$w(\bangm{u_1}A_1),\dots,w(\bangm{u_n}A_n)\indvdash {J} B$}
    \functorialLine
    \UnaryInfC{$f(w(\bangm{u_1}A_1)),\dots,f(w(\bangm{u_n}A_n))\indvdash I f(B)$}
    \DisplayProof \par}
    and by transitivity of $(\eqtype)$,\\
    \AxiomC{$\sigma_1$}
    \UnaryInfC{$w(\bangm{u_1}A_1),\dots,w(\bangm{u_n}A_n)\indvdash {J} B$}
    \functorialLine
    \UnaryInfC{$f(w(\bangm{u_1}A_1)),\dots,f(w(\bangm{u_n}A_n))\indvdash I f(B)$}
    \AxiomC{$f;w = \id$}
    \subtypLine
    \BinaryInfC{$\bangm{u_1}A_1,\dots,\bangm{u_n}A_n\indvdash I f(B)$}
    \DisplayProof\par
    {\raggedleft
    $\eqtype$
    \AxiomC{$\tau_1$}
    \UnaryInfC{$w(\bangm{u_1}A_1),\dots,w(\bangm{u_n}A_n)\indvdash {J} B$}
    \functorialLine
    \UnaryInfC{$f(w(\bangm{u_1}A_1)),\dots,f(w(\bangm{u_n}A_n))\indvdash I f(B)$}
    \AxiomC{$f;w = \id$}
    \subtypLine
    \BinaryInfC{$\bangm{u_1}A_1,\dots,\bangm{u_n}A_n\indvdash I f(B)$}
    \DisplayProof\par}
    Thus $\tau\eqtype\sigma$
  \item or the root cut is reduced by both, so that $\sigma$ is the proof\\
\AxiomC{$\sigma_1$}
    \UnaryInfC{$w(\bangm{u_1}A_1),\dots,w(\bangm{u_n}A_n)\indvdash {J} B$}
    \functorialLine
    \UnaryInfC{\shortVOnly{\hspace*{-.5em}}$f(w(\bangm{u_1}A_1)),\dots,f(w(\bangm{u_n}A_n))\indvdash I f(B)$}
    \AxiomC{\shortVOnly{\hspace*{-1.2em}}$f;w = \id$\shortVOnly{\hspace*{-.5em}}}
    \subtypLine
    \BinaryInfC{$\bangm{u_1}A_1,\dots,\bangm{u_n}A_n\indvdash I f(B)$}
    \AxiomC{$\sigma_2$}
    \UnaryInfC{\shortVOnly{\hspace*{-.5em}}$f(B) \indvdash I \Gamma$\shortVOnly{\hspace*{-.5em}}}
    \cutLine
    \BinaryInfC{$\bangm{u_1}A_1,\dots,\bangm{u_n}A_n\indvdash I \Gamma$}
    \DisplayProof\\
    and $\rho$ is the proof\\
\AxiomC{$\rho_1$}
    \UnaryInfC{$w(\bangm{u_1}A_1),\dots,w(\bangm{u_n}A_n)\indvdash {J} B$}
    \functorialLine
    \UnaryInfC{\shortVOnly{\hspace*{-.5em}}$f(w(\bangm{u_1}A_1)),\dots,f(w(\bangm{u_n}A_n))\indvdash I f(B)$\shortVOnly{\hspace*{-.5em}}}
    \AxiomC{\shortVOnly{\hspace*{-1.2em}}$f;w = \id$\shortVOnly{\hspace*{-.5em}}}
    \subtypLine
    \BinaryInfC{$\bangm{u_1}A_1,\dots,\bangm{u_n}A_n\indvdash I f(B)$}
    \AxiomC{$\rho_2$}
    \UnaryInfC{\shortVOnly{\hspace*{-.5em}}$f(B) \indvdash I \Gamma$\shortVOnly{\hspace*{-.5em}}}
    \cutLine
    \BinaryInfC{$\bangm{u_1}A_1,\dots,\bangm{u_n}A_n\indvdash I \Gamma$}
    \DisplayProof\\
    with $\pi_i\preduce \sigma_i$ and $\pi_i\preduce \rho_i$ by reducing strictly less cuts in the second, in which case, we can apply the IH to fine $\tau_i$ such that $\rho_i\preduce \tau_i$ and $\tau_i\eqtype \sigma_i$.\\
    By Theorems~\ref{th:bisim_equiv}, there is $\tau_1'$ such that\\
    \AxiomC{$\rho_1$}
    \UnaryInfC{$w(\bangm{u_1}A_1),\dots,w(\bangm{u_n}A_n)\indvdash {J} B$}
    \functorialLine
    \UnaryInfC{$f(w(\bangm{u_1}A_1)),\dots,f(w(\bangm{u_n}A_n))\indvdash I f(B)$}
    \DisplayProof $\preduce \tau_1'$\par
    {\raggedleft
    $\eqtype$ 
    \AxiomC{$\tau_1$}
    \UnaryInfC{$w(\bangm{u_1}A_1),\dots,w(\bangm{u_n}A_n)\indvdash {J} B$}
    \functorialLine
    \UnaryInfC{$f(w(\bangm{u_1}A_1)),\dots,f(w(\bangm{u_n}A_n))\indvdash I f(B)$}
    \DisplayProof.\par}
    By Theorems~\ref{th:bisim_equiv}, there is $\tau_1''$ such that\\
    \AxiomC{$\rho_1$}
    \UnaryInfC{$w(\bangm{u_1}A_1),\dots,w(\bangm{u_n}A_n)\indvdash {J} B$}
    \functorialLine
    \UnaryInfC{$f(w(\bangm{u_1}A_1)),\dots,f(w(\bangm{u_n}A_n))\indvdash I f(B)$}
    \AxiomC{$f;w = \id$}
    \subtypLine
    \BinaryInfC{$\bangm{u_1}A_1,\dots,\bangm{u_n}A_n\indvdash I f(B)$}
    \DisplayProof $\preduce \tau_1''$\par
    {\raggedleft
    $\eqtype$ 
    \AxiomC{$\tau_1$}
    \UnaryInfC{$w(\bangm{u_1}A_1),\dots,w(\bangm{u_n}A_n)\indvdash {J} B$}
    \functorialLine
    \UnaryInfC{$f(w(\bangm{u_1}A_1)),\dots,f(w(\bangm{u_n}A_n))\indvdash I f(B)$}
    \AxiomC{$f;w = \id$}
    \subtypLine
    \BinaryInfC{$\bangm{u_1}A_1,\dots,\bangm{u_n}A_n\indvdash I f(B)$}
    \DisplayProof.\par}
    We fix $\tau=$
    \AxiomC{$\tau''_1$}
    \UnaryInfC{$\bangm{u_1}A_1,\dots,\bangm{u_n}A_n\indvdash I f(B)$}
    \AxiomC{$\tau_2$}
    \UnaryInfC{$f(B) \indvdash I \Gamma$}
    \cutLine
    \BinaryInfC{$\bangm{u_1}A_1,\dots,\bangm{u_n}A_n\indvdash I \Gamma$}
    \DisplayProof
    \\
    By Theorems~\ref{th:equiv_base_change},\\
    \AxiomC{$\sigma_1$}
    \UnaryInfC{$w(\bangm{u_1}A_1),\dots,w(\bangm{u_n}A_n)\indvdash {J} B$}
    \functorialLine
    \UnaryInfC{$f(w(\bangm{u_1}A_1)),\dots,f(w(\bangm{u_n}A_n))\indvdash I f(B)$}
    \DisplayProof
    $=f(\sigma_1)$\shortVOnly{\par}
    {\shortVOnly{\raggedleft}
    $\eqtype f(\tau_1)=$
    \AxiomC{$\tau_1$}
    \UnaryInfC{$w(\bangm{u_1}A_1),\dots,w(\bangm{u_n}A_n)\indvdash {J} B$}
    \functorialLine
    \UnaryInfC{$f(w(\bangm{u_1}A_1)),\dots,f(w(\bangm{u_n}A_n))\indvdash I f(B)$}
    \DisplayProof\par}
    and by transitivity of $(\eqtype)$,\\
    \AxiomC{$\sigma_1$}
    \UnaryInfC{$w(\bangm{u_1}A_1),\dots,w(\bangm{u_n}A_n)\indvdash {J} B$}
    \functorialLine
    \UnaryInfC{$f(w(\bangm{u_1}A_1)),\dots,f(w(\bangm{u_n}A_n))\indvdash I f(B)$}
    \AxiomC{$w\circ f = \id$}
    \subtypLine
    \BinaryInfC{$\bangm{u_1}A_1,\dots,\bangm{u_n}A_n\indvdash I f(B)$}
    \DisplayProof\par
    {\raggedleft
    $\eqtype$
    \AxiomC{$\tau_1$}
    \UnaryInfC{$w(\bangm{u_1}A_1),\dots,w(\bangm{u_n}A_n)\indvdash {J} B$}
    \functorialLine
    \UnaryInfC{$f(w(\bangm{u_1}A_1)),\dots,f(w(\bangm{u_n}A_n))\indvdash I f(B)$}
    \AxiomC{$w\circ f = \id$}
    \subtypLine
    \BinaryInfC{$\bangm{u_1}A_1,\dots,\bangm{u_n}A_n\indvdash I f(B)$}
    \DisplayProof
    \par}
    Thus, using the transitivity of $\eqtype$, we have $\tau\eqtype\sigma$
  \end{itemize}
\end{proof}

\begin{theorem}[`Up-to' parallel strong confluence]\longVOnly{\ \\}
  If $\pi\preduce\pi_1$ and $\pi\preduce\pi_2$ then $\pi\preduce\pi'$, $\pi_1\preduce\pi'_1$ and $\pi_2\preduce\pi'_2$ with $\pi'\eqtype\pi_1'\eqtype \pi_2'$.
\end{theorem}
\begin{proof}
  The proof $\pi'$ is the parallel elimination of all cuts eliminated by either $\pi\preduce\pi_1$ or $\pi\preduce\pi_2$.
\end{proof}

\begin{theorem}
  \shortVOnly{This cut-elimination of $\IndLL$ is confluent, and normalising, up-to $(\eqtype)$.}
  \longVOnly{Ind$_\wedge\mu$LL, the fragment restricted to Park's fixpoint, is normalisable up-to $\eqtype$.}
\end{theorem}
\begin{proof}
  The confluence is obtained as usual from the strong confluence of $\preduce$ using that $(\reduce)\subseteq(\preduce)\subseteq(\reduce^*)$.\\
  Normalisation, or rather termination, comes from the normalisation result for cut elimination in LL, through the projection of $\IndLL$ cut elimination steps into LL ones.
\end{proof}

\section{Towards denotational semantics}\label{sec:sem}
We give some intuition on the denotational semantics for our $\IndLL . $ 
We have argued that a formula should be seen as an indexed family of idempotent intersection types. If we want to make this intuition semantically solid, we shall first look at the \emph{Scott semantics} of linear logic~\cite{er:scott}, which is known to express idempotent ITs.
Since our focus is on formulae defined over loci, is very natural to seek the interpretation of formuale in appropriate comma-like categories constructed from the Scott semantics framework.

For any set $I$, we define a category $\mathtt{ScottL}^{\wedge I}$ which objects represent formulae under the locus $I$ and which morphisms represent proof of sequents under $I$.
\begin{itemize}
\item Objects of $\mathtt{ScottL}^{\wedge I}$ are functions $A:I\rightarrow\underline{A}$ for $(\underline{A},\sqsubseteq_A)$ a preorder,
\item Morphisms $p\in\mathtt{ScottL}^{\wedge I}(A,B)$ are relations $r\supseteq\{(A(x),B(x))\mid x\in I\}$ which are monotonic, i.e., $a'\sqsupseteq arb\supseteq b'$ implies $a'rb'$.
\end{itemize}
The interpretation of formulae is then given so that $\sem{I}{A} : I\rightarrow \sem{\mathtt{ScottL}}{\underline A}$ by:
\begin{itemize}
\item $\sem{I}{\1}:I\rightarrow \{*\}$ is the terminal arrow,
\item $\sem{I}{A\otimes B} := \langle \sem{I}{A},\sem{I}{B}\rangle $ is the pairing,
\item $\sem\emptyset{\0} := \id_\emptyset$,
\item $\llbracket A\oplusm ij B\rrbracket_K:= \inv{[i,j]};(\sem{I}{A}\oplus\sem{J}{A})$,
\item $\sem{J}{\bangm u A}= (x\mapsto \{\sem{I}{A}(y)\mid u(y)=x\})$,
\item $\sem{I}{A^\bot}=\sem{I}{A}$ (but the preorder in the target is different).
\end{itemize}
\begin{proposition}
$\quad \sem{I}{\indfun{f}{A}} = f;\sem{J}{A} $
\end{proposition}
\longVOnly{
  \begin{proof}
    By induction on $A$:
    \begin{align*}
      \sem{I}{\indfun{f}{A\otimes B}}
      &= \langle f;\sem{J}{A},f;\sem{J}{B}\rangle\\
      &= f;\langle \sem{J}{A},\sem{J}{B}\rangle\\
      \sem{K}{\indfun{f}{A\oplusm ij B}}
      &= \inv{[\pb{i}{f},\pb{j}{f}]};(\pb{f}{i};\sem{I}{A})\oplus(\pb{f}{j};\sem{J}{A})\\
      &= \inv{[\pb{i}{f},\pb{j}{f}]}; (\pb{f}{i}\otimes\pb{f}{j});(\sem{I}{A}\oplus\sem{J}{A})\\
      &= \inv{[i,j]}; (\sem{I}{A}\oplus\sem{J}{A})\\
      \sem{I}{\indfun{f}{\bangm u A}}
      &= (x\mapsto \{\sem{I}{A}(\pb{f}{u}(y))\mid \pb{u}{f}(y)=x\})\\
      &= (x\mapsto \{\sem{I}{A}(\pb{f}{u}(y',x'))\mid \pb{u}{f}(x',y')=x, u(y')=f(x')\})\\
      &= (x\mapsto \{\sem{I}{A}(y')\mid u(y')=f(x)\})
    \end{align*}
  \end{proof}
}
The interpretation of the a proof  \AxiomC{$\pi$}\UnaryInfC{$\!\Gamma\vdash_I\Delta\!$}\DisplayProof is simply given by the ScottL interpretation of its target, i.e., $\sem{I}{\pi} = \sem{\mathtt{ScottL}}{\underline{\pi}}$, which happen to be a morphism of $\mathtt{ScottL}^{\wedge I}(\bigotimes \Gamma,\bigparr\Delta)$

\longVOnly{
  By induction on \AxiomC{$\pi$}\UnaryInfC{$\!\Gamma\vdash_I\Delta\!$}\DisplayProof, we show that for all $x\in I$ we have $(\sem{I}{\bigotimes \Gamma}(x),\sem{I}{\bigparr\Delta}(x))\in \sem{\mathtt{ScottL}}{\underline{\pi}}$.
  \begin{proof}
    \begin{itemize}
    \item If $(\sem{I}{\bigotimes \Gamma}(x),\sem{I}{A}(x))\in \sem{\mathtt{ScottL}}{\underline{\pi}}$ and $(\sem{I}{A}(x),\sem{I}{\bigparr\Delta}(x))\in \sem{\mathtt{ScottL}}{\underline{\pi'}}$ then $(\sem{I}{\bigotimes \Gamma}(x),\sem{I}{\bigparr\Delta}(x))$ is in the composition.
    \item if $(\sem{I}{\bigotimes \Gamma}(x),\sem{I}{A}(x))\in \sem{\mathtt{ScottL}}{\underline{\pi}}$ and $(\sem{I}{\bigotimes \Delta}(x),\sem{I}{B}(x))\in \sem{\mathtt{ScottL}}{\underline{\pi'}}$ then $((\sem{I}{\bigotimes \Gamma}(x),(\sem{I}{\bigotimes \Delta}(x)),(\sem{I}{A}(x),\sem{I}{B}(x)))\in \sem{\mathtt{ScottL}}{\underline{\pi}}\otimes \sem{\mathtt{ScottL}}{\underline{\pi'}}$, we conclude since $\sem{I}{A'\otimes B'}(x)=(\sem{I}{A'}(x),\sem{I}{B'}(x))$,
    \item if for all $x\in I$, $(\sem{I}{\bigotimes \Gamma}(i(x)),\sem{I}{A}(x))\in \sem{\mathtt{ScottL}}{\underline{\pi}}$ and if for all $y\in J$ $(\sem{J}{\bigotimes \Gamma}(j(y)),\sem{J}{B}(y))\in \sem{\mathtt{ScottL}}{\underline{\pi'}}$ then for all $z\in K$, either $z=i(x)$ for some $x$ and $\inv{[i,j]}(x)=z$ so that $(\sem{I}{\bigotimes \Gamma}(i(x)),(\sem{I}{A}\oplus\sem{J}{A})(\inv{[i,j]}(x)))\in [\sem{\mathtt{ScottL}}{\underline{\pi}},\sem{\mathtt{ScottL}}{\underline{\pi'}}]$.
    \item if for all $x\in I$, $(\sem{I}{\bigotimes \Gamma}(\inv{i}(x)),\sem{I}{A}(x))\in \sem{\mathtt{ScottL}}{\underline{\pi}}$ then for all $y$, $(\iota_1(\sem{I}{\bigotimes \Gamma}(\inv{i}(i(y))),(\sem{I}{A}\oplus\sem{J}{A})(\iota_1(i(y))))\in \sem{\mathtt{ScottL}}{\underline{\pi}};\mathtt{inj}_1$ we conclude since $\inv{i}(i(y))=y$ and $\iota_1(i(y))=\inv{[i,\init]}(y)$.
    \item if $(\sem{I}{\bigotimes \Gamma}(x),\sem{I}{\wnm{w}A}(x)\parr\sem{I}{\wnm{w}A}(x))\in \sem{\mathtt{ScottL}}{\underline{\pi}}$ then since $\sem{I}{\wnm{w}A}(x)=\{\sem{J}{A}(y)\mid u(y)=x\}$, we have $((\sem{I}{\wnm{w}A}(x),\sem{I}{\wnm{w}A}(x)),\sem{I}{\wnm{w}A}(x))\in c$ the contraction, thus the result.
    \item if $(\sem{I}{\bigotimes \Gamma}(x),\sem{I}{\indfun{f}{B}}(x))\in \sem{\mathtt{ScottL}}{\underline{\pi}}$ and $u\circ f=\id$, then, since $\sem{I}{\indfun{f}{B}}(x)=\sem{J}{B}(f(x))$ and $\sem{I}{\wnm{u}{B}}(x) = \{\sem{y}{B}(y)\mid u(y)=x\}$, $\sem{I}{\indfun{f}{B}}(x)\in \sem{I}{\wnm{u}{B}}(x)$ thus $(\sem{I}{\indfun{f}{B}}(x),\sem{I}{\wnm{u}{B}}(x))\in d$ the dereliction.
    \item if $(\sem{I}{\bigotimes \Gamma}(x),*)\in \sem{\mathtt{ScottL}}{\underline{\pi}}$ then $(\sem{I}{\bigotimes \Gamma}(x),a)\in \sem{\mathtt{ScottL}}{\underline{\pi}};w$ for all $a$, in particular for $a=\sem{I}{\wnm{u}{B}}(x)$.
    \item if, for all $x$, $(\sem{J}{\bigotimes \Gamma}(u(x)),\sem{J}{B}(x))\in \sem{\mathtt{ScottL}}{\underline{\pi}}$, then for $y\in I$, $(\{\sem{J}{\bigotimes \Gamma}(u(x))\mid x\in\inv{u}(y)\},\{\sem{J}{B}(x)\mid x\in\inv{u}(y)\})\in !\sem{\mathtt{ScottL}}{\underline{\pi}}$ we conclude since $\{\sem{J}{\bigotimes \Gamma}(u(x))\mid x\in\inv{u}(y)\}=\{\sem{J}{\bigotimes \Gamma}(y)\}$ and $\{\sem{J}{B}(x)\mid x\in\inv{u}(y)\}= \sem{I}{\wnm{u}{B}}(y)$.
    \end{itemize}
  \end{proof}
}
 
\section{Related work and Perspectives}\label{sec:conc}
\subsection{Bucciarelli and Ehrhard's $IndLL$}
We can recover the original system of~\cite{IndLL2} by a slight modification of ours. Formulae are the same and they are still defined over sets. However, in the definition of propositional variables $ f(X) ,$ the function $f $ must be injective. The base change operation has to be restricted to injective functions. The subtyping is the same, except that, in the exponential case, the function $g$ must be bijective; which collapses the subtyping relation into an equivalence. Multiplicative and additive rules are the same, only exponential are to be modified:
  \begin{center}
    \AxiomC{$i\Bot j$}
    \AxiomC{$\Gamma\indvdash{I}^{\!\!o} \wnm{w\circ i}\indfun{i}{A},\wnm{w\circ j}\indfun{j}{A}$}
\BinaryInfC{$\Gamma\indvdash{I}^{\!\!o} \wnm{w}A$}
    \DisplayProof\hskip \skiplength
    \AxiomC{$\Gamma\indvdash{I}^{\!\!o} \indfun{\inv{u}}{B}$}
    \AxiomC{\optional{$u$ bij}}
    \BinaryInfC{$\Gamma\indvdash{I}^{\!\!o} \wnm{u}B$}
    \DisplayProof\vspace{.7em}\\
    \AxiomC{$\Gamma\indvdash{I}^{\!\!o} $}
    \UnaryInfC{$\Gamma\indvdash{I}^{\!\!o} \wnm{\init}B$}
    \DisplayProof\hskip \skiplength
    \AxiomC{$\bangm{u_1}A_1,\dots,\bangm{u_n}A_n\indvdash{J}^{\!\!o} B$}
\UnaryInfC{$\bangm{v\circ u_1}A_1,\dots,\bangm{v\circ u_n}A_n\indvdash{I}^{\!\!o} \bangm vB$}
    \DisplayProof
  \end{center} 

  \longVOnly{
    It is known that non-idempotent intersection types behaves better with respect to type fixedpoint; this result is know as the sensibility of all reflexive relational models~\ref{Giulio?}. In our case, such property is expected to become:
    \begin{conjecture}[Characterization of productivity]
      Let $\pi$ a proof of $\mu^\infty \OIndLL$ with no empty locus appearing in any sequent of $\pi$, then the cut-elimination of $\pi$ is productive.

      Conversely, if a proof $t$ of LL has a productive cut-elimination, then there is a proof $\pi$ of $\mu^\infty  \OIndLL$ with no empty locus appearing in any sequent such that $\underline \pi=t$.
    \end{conjecture}
  }

\subsection{Generalising \texttt{QuasiInj} to (semi-)extensive categories with pullbacks}

Besides the theorem linking $\IndLL$ to intersection types, there is no result in this paper that fundamentally relies on the quasi-injective functions restriction. We could use finite sets and all functions. If one accepts to forget about the Seely isomorphisms, one can even go with injective functions only (which is not so interesting as it means obtaining sub-linear intersection types...). In fact, even more exotic categories such that of topological spaces or graphs could be used provided that they have all coproducts, all pullbacks and that those interact well together (the category has to be extensive).

\begin{definition}[Semi-extensive category~\cite{CarbLackWalt93}]\label{def:extcat}
A semi-cocartesian category  (i.e., a monoidal category with unit as initial object) $( \mathit{C},+,\0)$ is semi-extensive if the bifunctor
$$((f,g)\mapsto f{+}g):\mathit{C}/X\times \mathit{C} /Y \rightarrow \mathit{C}/(X{+}Y)$$
is an equivalence of category.

  More concretely, it is semi-extensive if the initial object $\0$ is strict (i.e., the morphisms targeting it are isomorphisms), and if pullbacks of injections $(i,j)$ along arbitrary morphisms $g$ exist and are orthogonal and jointly surjective. It is extensive if it is semi-extensive and is cocartesian (i.e., it has codiagonals).
\end{definition}
\begin{theorem}
  For any semi-extensive category~$\: \mathit{C}$ with all pullback, we can generalised $\IndLL$ to have objects of~$\: \mathit{C}$ as loci, and morphisms of~$\:\mathit{C}$ as indexes and base change. The resulting logic has a normalising cut-elimination procedure.

  If~$\:\mathit{C}$ is extensive, then we can also prove the Seely isomorphism.
\end{theorem}
\longVOnly{
  The prof is immediate, the only difficulty is with the initial object (the one that play the role of $\emptyset$) which can have isomorphic copies, which require generalising the $\oplus$-introduction to have $\init$ replaced by any morphism from one of those copies.
  }

\subsection{Beyond : Double categories of loci}

The former generalisation do not directly encompass $\OIndLL $. One can then aim at an even more general framework. We expect that such framework should go way further than intersection types and formalise systems such as graded linear logic~\cite{BrunelGMZ14}, Girard's bounded linear logics and its extensions~\cite{BLL,FujiiKM16}, Light linear logics~\cite{LLL} or even some fragments of MLSub~\cite{MLSub}. We envision a bigger project aiming to achieve this generalisation to study syntactic properties of type system refinements at the level of loci.

For this, we need to decorate base changes and exponential indexes as in the original $\IndLL$. Categorically we need the notion of \emph{double category}. We get two different kind of morphisms over the same objects, forming a double category, which vertical category is that of base changes and which horizontal morphisms are indexes of the exponentials.

In such a double category, there is an additional structure, the morphisms of the arrow category, which are squares\shortVOnly{\vspace{-.4em}}\linebreak
\raisebox{-1em}{
\begin{tikzpicture}
  \node (X) at (0,0) {};
  \node (Y) at (0,-.65) {};
  \node (Z) at (1,-.65) {};
  \node (W) at (1,0) {};
  \draw[->] (X.south) to node[left] {{\scriptsize $g$}} (Y.south);
  \draw[->] (X.south) to node[sloped]{$\shortmid$} node[below] {{\scriptsize $v$}} (W.south);
  \draw[->] (W.south) to node[auto] {{\scriptsize $f$}} (Z.south);
  \draw[->] (Y.south) to node[sloped]{$\shortmid$}  node[auto] {{\scriptsize $u$}} (Z.south);
\end{tikzpicture}
}
representing the subtyping
\raisebox{.35em}{
\AxiomC{$\indfun{g}{A}\subtype A'$}
\UnaryInfC{$\bangm{u}A \subtype \indfun{f}{\bangm{v}A'}$}
\DisplayProof
}.

Additional conditions must be required. Among them you can find (i) semi-extensivity of the vertical category,\footnote{Biproducts should also be sufficient.} in order to define the additives, or (ii) mixed pullback diagrams in the double categories in order to define $\indfun{f}{\bangm{u}A}$.

Our $\IndLL$ defined on a semi-extensive category $\mathit{C}$ with pullbacks  corresponds to the double category with $\mathit{C}$ as both vertical and horizontal category and commuting squares as arrow morphisms. In the above square, the $f$ can always be chosen an identity (due to the mixed pullback) and the $v$ becomes $g\circ u$. For  $\OIndLL$, we are using the framed bi-category $\Rel$ in which we inverse horizontal arrows, and restrict vertical arrows to monos.

\shortVOnly{
  \subsection{Adding type-fixpoints or other operators}
  Another, orthogonal, extension consists of adding operators to the logic. One can add first order, as in~\cite{IndLL2ndO} for example. In the long version, we consider logical fixpoints~\cite{phdDoumane} $\mum fX.A$ and $\num fX.A$, with their three kinds of rules (Park, circular, infinitary).

  We shall see that logical fixpoints of our idempotent $\IndLL$ do not compute the expected intersection types of untyped-lambda terms and do not characterise any kind of termination. This is due to those fixpoints characterising the largest intersection type system, i.e., the largest reflexive filter-model, which is not sensible. We will investigate if in the setting of $\OIndLL$ we can somehow still characterise proof productivity.
  }

\subsection{Further refinements}
Even though base changes and negation, in $\IndLL$, are syntactically similar (as transformation rather than operators) they are not interacting at all. One could consider a version inspired from tensorial logic~\cite{Mellies17} merging the two concepts, which would also merge the (non-)involutivity of negation and the (non-)functoriality of of base.

One can also consider other variants of intersection types, such as intersection type distributors~\cite{ol:intdist} where the subtyping also has a computational content and mix the notion to really use the proof relevance of our subtyping.

\printbibliography

\end{document}